\documentclass{article}[11pt]
\usepackage[top=1in, bottom=1in, left=1in, right=1in]{geometry}
\usepackage{graphicx} % Required for inserting images
\usepackage{amsmath}
\usepackage{amsfonts}
\usepackage{comment}
\usepackage{color}
\usepackage{amsthm}
\usepackage{hyperref}

\usepackage{subcaption}

\newtheorem{theorem}{Theorem}

\newtheorem{lemma}{Lemma}

\newtheorem{proposition}{Proposition}
\newtheorem{assumption}{Assumption}

\usepackage{booktabs}

\newenvironment{hproof}{%
  \proof}{\endproof}

\usepackage{natbib}
 \bibpunct[, ]{(}{)}{,}{a}{}{,}%

\usepackage{tikz}
\usepackage{pgfplots}
\pgfplotsset{compat=newest}

\title{Simple vs. Optimal Congestion Pricing}
\author{Devansh Jalota\\ Georgia Tech\\{\tt devansh.jalota@isye.gatech.edu}
        \and 
        Xuan Di\\ Columbia\\ {\tt sharon.di@columbia.edu}
        \and 
        Adam N. Elmachtoub\\ Columbia\\ {\tt adam@ieor.columbia.edu}
}
\begin{document}

\maketitle

\begin{abstract}
Congestion pricing has emerged as an effective tool for mitigating traffic congestion, yet implementing welfare or revenue-optimal dynamic tolls is often impractical. Most real-world congestion pricing deployments, including New York City’s recent program, rely on significantly simpler, often static, tolls. This discrepancy motivates the question of how much revenue and welfare loss there is when real-world traffic systems use static rather than optimal dynamic pricing. We address this question by analyzing the performance gap between static (simple) and dynamic (optimal) congestion pricing schemes in two canonical frameworks: Vickrey’s bottleneck model with a public transit outside option and its city-scale extension based on the Macroscopic Fundamental Diagram. In both models, we first analyze the revenue-optimal static and dynamic tolls, which have received limited attention in prior work. Across many practically relevant regimes, we show that revenue-optimal static tolls achieve at least two-thirds of the maximum achievable revenue and incur at most twice the minimum achievable system cost, with stronger and more general guarantees in the bottleneck model. %In the bottleneck model, we further show that revenue-optimal static tolls achieve at least half of the dynamic optimal revenue across all parameter regimes.
%In the worst-case, revenue-optimal static tolls achieve at least half the dynamic optimal revenue and at most twice the minimum achievable system cost across practically relevant parameter regimes, with stronger and more general guarantees in the bottleneck model. 
We corroborate our theory with numerical results demonstrating that static tolls achieve roughly 80-90\% of the dynamic optimal revenue while incurring at most an 8-25\% higher total system cost than the minimum achievable.
\end{abstract}

\section{Introduction}
\vspace{-3pt}
%Traffic congestion is a massive problem in large urban metropolises and congestion pricing is one of the main mechanisms developed to alleviate some of its issues - maybe mention here NYC's CP plan and its massive impact it is likely to have and that it has already had solid positive impacts to begin with in terms of congestion reductions and revenues raised

Traffic congestion has surged in major cities worldwide, straining infrastructure, degrading air quality, and imposing substantial economic costs~\citep{usCongestioncosts,usCongestioncosts2}. Against this backdrop, congestion pricing has emerged as an effective tool for mitigating traffic congestion's inefficiencies and empirical evidence from recent large-scale deployments underscores its promise. For instance, New York City’s congestion pricing program implemented in January 2025 has increased traffic speeds while generating substantial revenue earmarked for transit investments~\citep{cook2025short}. Similar successes have been documented in Stockholm and Singapore~\citep{stockholm-cp}. 

While many large-scale congestion pricing deployments, including those in London (2003), Stockholm (2007), and New York City (2025), have emerged in the past two decades, their intellectual roots date back to~\citet{pigou}. Under Pigou’s framework, the socially optimal toll charges each traveler the marginal cost they impose on others in the system, which may vary over time with congestion. However, implementing such dynamic congestion pricing schemes in practice is challenging and often impractical, as they require time-varying tolls that often respond to evolving traffic conditions, potentially at the granularity of individual network links. Such fine-grained pricing demands sophisticated sensing, communication, and computation infrastructure while placing substantial informational and cognitive burdens on travelers. Consequently, despite its theoretical appeal, fully dynamic tolling remains largely infeasible in real-world congestion pricing deployments.

Instead, most operational congestion pricing systems rely on much simpler, often static, tolling structures. For instance, fixed (static) charges are levied on specific bottlenecks, bridges, or crossings, such as the San Francisco (SF)–Oakland Bay Bridge in the SF Bay Area~\citep{GONZALES2015267}. % or the Bandra-Worli Sea Link in Mumbai, India~\cite{shah-koner-2022}. 
Likewise, large urban cordon or area-based systems, such as those in London and New York City, typically impose a fixed (static) entry or usage fee for designated congestion zones during peak periods. Such static tolling schemes are appealing as they are transparent, easy to communicate to users, and require less real-time information and computational infrastructure than theoretically optimal dynamic pricing. 

Yet, the simplicity of static tolling means that it may deviate substantially from idealized marginal-cost pricing, raising a central question: \emph{How much performance is sacrificed when real-world transportation systems use simple static pricing instead of optimal dynamic congestion pricing?} Understanding this gap is crucial for policymakers as cities seek congestion pricing schemes that are operationally feasible and deliver substantial welfare and revenue gains. This question is also increasingly relevant for transportation authorities exploring more dynamic approaches, such as managed express lanes that adjust tolls based on real-time conditions~\citep{PANDEY2018304,JANG201469}, or Singapore’s ERP 2.0 system, which aims to enable real-time, link-level pricing across its road network~\citep{erp-singapore-2.0}.

This work addresses the above question by analyzing the performance gap in terms of revenue and welfare (or system cost) between simple static and optimal dynamic congestion pricing schemes in two canonical modeling frameworks that capture the dominant real-world use cases of congestion pricing. In particular, we study (i) Vickrey’s bottleneck model~\citep{vickrey1969congestion}, a foundational and analytically tractable representation of peak-period congestion that captures queue formation and dissipation at a bottleneck, and (ii) its city-scale urban system extension based on the Macroscopic Fundamental Diagram (MFD)~\citep{DAGANZO200749}, which captures aggregate system-wide congestion through a state-dependent capacity, rather than a fixed bottleneck capacity. We focus on these models as they provide clean, tractable frameworks in which to isolate the core mechanisms governing static versus dynamic tolling performance, while also reflecting the structure of contemporary deployments, from bottleneck-priced corridors such as the SF–Oakland Bay Bridge to city-scale congestion pricing systems in New York City and London.

In studying these models, we focus on revenue-maximizing tolls, which have received limited attention in the literature that has largely emphasized welfare-optimal or system-cost-optimal tolls. Beyond filling this research gap, our focus on revenue maximization is motivated by several practical considerations. First, revenue generation is an explicit policy objective of many congestion pricing programs (e.g., New York City), where toll revenues are earmarked for transit investments~\citep{mta-cp-revenue}. Next, revenue objectives are central to public–private toll road partnerships, including tolled road networks operated by firms such as Transurban~\citep{https://doi.org/10.1111/1745-5871.12528}. Finally, our numerical experiments in Section~\ref{sec:experiments} demonstrate that static revenue-optimal tolls often coincide, or perform nearly identically to, system-cost-optimal static tolls across empirically relevant parameter regimes, reinforcing our focus on revenue-optimal tolling. Overall, while a key focus of this work is to study revenue-optimal tolling policies, we also evaluate their welfare (system cost) implications, providing a comprehensive assessment of performance across the metrics most relevant to policymakers.

\vspace{-3pt}
\subsection{Our Contributions} 
\vspace{-3pt}

This work analyzes the performance of static versus dynamic congestion pricing in two canonical settings: Vickrey’s bottleneck model with a public transit outside option (e.g., a subway) and its MFD-based urban system extension. In doing so, a central contribution is the analysis of revenue-optimal static and dynamic tolls in both models. For the bottleneck model with an outside option, in Section~\ref{sec:revenue-optimal-tolling-classical-bottleneck}, we leverage the equilibrium properties of Vickrey's framework to show that computing revenue-optimal static and dynamic tolls reduces to solving \emph{single-variable quadratic programs} that admit closed-form solutions. For urban systems, such closed-form results are typically intractable for arbitrary MFD relations; thus, in Section~\ref{sec:mfd}, we focus on the widely studied and empirically supported triangular fundamental diagram. In this setting, we characterize the revenue-optimal dynamic toll, showing that it maintains the system at the throughput-maximizing capacity, and derive a closed-form expression for the revenue generated by any static toll.

Leveraging these characterizations, we compare static revenue-optimal tolling with the corresponding optimal dynamic benchmarks on both revenue and total system cost metrics (see Section~\ref{subsec:perf-metrics} for details on these metrics), which are of direct relevance to policymakers. Our results establish lower bounds on the revenue and upper bounds on the total system costs achievable under static-revenue optimal tolls relative to the dynamic benchmarks. In Section~\ref{sec:revenue-optimal-tolling-classical-bottleneck}, in Vickrey’s bottleneck model with an outside option, we show that static revenue-optimal tolls obtain \emph{at least half} of the revenue achieved by dynamic revenue-optimal tolls across all parameter regimes and \emph{at least two-thirds} in many practically relevant regimes. Static revenue-optimal tolls do not always guarantee constant-factor approximations for system cost. However, when the public transit outside option is not substantially less attractive than driving a car, a condition that holds in many real-world settings, static revenue-optimal tolls incur \emph{at most twice} the optimal system cost. These guarantees also extend to a stochastic version of the bottleneck model with uncertain demand and capacity, where a single static revenue-optimal toll chosen ex ante retains the constant-factor revenue guarantee, as well as the system cost guarantee when the above outside-option condition holds across realizations, relative to adaptive dynamic benchmarks that can condition on the realized demand and capacity and vary tolls over time. In the MFD framework studied in Section~\ref{sec:mfd}, the guarantees we obtain are weaker (as the capacity is state-dependent) but still robust: across a range of practically relevant parameter regimes, static revenue-optimal tolls still achieve \emph{at least two-thirds} of the dynamic optimal revenue and incur \emph{at most twice} the minimal system cost.

Finally, in Section~\ref{sec:experiments}, we complement our theory with numerical experiments based on two real-world case studies: the SF–Oakland Bay Bridge for the bottleneck model and New York City’s congestion relief zone for the MFD framework. Our results show that static revenue-optimal tolls achieve roughly 80-90\% of the dynamic optimal revenue while incurring at most an 8-25\% higher total system cost than the minimum achievable for practically relevant parameter regimes, though this gap can widen when the public transit option is significantly less attractive (i.e., more costly), consistent with theory. Moreover, when evaluated through the lens of our model, we find that the static tolls currently implemented in practice in both case studies correspond to parameter regimes in which static revenue-optimal tolls capture nearly all the benefits of dynamic tolling, achieving roughly 97\% of the dynamic optimal revenue and increasing total system cost relative to the minimum achievable by less than 1\%. 
Beyond the attractiveness of public transit relative to driving, another parameter shaping the performance of static revenue-optimal tolling is the ratio of the maximum system throughput to the desired user arrival rate. When this ratio is close to one, as in the Bay Bridge study, static revenue-optimal tolls capture a larger fraction of the dynamic optimal revenue but incur a higher multiplicative gap relative to the minimum achievable system cost. In contrast, this ratio is substantially lower in the New York City study, aligning with its high transit mode share. Thus, in addition to highlighting the efficacy of static tolls, our numerical results illustrate a fundamental trade-off between revenue and system cost objectives under static tolling. 

%capture nearly all the benefits of dynamic tolling, 

%Our results show that for practically relevant parameter regimes, static revenue-optimal tolls incur only modest performance losses, with a revenue loss of at most 10–20\% relative to that of the dynamic revenue-optimal policies, though this gap can widen when the public transit option is significantly less attractive (i.e., more costly), consistent with theory. Similar patterns hold for total system cost relative to the dynamic system-cost-optimal benchmark. 

Overall, our results demonstrate that simple static tolls can deliver robust performance on both revenue and system cost metrics, underscoring their practical efficacy and shedding light on the strong empirical performance of the many static congestion pricing systems already in operation. Moreover, our comparisons abstract from the substantial informational and computational burdens dynamic tolling places on transportation planners and users, and accounting for such implementation frictions may further strengthen the case for static tolling.

%Moreover, since our comparisons abstract from the substantial informational and computational burdens that dynamic tolling places on both transportation planners and users, our theoretical bounds should be interpreted as worst-case guarantees. Incorporating realistic behavioral or operational constraints would only further strengthen the case for static tolling in real-world deployment.

%\textbf{Organization:} The remainder of this paper is organized as follows. Section~\ref{sec:related-literature} reviews related literature. Section~\ref{sec:model} introduces our model and reviews equilibrium outcomes in Vickrey's bottleneck model with and without an outside option, which form the foundation of our analysis. Then, Section~\ref{sec:revenue-optimal-tolling-classical-bottleneck} derives the revenue-optimal static and dynamic tolling policies in Vickrey's bottleneck model with an outside option and compares the performance of revenue-optimal static tolls to its dynamic benchmarks. Section~\ref{sec:mfd} extends these results to urban systems based on the MFD. Section~\ref{sec:experiments} presents numerical experiments. Section~\ref{sec:conclusion} concludes and provides directions for future work.

In the appendix, we provide further intuition and explanations of our results (Appendices~\ref{apdx:arrival-departure-overall} and~\ref{apdx:conceptual-bottleneck-mfd}), %formal mathematical definitions of the revenue and total system cost metrics (Appendix~\ref{apdx:perf-metrics}), 
omitted proofs (Appendix~\ref{apdx:proofs}), additional details on our numerical experiments (Appendix~\ref{apdx:additional-numerical-details}), and additional theoretical results (Appendix~\ref{apdx:sco-static-toll-overall}).

\vspace{-3pt}
\section{Related Literature} \label{sec:related-literature}
\vspace{-3pt}
The study of peak-period congestion originates from Vickrey’s bottleneck model~\citep{vickrey1969congestion}, which characterizes the equilibrium departure time choices of users traversing a fixed-capacity bottleneck. This framework has served as the foundation for an extensive literature, with extensions including models of parallel-route choice~\citep{ARNOTT1990209}, general networks~\citep{schwarz-2025}, modal split~\citep{TABUCHI1993414}, heterogeneous preferences~\citep{Newell198721,lindsey-2004-38}, carpooling~\citep{XIAO2016383,ostrovsky2025congestion}, stochastic demand~\citep{arnott1999information,doi:10.1287/trsc.17.4.430}, and bounded rationality~\citep{bounded-rationality-mahmassani}, among others. While the bottleneck model and its extensions capture queuing delays as a point queue, they do not account for hypercongestion, wherein system throughput declines once road density exceeds a critical threshold. This limitation has motivated city-scale extensions, including bathtub models of downtown traffic~\citep{ARNOTT2013110} and queuing formulations~\citep{FOSGERAU2013122} in which the bottleneck capacity degrades under heavy congestion. Relatedly, the MFD~\citep{DAGANZO200749,GEROLIMINIS2008759} provides a parsimonious representation of the relation between system-wide vehicle density and throughput. For a survey of the bottleneck model and its extensions, see~\citet{LI2020311}.

Building on these foundations, a substantial literature has examined congestion pricing as a mechanism for mitigating traffic, with a focus on system-cost-optimal tolls that internalize congestion externalities~\citep{pigou}. In the bottleneck model, the system-cost-optimal toll is time-varying, requiring continuously adjustable charges~\citep{vickrey1969congestion,ARNOTT1990111,Newell198721}. While such dynamic pricing is theoretically efficient, it is operationally complex, motivating the design of more implementable schemes, including uniform, stepwise, and coarse time-of-day tolls~\citep{CHU1999697,arnott1993structural,BRAID1989320}. 

Prior work has therefore studied the performance loss in terms of system cost (or welfare) from restricting continuously time-varying tolls to simpler schemes~\citep{CHU1999697,LAIH1994197,ROBINLINDSEY201246}. For example, \citet{CHU1999697} numerically compares uniform and one-step tolls with a smoothly varying optimal toll in a simulation model with heterogeneous commuters and mode choice, finding that flat tolls can perform favorably on revenue but less favorably on efficiency. Analytical guarantees in this literature, however, are derived primarily under a bottleneck benchmark with no outside option and a common preferred arrival time across users. Under this benchmark, dynamic system-cost-optimal tolls eliminate queuing and halve the total system cost relative to the setting without tolls. Moreover, \citet{ROBINLINDSEY201246} show that the gains from step tolling can depend on travelers' behavior around toll discontinuities, while \citet{LAIH1994197} shows that an optimal $m$-step toll eliminates an $m/(m+1)$ fraction of the queuing delay that can be removed by dynamic tolling; for a single-step toll, this implies that total cost is at most $3/2$ times the fully dynamic system-cost optimum. Unlike these works, our setting considers users with uniformly distributed heterogeneous desired arrival times as in~\citet{vickrey1969congestion}. In this setting, the ratio between total system costs of the untolled and dynamically optimal policies is no longer fixed at two and can be arbitrarily large. Accordingly, the classical step-toll guarantees do not carry over directly to this setting. Instead, we show that, even relative to a dynamic benchmark whose improvement over no tolling can be arbitrarily large, the system cost under a static revenue-optimal toll remains within a factor of two of the dynamic system-cost optimum across a range of practically relevant parameters (see Theorem~\ref{thm:sc-comp-static-v-dynamic}). Moreover, unlike much of the prior work comparing simpler and optimal tolling policies, we incorporate endogenous mode choice through a public transit outside option and focus on revenue-optimal rather than system-cost-minimizing tolls.

In this regard, our work connects to studies of revenue or profit maximization in bottleneck settings~\citep{BernsteinMuller1993,de2000private,DEPALMA2002217,SILVA201413,FU2018430}, and to extensions of Vickrey's framework with elastic demand or alternative travel modes~\citep{arnott1993structural,GONZALES20121519,BASSO201947}. On the revenue-maximization front, \citet{BernsteinMuller1993} characterize revenue-optimal static and step tolls under a common desired arrival time, whereas we consider uniformly distributed desired arrival times as in \citet{vickrey1969congestion} and also characterize the dynamic revenue-optimal policy. Subsequent work largely studies competition among privately operated facilities~\citep{de2000private,DEPALMA2002217,FU2018430}. Relatedly, \citet{SILVA201413} study a monopoly airline that controls its flight schedule and internalizes the congestion it creates, leading profit-maximizing and efficient scheduling to coincide. In our setting, by contrast, the road pricing authority controls only tolls and does not bear travelers' congestion or schedule delay costs, so revenue maximization and system cost minimization generally remain distinct. Complementary work incorporating an outside option in Vickrey's framework has largely focused on system-cost-optimal tolling and transit operations~\citep{arnott1993structural,GONZALES20121519}. Among these, \citet{BASSO201947} provide a particularly rich treatment of the transit alternative, jointly optimizing road pricing with bus frequency and transit fares while explicitly modeling passenger queuing. We instead focus on the road pricing problem when transit serves as an exogenous outside option, and the authority controls only the toll, allowing us to characterize revenue-maximizing static and dynamic tolls and derive explicit guarantees comparing static tolling with both revenue and system-cost-optimal dynamic benchmarks. That said, we also consider a setting that allows transit costs to depend on ridership in Appendix~\ref{apdx:general-transit-cost}.

Beyond the tolling literature rooted in Vickrey’s bottleneck framework, our work also connects to the broader work on congestion pricing in settings where first-best pricing is infeasible due to policy, behavioral, or infrastructural constraints. This includes research on second-best congestion pricing, which examines toll design when only a subset of network links can be priced~\citep{VERHOEF2002281,bilevel-labbe,Larsson1998,bilevel-patricksson,DI201674}, pricing in the presence of untolled alternatives~\citep{repec:eee:juecon:v:40:y:1996:i:3:p:279-302}, and multi-modal congestion management policies that jointly consider congestion pricing with transit fares, subsidies, frequency, or dedicated bus lanes~\citep{BASSO2012890,10.1257/pol.6.4.1}. Related work on optimal cordon pricing uses zone-based tolls to manage congestion in urban systems~\citep{MUN200321,ZHANG2004517}. Like these works, we study simple, implementable tolling policies; however, whereas much of this literature focuses on designing or numerically evaluating second-best policies in network or multi-modal settings, we provide a theoretical analysis of static and dynamic tolling policies in the bottleneck and MFD frameworks and, crucially, establish explicit guarantees on the performance of static tolling relative to fully dynamic benchmarks.

%however, unlike much of this literature that characterizes or numerically evaluates second-best policies in network or multi-modal settings, we focus on analytically characterizing optimal static and dynamic tolls in the bottleneck and MFD frameworks and deriving explicit guarantees on the performance of simple tolling relative to fully dynamic benchmarks.

%however, unlike these works that focus on characterizing or numerically evaluating second-best policies in network or multimodal settings, we derive closed-form optimal static and dynamic tolls in the bottleneck and MFD models, yielding transparent structural insights and explicit performance guarantees.

%unlike the large-scale bi-level or mixed-integer optimization models or simulation based approaches used in these works, we derive closed-form optimal static and dynamic tolls in both the bottleneck and MFD models, yielding transparent structural insights and theoretical guarantees.

%and optimal cordon pricing, which uses zone-based tolls to manage congestion in urban systems~\citep{MUN200321,ZHANG2004517}. Like these works, we study simple, implementable tolling policies; however, unlike the large-scale bi-level or mixed-integer optimization models or simulation based approaches used in these works, we derive closed-form optimal static and dynamic tolls in both the bottleneck and MFD models, yielding transparent structural insights and theoretical guarantees.

Our work also connects to the literature on simplicity versus optimality in algorithm and mechanism design~\citep{hartline2009simple,HART2017313}, including studies in revenue management~\citep{10.1145/3736252.3742552,doi:10.1287/opre.2020.2054} and online platforms~\citep{10.1145/2764468.2764527}, which show that simple, time-invariant prices can achieve strong performance relative to their optimal dynamic counterparts. In this spirit, we also show that simple static tolls can achieve performance guarantees relative to optimal dynamic tolling.

Finally, since we compare tolling policies in terms of revenue and welfare (or system cost), our work relates to the broader economics literature on the tension between revenue and welfare-maximizing pricing rules. This trade-off is well known in auction theory, most notably in the contrast between welfare-maximizing Vickrey auctions~\citep{a2fd2512-18e4-3e86-a4f2-2f4443994eb2} and Myerson’s revenue-maximizing mechanisms~\citep{doi:10.1287/moor.6.1.58}, and has also been studied in settings including security games~\citep{jalota2024simplenearoptimalsecuritygames}, congestion games~\citep{zhang2024designinghighoccupancytolllanes}, and revenue management~\citep{doi:10.1287/mnsc.2017.2943,doi:10.1049/icp.2025.2957}. %Contributing to this literature, we analyze revenue-welfare tradeoffs in a congestion pricing context and quantify how much revenue can be retained and how much welfare is lost when moving from optimal dynamic pricing to simple static tolls.

\vspace{-2pt}

\section{Model and Background: Bottleneck with Outside Option} \label{sec:model}

\vspace{-2pt}

This section presents Vickrey's bottleneck model with an outside option (Section~\ref{subsec:preliminaries-model}), reviews its untolled equilibrium outcome established in~\citet{GONZALES20121519} (Section~\ref{subsec:background-equilibria}), and defines the revenue and system cost metrics to evaluate the tolling policies we study (Section \ref{subsec:perf-metrics}).

\vspace{-2pt}

\subsection{Setup} \label{subsec:preliminaries-model}

\vspace{-2pt}

We study Vickrey's bottleneck model~\citep{vickrey1969congestion}, a canonical representation of peak-period congestion in capacity-constrained facilities such as bridges and tunnels, augmented with a public transit outside option. In this framework, a mass of $\Lambda$ users makes morning commute trips using either (i) a car, which requires traversing a bottleneck of capacity (service rate) $\mu$, or (ii) a public transit alternative (e.g., a subway). As in~\citet{vickrey1969congestion}, users' desired bottleneck crossing times are uniformly distributed over the time interval $[t_1, t_2]$ (e.g., morning rush). In the absence of the bottleneck, users would arrive at their destination at their desired times, with a corresponding desired arrival rate $\lambda = \frac{\Lambda}{t_2-t_1}$. Unlike stochastic queuing models, e.g., based on Poisson processes, we clarify that arrivals and service are deterministic; stochastic variation in arrivals and bottleneck capacity is considered in Section~\ref{sec:stochastic-extension}. In addition to bottleneck congestion that may induce deviations from users' desired bottleneck crossing times and influence mode choices, users' travel decisions can also be influenced via a (possibly time-varying) toll $\Tilde{\tau}(t)$ on cars crossing the bottleneck at time $t$. % in Vickrey's bottleneck model.

%the administrator may additionally impose a (possibly time-varying) toll $\tau(t)$ on vehicles crossing the bottleneck at time $t$.

%To influence both departure-time and mode choices, the administrator may impose a (possibly time-varying) toll $\tau(t)$ on vehicles crossing the bottleneck at time $t$.

The distinction between car and public transit is how they influence users' travel costs. For car users, the travel cost depends on their desired bottleneck crossing time $t^*$, the prevailing level of congestion, and tolls. Normalizing the queuing delay at the bottleneck to zero under free-flow conditions (i.e., in the absence of congestion), the cost of a car trip comprises: (i) a congestion independent free-flow generalized cost $\Tilde{z}_C$, (ii) queuing (or waiting) delay $w(t)$ to traverse the bottleneck due to congestion, (iii) schedule delay $|t^* - t|$, denoting the deviation between a user's desired and actual bottleneck crossing time, and (iv) the toll $\Tilde{\tau}(t)$. Let $c_W$ denote the penalty for incurring a unit of waiting time delay, and $c_e$ and $c_L$ be the schedule delay penalties for arriving early or late, respectively, which are assumed to be identical across users, with $0< c_e < c_W$ and $c_L > 0$, as is standard in the literature~\citep{vickrey1969congestion}. Then, the total travel cost for a car user with a desired bottleneck crossing time $t^*$ who crosses the bottleneck at time $t$ is: 
{\setlength{\abovedisplayskip}{1pt}
\setlength{\belowdisplayskip}{1pt}
\begin{align} \label{eq:userCost}
    \Tilde{c}(t, t^*) = \Tilde{z}_C + c_W w(t) + c_e (t^* - t)_+ + c_L (t - t^*)_+ + \Tilde{\tau}(t).
\end{align}}
%As is standard, we assume piecewise constant schedule delay costs, given by $c_S = c_e$ for early arrivals ($t < t^*$) and $c_S = c_L$ for late arrivals ($t \geq t^*$), where $0< c_e < c_W$ and $c_L < 0$~\cite{vickrey1969congestion}. 

Unlike car users, we assume that transit users incur a fixed generalized cost $\Tilde{z}_T$, as is standard in the literature~\citep{GONZALES20121519}. This specification is reflective of settings in which the public transit system charges a fixed fare, operates at fixed headways, and is segregated from the road network, features typical of metro or subway systems (e.g., New York City's subway), and thus does not experience congestion delays. While richer models could also allow the transit cost to depend on ridership (e.g., to capture crowding or discomfort effects), a fixed transit cost enables analytical tractability, allowing us to isolate the core mechanisms through which static and dynamic pricing influence user behavior and system outcomes. That said, in Appendix~\ref{apdx:general-transit-cost}, we relax this assumption and extend our comparison of static revenue-optimal tolling and the corresponding dynamic optimal benchmarks to settings with ridership-dependent transit costs.

%Nevertheless, relaxing this assumption to incorporate crowding effects or other dependencies of the transit cost on ridership is a valuable direction for future work.

In the remainder of this work, for analytical clarity, we normalize user costs by the waiting time penalty $c_W$. Then, the normalized transit cost is $z_T = \frac{\Tilde{z}_T}{c_W}$, and letting $z_C = \frac{\Tilde{z}_C}{c_W}$, $\tau(t) = \frac{\Tilde{\tau}(t)}{c_W}$, $e = \frac{c_e}{c_W}$, and $L = \frac{c_L}{c_W}$, the normalized cost for a car user from Equation~\eqref{eq:userCost} is given by: %with desired crossing time $t^*$ who crosses the bottleneck at time $t$ is:
{\setlength{\abovedisplayskip}{1pt}
\setlength{\belowdisplayskip}{1pt}
\begin{align} \label{eq:userCostNormalized}
    c(t, t^*) = z_C + w(t) + e (t^* - t)_+ + L (t - t^*)_+ + \tau(t).
\end{align}}

\vspace{-3pt}
\subsection{Equilibria in the Untolled Bottleneck Model} \label{subsec:background-equilibria}
\vspace{-3pt}
%Given the setup and user costs defined in the previous section, we review here well-known equilibrium results for the classical bottleneck model, both without and with a public-transit outside option. 

Under the user costs defined above, a Wardrop \emph{equilibrium} arises when all users minimize travel costs by choosing a mode (car or transit) and, if traveling by car, a bottleneck crossing time, resulting in departure times, mode choices, and induced congestion such that no user has an incentive to deviate~\citep{vickrey1969congestion,hendrickson-1981-schedule,wardrop_theoretical_1952}. Given cost-minimizing user behavior, we now review well-established equilibrium characterization results of the bottleneck model in the untolled setting ($\tau(t)=0$ for all $t$), both without and with a public transit outside option, forming the foundation for our analysis of static and dynamic tolling in Sections~\ref{sec:revenue-optimal-tolling-classical-bottleneck} and~\ref{sec:mfd}.

\emph{Bottleneck without Outside Option:} To build intuition, we first consider the setting without an outside option, in which all users travel by car and choose only their bottleneck crossing time $t$. Specifically, a user with a desired crossing time $t^*$ selects a time $t$ to pass the bottleneck to minimize their travel cost $c(t, t^*)$ in Equation~\eqref{eq:userCostNormalized}. If the bottleneck service rate is at least the desired arrival rate, i.e., $\mu \geq \lambda$, no queue forms at equilibrium and all users cross the bottleneck at their desired times, incurring a cost $z_C$. When $\mu < \lambda$, as is typical during morning and evening rush periods, the bottleneck cannot serve all $\Lambda$ users during the interval $[t_1, t_2]$ and congestion delays arise. 

In this congested regime, users incur waiting and schedule delay costs, and the $\Lambda$ users are served at the bottleneck capacity $\mu$ over an interval $[t_A^{(No)}, t_D^{(No)}]$, where $t_D^{(No)} - t_A^{(No)} = \frac{\Lambda}{\mu} > t_2 - t_1$. Here, the superscript \emph{No} denotes the setting with \emph{no} outside option. Over this interval, the equilibrium waiting time function $w(t)$ is determined by the first-order optimality condition for cost-minimizing users, $c'(t, t^*) = 0$, under which no user can reduce their travel cost through a marginal adjustment in their bottleneck crossing time. This condition implies that the slope of the waiting time function satisfies $w'(t) = e$ for early arrivals and $w'(t) = -L$ for late arrivals, which is independent of users’ desired crossing times $t^*$, thus characterizing the equilibrium condition for all users.

To characterize the waiting time profile from this differential equation, note that the user with the maximum waiting time at equilibrium must cross the bottleneck at their desired time; otherwise, a marginal shift in their departure time would reduce schedule delay without increasing waiting time, thus lowering their travel cost. Let $t^* \! = \! \Tilde{t}$ denote this user's desired crossing time. Then, the waiting time profile is triangular with $w'\!(t) \! = \! e$ for $t \! < \! \Tilde{t}$ and $w'\!(t) \! = \! -L$ for $t \! > \! \Tilde{t}$, as shown in Figure~\ref{fig:equilibrium_both_models} (left). The critical user with a desired crossing time $\Tilde{t}$ incurs the maximum travel cost, and experiences the maximum waiting time $T_C \! = \! \frac{\Lambda e L}{\mu (e + L)} = w(\Tilde{t})$. %For more details on bottleneck equilibria and the resulting equilibrium bottleneck arrival and departure (crossing) time profiles of users, see~\citet{vickrey1969congestion,hendrickson-1981-schedule} and Figure~\ref{fig:arrival_profile_bottlenck_combined} in Appendix~\ref{apdx:arrival-departure-profiles}. 

\vspace{-8pt}

\begin{figure}[tbh!]
      \centering
      \includegraphics[width=130mm]{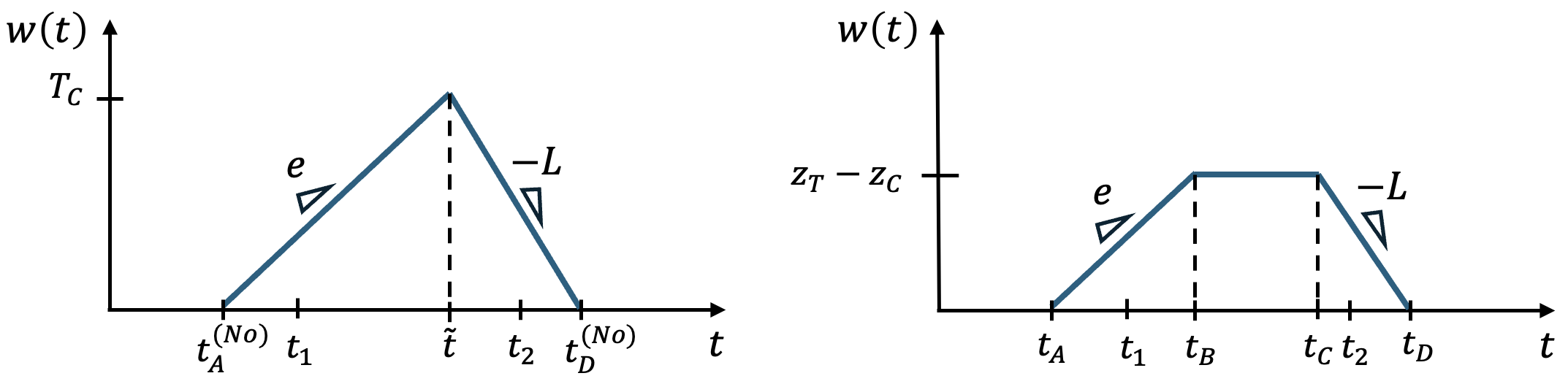}
      \vspace{-9pt}
      \caption{\small \sf Equilibrium waiting time profiles in the bottleneck model without an outside option (left) and with an outside option in a mixed-mode equilibrium where both car and transit are used (right) in the setting when $\mu < \lambda$. The slopes $e$ and $L$ denote the normalized schedule delay penalties for early and late arrivals, respectively, and $[t_1,t_2]$ represents users’ desired bottleneck crossing times. %All users passing the bottleneck at $t< \Tilde{t}$ pass the bottleneck earlier than their desired time while those passing the bottleneck at $t > \Tilde{t}$ arrive later than their desired depart
      }
      \label{fig:equilibrium_both_models} 
   \end{figure}  

\vspace{-5pt}

\emph{Bottleneck with Outside Option:} When a public transit outside option is available, users choose not only when to travel but also whether to travel by car or transit. While car users may incur congestion delays, transit users (e.g., subway riders) face a fixed cost $z_T$ and can arrive at their destination at their desired time; hence, they do not face a departure-time decision. Then, equilibrium mode choice and, for car users, departure-time decisions are determined by the relative magnitudes of the costs $z_T$, $z_C$, and the maximum waiting time $T_C$ in the car-only bottleneck. This yields three parameter regimes, corresponding to different equilibrium mode-use patterns, summarized below.

\emph{Case 1 ($z_T < z_C$):} All users take transit, as it is strictly preferred to a car trip at free-flow. % (i.e., without congestion delays).

%In this case, a transit trip is strictly preferred to a car trip at free-flow (i.e., without congestion delays); hence, all users take transit.

\emph{Case 2 ($z_T \geq z_C + T_C$):} Transit is never cost-effective relative to car travel, even for the user facing the highest travel cost with a wait time of $T_C$ in the car-only bottleneck. Thus, all users travel by car, and the equilibrium coincides with the bottleneck model without an outside option.

\emph{Case 3 ($z_T \in [z_C, z_C + T_C]$):} There exists a mixed-mode equilibrium in which both car and transit are used, as characterized by the following result from~\citet{GONZALES20121519}.

%Since we model the generalized cost of transit as a fixed value $z_T$, independent of ridership levels and road congestion, transit users can always choose to pass the bottleneck and arrive at their destination at their desired time. 

%Since we consider a fixed cost $z_T$ of using transit, i.e., the cost of transit is not influenced by congestion, all transit riders arrive at their destination at their wished time. At equilibrium, users choose their mode of travel (i.e., car or transit) and time of departing the bottleneck in the case of a car to minimize their cost. Then, the resulting equilibrium outcome will depend on the generalized costs $z_T$ and $z_C$, and the maximum travel cost $T_C$ as follows:

\vspace{-3pt}
\begin{proposition}[Two-Mode Equilibrium~\citep{GONZALES20121519}] \label{prop:user-eq-two-modes}
    Suppose users choose between traveling by car through a bottleneck with a free-flow cost $z_C$, and public transit with a fixed cost $z_T$, where $z_T \geq z_C$. If $z_T - z_C < T_C$ and users pass the bottleneck in order of their desired bottleneck departure times, there exists an equilibrium such that (see right of Figure~\ref{fig:equilibrium_both_models}): %\vspace{-2pt}
    \begin{enumerate}
        \item The number of early car users $N_e \! = \! \frac{\mu (z_T - z_C)}{e}$, who travel at the start of the rush during $[t_A, t_B]$. The number of late car users $N_L \! = \! \frac{\mu (z_T - z_C)}{L}$, who travel at the end of the rush during $[t_C, t_D]$.
        \item The number of on time car users and transit riders $N_o(\cdot)$, $N_T(\cdot)$ are strictly decreasing functions of the cost difference $z_T - z_C$ and they travel in the middle of the rush between $[t_B, t_C]$.
    \end{enumerate}
\end{proposition} 
\vspace{-3pt}
Proposition~\ref{prop:user-eq-two-modes} characterizes the equilibrium structure when $z_T \in [z_C, z_C + T_C]$, arguably the most practically relevant regime: while public transit is typically less attractive than driving under free-flow conditions, it is often not sufficiently worse for all users to travel by car, so that both modes are used at equilibrium. In this regime, Proposition~\ref{prop:user-eq-two-modes} shows that the waiting time profile for car users is trapezoidal (see right of Figure~\ref{fig:equilibrium_both_models}), which arises as the queuing delays at the car bottleneck cannot exceed $z_T - z_C$, since any such user would prefer to switch to transit. We also note that the interval $[t_A, t_D]$ in Proposition~\ref{prop:user-eq-two-modes} need not coincide with, and is generally a subset of, the corresponding interval $[t_A^{(\mathrm{No})}, t_D^{(\mathrm{No})}]$ in the car-only bottleneck. For more details on equilibria in the bottleneck model with or without an outside option and the resulting equilibrium bottleneck arrival and departure (crossing) time profiles, see~\citet{vickrey1969congestion,hendrickson-1981-schedule,GONZALES20121519} and Figure~\ref{fig:arrival_profile_bottlenck_combined} in Appendix~\ref{apdx:arrival-departure-profiles}.

%We also note that the interval $[t_A, t_D]$ in the statement of the proposition may not coincide with (and is a subset of) the corresponding interval $[t_A^{(No)}, t_D^{(No)}]$ in the car-only bottleneck.

\begin{comment}
\begin{figure}[tbh!]
      \centering
      \includegraphics[width=80mm]{Fig/equilibrium_with_outside_option.png}
      \vspace{-15pt}
      \caption{\small \sf Depiction of the equilibrium waiting times in the classical bottleneck model with an outside option in the setting when $z_C + T_C > z_T$. %All users passing the bottleneck at $t< \Tilde{t}$ pass the bottleneck earlier than their desired time while those passing the bottleneck at $t > \Tilde{t}$ arrive later than their desired depart
      }
      \label{fig:equilibrium-with-outside-option} 
   \end{figure}    
\end{comment}

%minimize their total travel cost
\vspace{-4pt}
\subsection{Performance Metrics to Evaluate Equilibria Induced by Tolling Policies} \label{subsec:perf-metrics}
\vspace{-3pt}
%Having characterized the equilibrium outcomes in the classical bottleneck model with and without an outside option, we now turn to the problem of 

%Having presented the untolled equilibria in the bottleneck model with and without an outside option, we present the metrics to assess the equilibria induced by a (possibly time-varying) tolling policy $\tau(\cdot)$. We focus on revenue and total system cost, both normalized by the waiting time penalty $c_W$. The revenue $R(\tau(\cdot))$ is the total toll payments collected from car users at the resulting equilibrium. The total system cost, $SC(\tau(\cdot))$, measures the aggregate travel burden borne by users and equals the sum of waiting time costs, schedule delay costs, and generalized costs (i.e., $z_T$ or $z_C$). Equivalently, it equals users' total travel costs at equilibrium \emph{minus} their toll payments. Consistent with standard practice in the congestion pricing literature~\citep{vickrey1969congestion,GONZALES20121519}, we include the costs $z_C$ and $z_T$ in the total system cost metric but exclude toll payments as transfers between users and the planner. For formal definitions of both metrics, see Appendix~\ref{apdx:perf-metrics}.

Having presented the untolled equilibria in the bottleneck model with and without an outside option, we present the revenue and total system cost metrics to assess the equilibria induced by a (possibly time-varying) tolling policy $\tau(\cdot)$. Both metrics are normalized by the waiting time penalty $c_W$.

%We evaluate tolling policies along two performance metrics: (i) cumulative toll revenue and (ii) total system cost.

\emph{Revenue:} The cumulative toll revenue $R(\tau(\cdot))$ is the total toll payments collected from car users at the equilibrium induced by a policy $\tau(\cdot)$. Under such a policy, let $t^{\tau(\cdot)}(t^*)$ denote the time at which a user with a desired bottleneck crossing time $t^*$ passes the bottleneck under $\tau(\cdot)$ at equilibrium, and let $x_C^{\tau(\cdot)}(t^*)\in[0,1]$ denote the equilibrium fraction of users with desired crossing time $t^*$ who travel by car (with the remaining fraction traveling by transit). Then, the cumulative toll revenue under $\tau(\cdot)$, normalized by $c_W$, is: $R(\tau(\cdot)) := \int_{t_1}^{t_2} \lambda \cdot x_C^{\tau(\cdot)}(t^*) \cdot \tau(t^{\tau(\cdot)}(t^*)) dt^*.$

For the tolling policies studied in this work, this revenue expression can be simplified as users cross the bottleneck over a connected interval at rate $\min \{ \mu, \lambda\}$. Accordingly, let $[t_A^{\tau(\cdot)}, t_D^{\tau(\cdot)}]$ denote the equilibrium interval over which users pass the bottleneck under the policy $\tau(\cdot)$. Then, in the case $z_T \ge z_C$, the revenue $R(\tau(\cdot)) = \int_{t_A^{\tau(\cdot)}}^{t_D^{\tau(\cdot)}} \min \{ \mu, \lambda\} \tau(t) dt$. When $z_T < z_C$, all users strictly prefer transit, so $x_C^{\tau(\cdot)}(t^*) = 0$ for all $t^*$ and no toll revenue is generated. In this work, for brevity, we drop the superscript $\tau(\cdot)$ in the notation for the equilibrium interval $[t_A^{\tau(\cdot)}, t_D^{\tau(\cdot)}]$ when it is clear from context. Furthermore, in studying equilibrium outcomes, we adopt a revenue-favoring equilibrium selection convention, selecting an equilibrium with the highest toll revenue.

\emph{Total System Cost:} The total system cost, $SC(\tau(\cdot))$, measures the aggregate travel burden borne by users and equals the sum of waiting time costs, schedule delay costs, and generalized costs (i.e., $z_T$ or $z_C$) at the equilibrium induced by $\tau(\cdot)$. Equivalently, it is the sum of users’ total travel costs at equilibrium \emph{minus} toll payments, i.e., $SC(\tau(\cdot))
=
\int_{t_1}^{t_2} \lambda  \left[
x_C^{\tau(\cdot)}(t^*)\Big(c\big(t^{\tau(\cdot)}(t^*),t^*\big)-\tau\big(t^{\tau(\cdot)}(t^*)\big)\Big)
+
\big(1-x_C^{\tau(\cdot)}(t^*)\big) z_T
\right] dt^*.$ Consistent with the congestion pricing literature~\citep{vickrey1969congestion,GONZALES20121519}, we include the costs $z_C$ and $z_T$ in the total system cost but exclude toll payments as transfers between users and the planner. %where $c(t,t^\star)$ is the travel cost (in time units) for a car user from Equation~\eqref{eq:userCostNormalized}. %The above expression subtracts toll payments since tolls are transfers between users and the central planner, consistent with standard practice of defining system cost of welfare metrics~\citep{vickrey1969congestion,GONZALES20121519,a2fd2512-18e4-3e86-a4f2-2f4443994eb2}. % that excludes tolls (or prices) when defining system cost or welfare metrics.

\vspace{-6pt}
\section{Static vs. Dynamic Tolling in the Bottleneck Model} \label{sec:revenue-optimal-tolling-classical-bottleneck}
\vspace{-4pt}
%This section begins our performance comparison between static revenue-optimal tolling and its revenue and system-cost-optimal dynamic tolling counterparts in the bottleneck model with an outside option. 
In this section, we first derive closed-form expressions for revenue-optimal static and dynamic tolls in Sections~\ref{subsec:rev-opt-static} and~\ref{subsec:rev-opt-dynamic}, respectively. Leveraging these characterizations, we then quantify the performance gap between static revenue-optimal tolling and its dynamic optimal benchmarks by establishing lower bounds on the fraction of optimal revenue attained and multiplicative upper bounds on total system cost in Sections~\ref{subsec:rev-comp-static-dynamic} and~\ref{subsec:system-cost-bottleneck-comp}. Finally, in Section~\ref{sec:stochastic-extension}, we extend these guarantees to a setting in which the desired arrival rate and bottleneck capacity are stochastic.

We focus on the non-trivial and empirically relevant regimes when $\mu < \lambda$, i.e., the congested regime when the bottleneck service rate is strictly below the desired bottleneck departure rate, and $z_T \geq z_C$, i.e., car travel under free-flow is preferred to transit. Note that when $\mu \geq \lambda$, congestion does not arise and the static and dynamic optimal tolls on both revenue and system cost metrics coincide, corresponding to a uniform toll of $\max \{ z_T - z_C, 0 \}$. Moreover, when $z_T<z_C$, all users strictly prefer transit under any tolling policy, yielding no revenue and a total system cost of $z_T\Lambda$. %Throughout this section, when multiple equilibria arise due to indifference between car and transit, we adopt a revenue-favoring equilibrium selection convention, selecting an equilibrium with the highest toll revenue.

\vspace{-3pt}
\subsection{Static Revenue-Optimal Tolls} \label{subsec:rev-opt-static}
\vspace{-3pt}
This section characterizes the revenue-optimal static toll $\tau$ in the bottleneck model with an outside option. %, where the toll is constant over time (i.e., $\tau(t) = \tau$ for all $t$). 
Under a static toll $\tau$, let $N_T(\tau)$ denote the number of users that choose transit at equilibrium. The toll revenue is given by $R(\tau) = \tau (\Lambda - N_T(\tau))$. Proposition~\ref{prop:rev-opt-static-tolls} derives an expression for the revenue-optimal static toll $\tau_s^*$, showing that it depends on the magnitude of the cost difference between transit and car travel at free-flow, given by $z_T - z_C$. When this difference is below a threshold, revenue is maximized by setting the toll at its highest feasible level, $z_T-z_C$, since any higher toll would make even a congestion-free car trip more costly than transit and induce all users to switch modes. In contrast, when this difference is larger than that threshold, the revenue-optimal static toll is strictly lower and depends on other model parameters, as characterized in Proposition~\ref{prop:rev-opt-static-tolls}.

%shows that the revenue-optimal static toll depends on the magnitude of $z_T - z_C$, the cost difference between transit and car travel at free-flow. When this difference is below a certain threshold, it is optimal to set the toll at its maximum feasible level, $z_T - z_C$, since any higher toll would induce all users to switch to transit. In contrast, when this difference is larger than that threshold, the revenue-optimal toll is strictly lower and depends on the remaining model parameters, as characterized in Proposition~\ref{prop:rev-opt-static-tolls}.

%Under a static toll, letting $N_T(\tau)$ denote the number of users taking transit at equilibrium under the toll $\tau$, the expression for the total toll revenue simplifies to $R(\tau) = \tau (\Lambda - N_T(\tau))$. 

%When this magnitude is small, it is optimal to set the maximum possible toll of $z_T - z_C$ (as any higher toll would make transit more attractive at equilibrium), and when this is large, then it is optimal to set the toll to a smaller value that depends on the other problem parameters.

\vspace{-3pt}
\begin{proposition}[Revenue-Optimal Static Tolls] 
\label{prop:rev-opt-static-tolls}
    Suppose $\mu < \lambda$ and users choose between traveling by car through a bottleneck with a free flow cost of $z_C$, and using public transit with a fixed cost of $z_T$, where $z_T \geq z_C$.
    %Under Assumption~\ref{asmpn:basic-sec4}, 
    Then, the revenue-optimal static toll is:
    {\setlength{\abovedisplayskip}{2pt}
\setlength{\belowdisplayskip}{2pt}
    \begin{align} \label{eq:optimal-static-toll}
    \tau^* = 
    \begin{cases}
    \max \left\{ \frac{z_T - z_C}{2} + \frac{\Lambda eL}{2(\lambda - \mu)(e+L)}, z_T - z_C - \frac{\Lambda e L}{\mu(e+L)} \right\}, & \text{if } z_T - z_C \geq \frac{\Lambda eL}{(\lambda - \mu)(e+L)} \\[-2pt]
    z_T - z_C, & \text{if } 0 < z_T - z_C < \frac{\Lambda eL}{(\lambda - \mu)(e+L)}. %\\
    %0, & \text{if } z_T - z_C < 0
\end{cases}
\end{align}}
\end{proposition}

\begin{proof}%[Proof of Proposition~\ref{prop:rev-opt-static-tolls}]
First note that a static toll $\tau \geq 0$ must be such that $z_T - z_C - \tau \leq T_C$. Note that if $z_T - z_C - \tau > T_C$, then $\tau$ can be increased to $z_T - z_C - T_C$ without changing the number of car users, resulting in a higher revenue. Next, for any $\tau \geq \max \{ 0, z_T - z_C - T_C \}$, note that the equilibrium is as specified in Proposition~\ref{prop:user-eq-two-modes} other than that the maximum waiting time is $\Bar{w} = z_T - z_C - \tau$. 

%Next, %to derive an expression for the revenue, $R(\tau) = \tau \left[ \Lambda - N_T(\tau) \right]$, 
Next, we derive an expression for the number of users traveling by transit under the toll $\tau$. By the linearity of the waiting time profile, all car users that pass the bottleneck either early or late (i.e., between $[t_A, t_B]$ and $[t_C, t_D]$, respectively, on the right of Figure~\ref{fig:equilibrium_both_models}) correspond to a $\frac{\Bar{w}}{T_C}$ fraction of all users. The remaining $1-\frac{\Bar{w}}{T_C}$ fraction arrive exactly on time (either using transit or car), where the total on-time car users is $N_o(\tau) = (1-\frac{\Bar{w}}{T_C})(t_2 - t_1)\mu = (1-\frac{\Bar{w}}{T_C})\frac{\Lambda}{\lambda}\mu$ and transit users is $N_T(\tau) = (1-\frac{\Bar{w}}{T_C})\Lambda \left( 1 - \frac{\mu}{\lambda}\right)$. We then obtain the following expression for the revenue as a function of $\tau$: $R(\tau) \! = \! \tau \left[ \Lambda - N_T(\tau) \right] \! = \! \tau \left[ \Lambda \! - \! \left(1-\frac{z_T - z_C - \tau}{T_C} \right)\Lambda \left( 1 \! - \! \frac{\mu}{\lambda}\right) \right] \! \stackrel{(a)}{=} \! \mu \tau \left[ \frac{\Lambda}{\lambda} \! + \! \frac{(z_T  - z_C - \tau)(e \! + \! L)}{eL} \left( 1 \! - \! \frac{\mu}{\lambda} \right) \right],$ %$R(\tau) \! = \! \tau \left[ \Lambda - N_T(\tau) \right] \! = \! \tau \left[ \Lambda \! - \! \left(1-\frac{z_T - z_C - \tau}{T_C} \right)\Lambda \left( 1 \! - \! \frac{\mu}{\lambda}\right) \right] \! \stackrel{(a)}{=} \! \mu \tau \left[ \frac{\Lambda}{\lambda} \! + \! \frac{(z_T  - z_C - \tau)(e \! + \! L)}{eL} \left( 1 \! - \! \frac{\mu}{\lambda} \right) \right],$ %\begin{comment}
%Then, since $N_T = (1-\frac{T}{T_C})\Lambda \left( 1 - \frac{\mu}{\lambda}\right)$ by Lemma~\ref{lem:number of on-time arrivals by car}, it follows that the revenue as a function of the toll $\tau$ is: % given by:
%\end{comment}
where (a) follows as $T_C = \frac{\Lambda e L}{\mu (e + L)}$. Consequently, the static revenue optimization problem is:
{\setlength{\abovedisplayskip}{2pt}
\setlength{\belowdisplayskip}{2pt}
\begin{align} \label{eq:rev-opt-static-bottleneck}
    \max_{\tau \in \mathbb{R}} R(\tau) \! := \! \mu \tau \left[ \frac{\Lambda}{\lambda} + \frac{(z_T - z_C - \tau)(e+L)}{eL} \left( 1 - \frac{\mu}{\lambda} \right) \right] \quad \! \! \! \text{s.t.} \! \! \! \! \quad \max \{ 0, z_T - z_C - T_C \} \! \leq \! \tau \! \leq \! z_T - z_C.
\end{align}}
Taking the first-order condition of this single-dimensional quadratic optimization problem, we can show that its optimal solution corresponds to the static tolls in the statement of the proposition.
\end{proof}
%\vspace{-2pt}

Proposition~\ref{prop:rev-opt-static-tolls} states that when the cost difference $z_T - z_C$ is small, transit is relatively attractive to car travel at free flow and revenue is maximized at the highest feasible toll, $\tau = z_T - z_C$. As $z_T - z_C$ increases, the planner trades off a higher toll against diverting users to transit. In this regime, it suffices to compare two candidate tolls: the unconstrained maximizer of the quadratic revenue objective in Equation~\ref{eq:rev-opt-static-bottleneck}, and the boundary toll $z_T-z_C-T_C$, below which all users travel by car. Taking the maximum revenue of these two candidates also ensures that the revenue-optimal toll is non-negative. To obtain this characterization of the revenue-optimal static toll, we cast the static revenue maximization problem as a single-variable quadratic program, where the optimal toll depends on the magnitude of $z_T - z_C$ relative to the threshold $\frac{\Lambda eL}{(\lambda-\mu)(e+L)}$. A key step in proving this result requires extending the equilibrium analysis in Proposition~\ref{prop:user-eq-two-modes}~\citep{GONZALES20121519}, which characterizes monotonicity of the number of users choosing transit at equilibrium, to obtain an explicit expression for this quantity. This characterization enables closed-form expressions for both revenue and total system cost under the optimal static toll across parameter regimes, which underpins our performance comparison between static and dynamic tolling in Sections~\ref{subsec:rev-comp-static-dynamic} and~\ref{subsec:system-cost-bottleneck-comp}.

\vspace{-5pt}
\subsection{Dynamic Revenue-Optimal Tolls} \label{subsec:rev-opt-dynamic}
\vspace{-4pt}
We now characterize the revenue-optimal dynamic tolling policy $\tau_d^*(\cdot)$. Theorem~\ref{thm:rev-opt-dynamic-tolls} establishes that the revenue-optimal dynamic toll has the trapezoidal structure shown in Figure~\ref{fig:rev-opt-dynamic-tolls}, remaining constant at $z_T - z_C$ over the middle of the rush, i.e., over the interval $[t_B^*, t_C^*]$, and decreasing linearly on either side of this interval. The corresponding equilibrium arrival and departure (crossing) time profiles of users at the bottleneck are as depicted in Figure~\ref{fig:arrival_profile_dynamic_opt_combined} in Appendix~\ref{apdx:arrival-departure-opt}.

%Theorem~\ref{thm:rev-opt-dynamic-tolls} shows that the optimal toll takes the form shown in Figure~\ref{fig:rev-opt-dynamic-tolls}, where the toll is constant and fixed at $z_T - z_C$ in the middle of the rush and decreases linearly as we move away from the rush.

%For notational brevity, we drop the superscript $\tau^*(\cdot)$ in the notation of the times $t_A, t_B, t_C$, and $t_D$

%we denote $[t_A, t_D]$ as the equilibrium interval over which users pass the bottleneck under the optimal tolling policy 
\vspace{-3pt}
\begin{theorem}[Revenue-Optimal Dynamic Tolls] \label{thm:rev-opt-dynamic-tolls}
Suppose $\mu < \lambda$ and users choose between traveling by car through a bottleneck with a free flow cost $z_C$ and using transit with a fixed cost $z_T$, where $z_T \geq z_C$. 
%Suppose Assumption~\ref{asmpn:basic-sec4} holds, 
Then, the revenue-optimal dynamic tolling policy $\tau_d^*(\cdot)$ is such that (see Figure~\ref{fig:rev-opt-dynamic-tolls}): \vspace{-2pt}
\begin{enumerate}
    \item The toll is constant and fixed at $\tau_d^*(t) = z_T - z_C$ for all $t \in [t_B^*, t_C^*]$ during the middle of the rush. Moreover, the time interval $[t_B^*, t_C^*]$ over which the toll remains constant corresponds to a fraction $f^* = \max \{ 1 - \frac{(z_T - z_C) \mu (e+L)}{\Lambda e L} \left( 1 - \frac{\mu}{\lambda} \right), 0 \}$ of the interval $[t_1, t_2]$, i.e., $\frac{t_C^* - t_B^*}{t_2 - t_1} = f^*$.
    %The number of early car users $N_e = \frac{\mu T}{e}$ and they travel at the beginning of the morning rush between $[t_A, t_B]$. The number of late car users $N_L = \frac{\mu T}{L}$ and they travel at the end of the morning rush between $[t_C, t_D]$.
    \item For all $t \in [t_A^*, t_B^*]$, $\tau^*(t) = z_T - z_C - e (t_B^* - t)$, and for all $t \in [t_C^*, t_D^*]$, $\tau^*(t) = z_T - z_C - L (t - t_C^*)$, where $t_B^* - t_A^* = \frac{L}{e+L} \frac{(1 - f^*) \Lambda}{\mu}$, and $t_D^* - t_C^* = \frac{e}{e+L} \frac{(1 - f^*) \Lambda}{\mu}$. For $t<t_A^*$ and $t>t_D^*$, extend $\tau^*(t)$ linearly with slopes $e$ and $-L$, respectively, until the toll reaches zero, and set it to zero thereafter.
\end{enumerate}
Moreover, under this tolling policy, the optimal revenue is given by: 
{\setlength{\abovedisplayskip}{0pt}
\setlength{\belowdisplayskip}{0pt}
\setlength{\jot}{-1pt}
\begin{align} \label{eq:dynamic-opt-rev}
    R^* =
    \begin{cases}
        (z_T - z_C) \Lambda \frac{\mu}{\lambda} + \frac{(z_T - z_C)^2 \mu (e+L)}{2 e L} \left( 1 - \frac{\mu}{\lambda} \right)^2, & \text{if } z_T - z_C \leq \frac{\Lambda e L}{ (\lambda - \mu) (e+L)} \frac{\lambda}{\mu}, \\[-2pt]
        (z_T - z_C) \Lambda - \frac{\Lambda^2}{2 \mu} \frac{eL}{e+L}, & \text{if } z_T - z_C > \frac{\Lambda e L}{ (\lambda - \mu) (e+L)} \frac{\lambda}{\mu}.
    \end{cases}
\end{align}}
\end{theorem}

\vspace{-15pt}

\begin{figure}[tbh!]
      \centering
      \includegraphics[width=75mm]{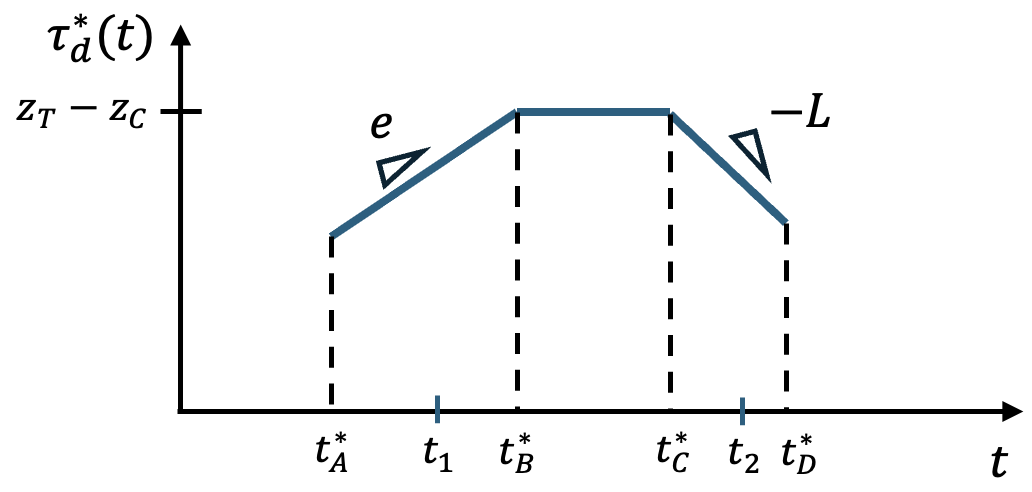}
      \vspace{-13pt}
      \caption{\small \sf Depiction of the revenue-optimal dynamic tolling policy. Here, $[t_A^{*}, t_D^{*}]$ denotes the equilibrium interval over which users pass the bottleneck under this tolling policy $\tau_d^*(\cdot)$, analogous to the no-toll equilibrium interval $[t_A, t_D]$ shown in the right of Figure~\ref{fig:equilibrium_both_models}. The sub-interval $[t_B^{*}, t_C^{*}]$ corresponds to the middle of the rush, during which the toll is constant and equal to $z_T - z_C$.}
      \label{fig:rev-opt-dynamic-tolls} 
   \end{figure}  

%\vspace{-4pt}

\vspace{-10pt}

\begin{proof}
%Recall that the total cost for a car user departing the bottleneck at time $t$ with a desired bottleneck departure time of $t^*$ is given by Equation~\eqref{eq:userCostNormalized}. 
Under the cost function in Equation~\eqref{eq:userCostNormalized}, car users at equilibrium satisfy $c'(t, t^*) = 0$, implying $w'(t) + \tau'(t) = e$ for early arrivals and $w'(t) + \tau'(t) = -L$ for late arrivals. Moreover, %at equilibrium,
$w(t) + \tau(t) \leq z_T - z_C$ for all $t$, as otherwise users would switch to transit. Under the standard equilibrium ordering of car users by desired crossing times in the bottleneck model~\citep{GONZALES20121519}, these conditions imply that the sum of the equilibrium waiting and toll costs are as shown in Figure~\ref{fig:eq-waiting-tolls}, where the times $t_A', t_B', t_C', t_D'$ are determined endogenously by the tolling policy.

%a necessary condition at equilibrium is that all users minimize travel costs when using a car, i.e., the following condition must hold: 

%the equilibrium structure depicted in Figure~\ref{fig:eq-waiting-tolls},

%Thus, the sum of the equilibrium waiting and toll costs take the form depicted in Figure~\ref{fig:eq-waiting-tolls}, where the exact times $t_A', t_B', t_C', t_D'$ shown in the figure are determined endogenously by the tolling policy.

%the sum of the waiting and toll costs when using a car is at most $z_T - z_C$, i.e., 

%\vspace{-4pt}
%\begin{align*}
%    c'(t, t^*) = 0 \quad \implies \quad w'(t) + \tau'(t) = \frac{c_S}{c_W}.
%\end{align*}

%are given by Figure~\ref{fig:eq-waiting-tolls} for some periods $t_A', t_B', t_C', t_D'$, which depend on the tolling policy.

\vspace{-5pt}
\begin{figure}[tbh!]
      \centering
      \includegraphics[width=75mm]{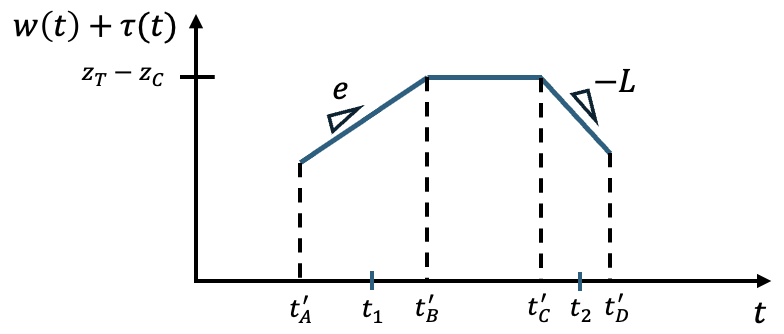}
      \vspace{-10pt}
      \caption{\small \sf Sum of waiting plus toll costs at equilibrium. The times $t_A', t_B', t_C', t_D'$ are determined endogenously by the tolling policy and may differ from the time points shown in Figure~\ref{fig:equilibrium_both_models} (right) and Figure~\ref{fig:rev-opt-dynamic-tolls}, corresponding to the no-toll and revenue-optimal dynamic toll settings, respectively.}
      \label{fig:eq-waiting-tolls} 
   \end{figure} 

\vspace{-5pt}

Leveraging the equilibrium structure of the sum of the waiting and toll costs in Figure~\ref{fig:eq-waiting-tolls}, we show that for dynamic revenue-optimal policies, it suffices to restrict attention to tolling schemes with $w(\cdot) \equiv 0$. To see this, for any policy $\tau(\cdot)$ with associated periods $t_A', t_B', t_C', t_D'$ in Figure~\ref{fig:eq-waiting-tolls}, define a policy $\Tilde{\tau}(\cdot)$, where $\Tilde{\tau}(t) = w(t) + \tau(t)$ for all $t$. Under \(\Tilde{\tau}(\cdot)\), the original bottleneck crossing times and mode choices under $\tau(\cdot)$ constitute an equilibrium with \(w(t)=0\). Feasibility is preserved because the same users pass the bottleneck at the same times, while each user's generalized cost at every possible crossing time is unchanged since \(\tilde{\tau}(t)=w(t)+\tau(t)\). Then, we have $R(\tau(\cdot)) = \int_{t_A'}^{t_D'} \mu \tau(t) dt \leq \int_{t_A'}^{t_D'} \mu \Tilde{\tau}(t) dt = R(\Tilde{\tau}(\cdot)),$ where the inequality follows as $\Tilde{\tau}(t) \geq \tau(t)$ for all $t$.

%Consider the queue-free allocation in which each car user crosses the bottleneck at the same time as under the original equilibrium, which is feasible since the original crossing rate never exceeds \(\mu\). Moreover, for any user with desired arrival time \(t^*\) and any candidate crossing time \(t\),

Thus, focusing on tolling policies $\Tilde{\tau}(\cdot)$ where the waiting time $w(t) = 0$ for all $t$, we have: %. Then, the revenue of the policy $\Tilde{\tau}(\cdot)$ with no waiting time is given by:
{\setlength{\abovedisplayskip}{1pt}
\setlength{\belowdisplayskip}{1pt}
\begin{align*}
    R(\Tilde{\tau}(\cdot)) &= \int_{t_A'}^{t_D'} \mu \Tilde{\tau}(t) dt = \int_{t_A'}^{t_B'} \mu \Tilde{\tau}(t) dt + \int_{t_B'}^{t_C'} \mu \Tilde{\tau}(t) dt + \int_{t_C'}^{t_D'} \mu \Tilde{\tau}(t) dt, \\
    &= \mu \bigg[ \int_{t_A'}^{t_B'} z_T - z_C - e(t_B' - t) dt + \int_{t_B'}^{t_C'} (z_T - z_C) dt + \int_{t_C'}^{t_D'} z_T - z_C - L(t - t_C') dt \bigg], \\
    &= \mu \bigg[ (z_T - z_C) (t_D' - t_A') - \frac{e (t_B' - t_A')^2}{2} - \frac{L (t_D' - t_C')^2}{2} \bigg].
\end{align*}}
Let $f$ be the fraction of time in the horizontal portion of the waiting time curve, i.e., $\frac{t_C' - t_B'}{t_2 - t_1} = f$. Then, $t_C' - t_B' = f (t_2 - t_1) = f \frac{\Lambda}{\lambda}$. Since the remaining $(1-f)$ fraction of users all use a car, it follows that $\mu (t_B' - t_A' + t_D' - t_C') = (1-f) \Lambda$. Thus, the above revenue expression reduces to: $R(\Tilde{\tau}(\cdot)) = (z_T - z_C) \left(f \Lambda \frac{\mu}{\lambda} + (1 - f)\Lambda \right) - \frac{e \mu (t_B' - t_A')^2}{2} - \frac{L \mu (t_D' - t_C')^2}{2}.$

Next, we know that $t_B' - t_1 + t_2 - t_C' = (1-f) \frac{\Lambda}{\lambda}$, where $\lambda(t_B' - t_1) = \mu (t_B' - t_A')$ and $\lambda (t_2 - t_C') = \mu (t_D' - t_C')$. Then, it is straightforward to see that to minimize $\frac{e \mu (t_B' - t_A')^2}{2} + \frac{L \mu (t_D' - t_C')^2}{2}$, it must be that $t_B' - t_A' = \frac{L}{e+L} \frac{(1-f) \Lambda}{\mu}$ and $t_D' - t_C' = \frac{e}{e+L} \frac{(1-f) \Lambda}{\mu}$. We thus obtain the following expression for the revenue, which, with a slight abuse of notation, we re-express as a function of the fraction $f$:
{\setlength{\abovedisplayskip}{1pt}
\setlength{\belowdisplayskip}{1pt}
\begin{align} \label{eq:rev-exp-f}
    R(f) = (z_T - z_C) \left(f \Lambda \frac{\mu}{\lambda} + (1 - f)\Lambda \right) - \frac{\Lambda^2}{2 \mu} \frac{eL}{e+L}(1-f)^2.
\end{align}}

Given this relation for the revenue, we have the following revenue optimization problem:
{\setlength{\abovedisplayskip}{2pt}
\setlength{\belowdisplayskip}{2pt}
\begin{align} \label{eq:constrained-dynamic-rev-opt}
    \max_{f \in [0, 1]} \! \! R(f) \! = \! (z_T - z_C) \! \left(f \Lambda \frac{\mu}{\lambda} \! + \! (1 - f)\Lambda \right) \! - \! \frac{\Lambda^2}{2 \mu} \frac{eL}{e+L}(1-f)^2 \quad \! \! \! \text{s.t.} \! \! \! \! \! \quad 1 - \min \left\{ \frac{z_T - z_C}{T_C}, 1 \right\} \! \leq \! f \! \leq \! 1.
\end{align}}
The first term in the objective represents the revenue if all car users paid the maximum toll $z_T - z_C$, while the second captures the revenue reduction from the linearly declining tolls on the early and late portions of the rush. We have thus reduced the dynamic revenue optimization problem from one of optimizing over a set of tolling functions to optimizing over a single dimensional variable $f$. Note here that when $z_T - z_C \leq T_C = \frac{\Lambda e L}{\mu (e+L)}$, the lower bound constraint reduces to $f \geq 1 - \frac{z_T - z_C}{T_C}$ to ensure that the tolls are non-negative at all periods. Then, taking the above problem's first order condition, we obtain $f^* = \max \left\{ 1 - \frac{(z_T - z_C) \mu (e+L)}{\Lambda e L} \left( 1 - \frac{\mu}{\lambda} \right), 0 \right\}$. Finally, substituting $f^*$ in~\eqref{eq:rev-exp-f}, we obtain~\eqref{eq:dynamic-opt-rev}. %, which establishes our claim.
\end{proof}
%\vspace{-3pt}
Theorem~\ref{thm:rev-opt-dynamic-tolls} characterizes the revenue-optimal dynamic toll in the bottleneck model with an outside option. Under this tolling policy, no congestion delays arise at equilibrium, which intuitively follows as the waiting time and tolls both enter users' travel costs, but only tolls generate revenue. Thus, the planner can increase revenue by replacing waiting costs with toll payments while preserving users' travel costs. Leveraging this structure, the proof of Theorem~\ref{thm:rev-opt-dynamic-tolls} reduces the revenue-maximization problem, which requires optimizing over tolling functions, to a single-variable quadratic program. To our knowledge, this is the first analysis of dynamic revenue maximization in the bottleneck model with an outside option. The resulting revenue-optimal policy closely resembles the system-cost-optimal dynamic toll in~\citet{GONZALES20121519}, which also eliminates congestion delays at equilibrium and has a trapezoidal form that is, in particular, identical to the equilibrium waiting time profile shown on the right of Figure~\ref{fig:equilibrium_both_models}. Despite this similarity, the two policies and their analyses differ in key ways, as~\cite{GONZALES20121519} study system cost minimization for which the tolling profile is pinned down by the equilibrium waiting time profile in Figure~\ref{fig:equilibrium_both_models} (right), whereas our revenue objective requires optimizing over the duration of the flat segment within this trapezoidal family. Thus, by comparing the functions in Figure~\ref{fig:equilibrium_both_models} (right) and Figure~\ref{fig:rev-opt-dynamic-tolls}, the revenue-optimal policy has the same trapezoidal structure as the system-cost-optimal policy, with the duration of the flat segment, where the toll equals $z_T - z_C$, chosen to maximize revenue.

\vspace{-3pt}

\subsection{Revenue Comparison for Static vs. Dynamic Revenue-Optimal Tolls} \label{subsec:rev-comp-static-dynamic}

\vspace{-3pt}

Leveraging the characterizations of the static and dynamic revenue-optimal tolls, we compare their performance by deriving lower bounds on the fraction of the optimal dynamic revenue attained under the static revenue-optimal policy. Theorem~\ref{thm:rev-comp-static-v-dynamic} shows that the revenue gap between the two policies depends on the magnitude of $z_T - z_C$, the difference between the cost of transit and car travel at free-flow. When $z_T - z_C$ lies below a lower threshold or above an upper threshold, the static revenue-optimal toll achieves at least two-thirds of the optimal dynamic revenue. In the intermediate regime, this revenue fraction is bounded below by one-half, implying that static revenue-optimal tolls achieve at least half of the dynamic optimal revenue across all parameter regimes. Sharper, regime-specific bounds are provided in the theorem statement.

\vspace{-3pt}
\begin{theorem}[Revenue Ratio: Static vs. Dynamic Revenue-Optimal Tolls] \label{thm:rev-comp-static-v-dynamic}
Suppose $z_T \geq z_C$ and $\mu < \lambda$. 
%Suppose Assumption~\ref{asmpn:basic-sec4} holds. 
Moreover, let $\tau^*_{s}$ be the static revenue-optimal toll, $\tau^*_d(\cdot)$ be the dynamic revenue-optimal policy, and let $s = \frac{\Lambda e L}{(z_T - z_C)(\lambda - \mu)(e+L)} \geq 0$. Then, the ratio of their revenues satisfies:
{\setlength{\abovedisplayskip}{1pt}
\setlength{\belowdisplayskip}{1pt}
\begin{align*}
    \frac{R(\tau_s^*)}{R(\tau_d^*(\cdot))} \geq 
    \begin{cases}
    \frac{2}{3 - \frac{\mu}{\lambda}}, & \text{if } 0 < z_T - z_C < \frac{\Lambda eL}{(\lambda - \mu)(e+L)} \\
    \min \left\{ \frac{2+s}{4}, \frac{1}{2(1 - \frac{\mu}{\lambda})} \right\}, & \text{if } z_T - z_C \in \left[ \frac{\Lambda eL}{(\lambda - \mu)(e+L)}, \frac{\Lambda eL}{(e+L)} \left( \frac{1}{\lambda - \mu} + \frac{2}{\mu} \right) \right] \\
    \frac{2}{3}, & \text{if } z_T - z_C > \frac{\Lambda eL}{(e+L)} \left( \frac{1}{\lambda - \mu} + \frac{2}{\mu} \right).
\end{cases}
\end{align*}}
\end{theorem}
\vspace{-3pt}

We establish this result by bounding the ratio of the revenue of the static revenue-optimal toll to the dynamic optimal revenue in Equation~\eqref{eq:dynamic-opt-rev} in the three regimes for $z_T - z_C$ in the theorem statement, each corresponding to a distinct static revenue-optimal toll in Proposition~\ref{prop:rev-opt-static-tolls}. For a proof, see Appendix~\ref{apdx:rev-comp-stativ-v-dynamic}. Theorem~\ref{thm:rev-comp-static-v-dynamic} establishes that revenue-optimal static tolls achieve a constant fraction of the optimal dynamic revenue, at least one-half in some parameter regimes and at least two-thirds in others, highlighting the efficacy of simple static tolls. Importantly, the obtained one-half guarantee is robust to allowing the transit cost to depend on ridership. In Theorem~\ref{thm:general-transit-revenue} of Appendix~\ref{apdx:general-transit-cost}, we show that this one-half guarantee is independent of the specific functional form of transit crowding costs and, in particular, holds for any non-decreasing ridership-dependent transit cost function.

In the fixed transit cost setting in Theorem~\ref{thm:rev-comp-static-v-dynamic}, the regime-specific guarantees reflect the price-participation trade-off underlying revenue maximization in this setting. When the cost difference $z_T - z_C$ is small, transit is relatively attractive to car travel at free flow and the static revenue-optimal toll charges the highest feasible toll $z_T - z_C$, while the dynamic policy can gain only by lowering tolls at the shoulders to retain additional car users. Since the revenue available from these additional users is also limited by the small cost difference between car travel and transit at free flow, static tolling achieves at least two-thirds of the dynamic optimal revenue. At intermediate values of $z_T - z_C$, dynamic pricing has greater scope to exploit time variation, resulting in a weaker one-half guarantee of the ratio of the static to dynamic optimal revenue. However, this revenue ratio remains bounded below by one-half as the dynamic revenue-optimal toll has a trapezoidal structure while the number of car users under a static toll decreases linearly with the toll. Geometrically, since cars are served at the fixed bottleneck capacity $\mu$, dynamic revenue is proportional to the area under the trapezoidal toll profile, whereas a static toll generates a rectangular revenue area, with the best such rectangle capturing at least half of the trapezoidal area. Finally, when $z_T-z_C$ is large, transit becomes sufficiently unattractive that most users travel by car under both static and dynamic revenue-optimal policies, so time variation has little additional effect on the number of users choosing car travel, yielding the stronger two-thirds guarantee.

\vspace{-3pt}
\subsection{System Cost Bounds of Revenue-Optimal Tolling Policies} \label{subsec:system-cost-bottleneck-comp}
\vspace{-3pt}

We now compare the total system cost under static and dynamic revenue-optimal tolls to that under the dynamic system-cost-optimal policy. Unlike the revenue comparisons in Section~\ref{subsec:rev-comp-static-dynamic}, constant-factor system cost guarantees do not hold across all parameters (see Proposition~\ref{prop:unbounded-sc-ratios} at the end of this section). Nevertheless, in the regime where both car and transit are used (or some user is indifferent between car and transit) at the untolled equilibrium i.e., when $z_T - z_C \leq T_C = \frac{\Lambda eL}{\mu (e+L)}$ (see Proposition~\ref{prop:user-eq-two-modes}), a condition characteristic of cities with viable public transit and satisfied by the real systems in our numerical experiments in Section~\ref{sec:experiments}, we establish a constant-factor guarantee. Specifically, in this regime, both static and dynamic revenue-optimal tolls incur at most twice the minimum achievable system cost, as elucidated by the following theorem.

\vspace{-3pt}
\begin{theorem}[System Cost Comparison] \label{thm:sc-comp-static-v-dynamic}
Suppose $\mu < \lambda$, 
%Suppose Assumption~\ref{asmpn:basic-sec4} holds, 
and let $\tau^*_{s}$ be the static revenue-optimal toll, $\tau^*_d(\cdot)$ be the dynamic revenue-optimal toll, and $SC^*$ be the minimum achievable system cost. If the cost difference between transit and car travel at free-flow satisfies $z_T - z_C \leq \frac{\Lambda eL}{\mu (e+L)}$, then %both policies incur at most twice the minimum system cost, i.e., 
$SC(\tau^*_d(\cdot)) \leq 2 SC^*$ and $SC(\tau^*_s) \leq 2 SC^*$.
\end{theorem}
\vspace{-3pt}

\begin{hproof}
%\proof{Proof (Sketch).}
To prove this claim, we first leverage the characterizations of the revenue-optimal static and dynamic tolling policies in Proposition~\ref{prop:rev-opt-static-tolls} and Theorem~\ref{thm:rev-opt-dynamic-tolls}, respectively, to derive expressions for their corresponding total system costs at equilibrium. In particular, for both policies, we compute the total system cost by summing across all users the following four terms: (i) the cost of using transit, (ii) the free-flow cost of using a car, (iii) schedule delay costs, and (iv) waiting time delays, where, recall from the proof of Theorem~\ref{thm:rev-opt-dynamic-tolls} that there are no waiting time delays under the dynamic revenue-optimal policy. Summing these relations, we then prove that the total system cost under both policies remains upper bounded by $z_C \Lambda \frac{\mu}{\lambda} + z_T \left( 1 - \frac{\mu}{\lambda} \right) \Lambda$. Then, comparing this upper bound to the minimum achievable total system cost derived in~\citet{GONZALES20121519}, combined with the fact that $z_T - z_C \leq \frac{\Lambda eL}{\mu (e+L)}$, we obtain our desired bounds.
%\Halmos
%\endproof
%For both tolling policies, we derive expressions for, the sum across all users of, (i) the cost of using transit, (ii) the free-flow cost of using a car, (iii) schedule delay costs, and (iv) waiting time delays, where, recall that the total system cost is the sum of these four terms, where we also note by Observation~\ref{obs:zero-waiting-time} that there are no waiting time delays under the dynamic revenue-optimal policy. Summing these four relations, we find that the total system cost under revenue-optimal static and dynamic tolls always remains upper bounded by $z_C \Lambda \frac{\mu}{\lambda} + z_T \left( 1 - \frac{\mu}{\lambda} \right) \Lambda$.
\end{hproof}

For a complete proof of Theorem~\ref{thm:sc-comp-static-v-dynamic}, see Appendix~\ref{apdx:pf-sc-comp}. Theorem~\ref{thm:sc-comp-static-v-dynamic} shows that when transit provides a viable outside option, i.e., when $z_T - z_C \leq \frac{\Lambda eL}{\mu(e+L)}$, revenue-optimal tolls achieve total system costs within a modest factor of the system-cost-optimal benchmark. We further extend this system-cost guarantee to settings with general transit crowding costs in Appendix~\ref{apdx:general-transit-cost}.

Together, Theorems~\ref{thm:rev-comp-static-v-dynamic} and~\ref{thm:sc-comp-static-v-dynamic} show that static revenue-optimal tolling achieves strong revenue and system cost guarantees by capturing at least a constant fraction of the optimal dynamic revenue while incurring at most twice the minimum achievable system cost (when $z_T - z_C \leq \frac{\Lambda eL}{\mu (e+L)}$). %Moreover, our comparisons abstract from the substantial informational and computational burdens dynamic tolling places on transportation planners and users, and accounting for such implementation frictions may further strengthen the case for static tolling.
%Since our comparisons abstract from the substantial informational and computational burdens that dynamic tolling places on traffic planners and users, these bounds should be interpreted as worst-case relative guarantees. Incorporating realistic behavioral or operational constraints would likely strengthen these bounds, further reinforcing the case for static tolling in practical applications.
The constant-factor guarantees in Theorem~\ref{thm:sc-comp-static-v-dynamic} rely on the cost gap between transit and car at free flow being bounded by $\frac{\Lambda eL}{\mu (e+L)}$, ensuring that transit remains a viable outside option (see Section~\ref{subsec:background-equilibria}). When this condition fails, constant-factor system cost ratios need not hold. To show this, we consider a setting with $z_C=0$ and sufficiently high transit costs, as formalized in the following proposition. %, as elucidated by the following observation. %This regime may arise, for instance, in settings with autonomous driving, where car travel at free flow becomes substantially more attractive than public transit~\cite{ostrovsky2025congestion}.

%We reiterate that we obtain these guarantees without accounting for the informational and operational complexity required to implement dynamic tolling, thus providing a strong practical case for static pricing. Incorporating realistic behavioral or operational constraints would likely strengthen these bounds, further reinforcing the case for static tolling in practical applications.

%which is possible with autonomous driving where using a car at free-flow may become much more attractive compared to transit, constant-factor system cost ratios may, in general, not be possible, as elucidated by the following proposition.

\vspace{-3pt}

\begin{proposition}[Unbounded System Cost Ratios] \label{prop:unbounded-sc-ratios}
Let $\tau^*_{s}$ be the static revenue-optimal toll and $SC^*$ be the minimum achievable total system cost. Further, suppose that $z_C = 0$ and $z_T > \frac{\Lambda eL}{\mu (e+L)} \left( \frac{2\lambda - \mu}{\lambda - \mu} \right) \geq \frac{\Lambda eL}{\mu (e+L)}$. Then, $SC(\tau_s^*) = \left( 1 + \frac{1}{1 - \frac{\mu}{\lambda}} \right) SC^*$, which is unbounded when $\frac{\mu}{\lambda} \rightarrow 1$.
%$SC(\tau^*_d(\cdot)) \leq 2 SC^*$ and $SC(\tau^*_s) \leq 2 SC^*$.
\end{proposition}

\vspace{-3pt}

For a proof of Proposition~\ref{prop:unbounded-sc-ratios}, see Appendix~\ref{apdx:unbounded-sc-ratios}. %This result also extends when comparing the system cost of dynamic revenue-optimal tolls to the minimum achievable system cost. 
The main driver behind this unbounded system cost ratio is that when $z_T - z_C > \frac{\Lambda eL}{\mu (e+L)}$, transit is sufficiently unattractive that revenue-optimal tolls scale with $z_T - z_C$ to extract maximal revenue. In contrast, system-cost-optimal tolls are upper bounded by $\frac{\Lambda eL}{\mu (e+L)}$ and are calibrated to eliminate congestion rather than exploit users’ willingness to pay. Thus, when $z_T - z_C > \frac{\Lambda eL}{\mu (e+L)}$, the revenue and system-cost-optimal policies operate on fundamentally different scales, precluding constant factor system cost guarantees. Despite this result, Theorem~\ref{thm:sc-comp-static-v-dynamic} demonstrates the robustness of revenue-optimal tolling on the system cost metric in regimes where both car and public transit are used at the untolled equilibrium, a practically relevant setting for transportation systems with viable transit alternatives.

\vspace{-3pt}
\subsection{Extension to Stochastic Demand and Bottleneck Capacity} \label{sec:stochastic-extension}
\vspace{-3pt}

Thus far, we have considered a deterministic setting in which the desired arrival rate $\lambda$ and bottleneck capacity $\mu$ are fixed. We now allow both quantities to be stochastic, capturing fluctuations in travel demand or capacity disruptions due to incidents, as studied in several extensions of the bottleneck model~\citep{arnott1999information,xiao2015congestion,liu2025pricing}. In this section, following~\citet{arnott1999information}, we assume that the realized demand and capacity are observed before users make their mode and departure-time choices, while the tolling policy may or may not be able to adapt to these realizations, as specified below.

Formally, we model the uncertainty through a random pair $(\lambda_{\omega}, \mu_{\omega})$, where $\lambda_{\omega}>0$ denotes the desired arrival rate and $\mu_{\omega}>0$ the bottleneck capacity in realization $\omega \in \Omega$. We allow $(\lambda_{\omega},\mu_{\omega})$ to have an arbitrary joint distribution so demand and capacity may be correlated, and assume $\mathbb{E}[\lambda_{\omega}] < \infty$. Conditional on $\omega$, the remainder of the setup is as in the deterministic model studied thus far. Specifically, users' desired arrival times are uniformly distributed over $[t_1,t_2]$, and the total mass of users is $\Lambda_{\omega}:=\lambda_{\omega}(t_2-t_1)$. All other model primitives remain fixed across realizations, including the schedule delay parameters $e$ and $L$, the transit cost $z_T$, and the free-flow car cost $z_C$.

In this setting, we compare the performance of two tolling policies: (i) a fixed static toll that selects a single toll $\tau$ before $(\lambda_{\omega}, \mu_{\omega})$ is realized and applies the same toll in every realization $\omega \in \Omega$, and (ii) an adaptive dynamic toll that observes the realization and can choose a state-dependent, time-varying tolling policy $\tau_{\omega}(\cdot)$. Then, letting $R_{\omega}(\cdot)$ denote the toll revenue under realization $\omega$, we define the optimal expected revenues under these two policy classes as
{\setlength{\abovedisplayskip}{1pt}
\setlength{\belowdisplayskip}{1pt}
\begin{equation}
R_s^{\mathrm{fix}}
:=\max_{\tau \geq 0}\mathbb{E}[R_\omega(\tau)],
\qquad
R_d^{\mathrm{adapt}}
:=\mathbb{E}\!\left[\max_{\tau_\omega(\cdot)\geq0}R_\omega\!\left(\tau_\omega(\cdot)\right)\right].
\label{eq:stochastic_revenues}
\end{equation}}
Thus, the adaptive dynamic policy can adjust to the realized demand and capacity while varying tolls over time, whereas the fixed static toll is chosen ex ante and remains unchanged across realizations and over time. Note that if a static toll were instead allowed to condition on the realization $\omega$, then, for each realization, the performance comparison between static and dynamic tolling would reduce to the deterministic analysis developed above. Instead, we consider a fixed static toll that is chosen before the uncertainty is realized, consistent with commonly deployed congestion pricing mechanisms, such as the tolls in New York City and on the SF-Oakland Bay Bridge, which do not adjust in response to realized fluctuations in demand or capacity. Moreover, we compare this weaker fixed static toll against the stronger adaptive dynamic tolling benchmark, thereby obtaining a worst-case guarantee on the performance loss from using a single realization-independent static toll and an assessment of the robustness of static tolling in a stochastic environment.

%thereby providing a worst-case guarantee for the performance loss from using a single realization-independent static toll while allowing us to assess the robustness of static tolling in a stochastic environment.

We now establish that a single fixed static revenue-optimal toll achieves at least one-half of the expected revenue of the adaptive dynamic benchmark, while preserving the factor-two system cost guarantee under an analogous condition to that in the deterministic setting. We first state our revenue guarantee.

%For brevity, in the following, we state the revenue guarantee and defer the system cost guarantee to Appendix~\ref{apdx:stochastic-sc}.

%Theorem~\ref{thm:stochastic-system-cost} in Appendix~\ref{apdx:stochastic-sc}, where we introduce additional notation and present proof details.

\vspace{-3pt}
\begin{theorem}[Revenue Guarantee in Stochastic Bottleneck] \label{thm:stochastic-extension}
Suppose $(\lambda_{\omega},\mu_{\omega})$ is drawn from an arbitrary joint distribution with $\lambda_{\omega}>0$, $\mu_{\omega}>0$, $\mathbb{E}[\lambda_{\omega}]<\infty$, and the realization $\omega$ is observed before users make their mode and departure-time choices. Then, the optimal fixed static toll achieves at least half the expected revenue of the adaptive dynamic benchmark, i.e., $R_s^{\mathrm{fix}}
\geq \frac{1}{2}R_d^{\mathrm{adapt}}.$
\end{theorem}
\vspace{-3pt}

The key intuition is that, for each realization $\omega$, the optimal dynamic revenue is bounded above by the area under the curve that maps a static toll to the resulting number of car users. The proof shows that the rectangle generated by the midpoint toll $\frac{z_T-z_C}{2}$ has area at least one-half of this area for every realization. Since the same midpoint toll applies regardless of the realized demand or capacity, taking expectations preserves the guarantee. Hence, a single static toll chosen ex ante can capture at least one-half of the expected dynamic revenue. For a detailed proof of Theorem~\ref{thm:stochastic-extension}, see Appendix~\ref{apdx:rev-stochastic-extension}. 

We further obtain the following stochastic analogue of the deterministic system cost guarantee in Theorem~\ref{thm:sc-comp-static-v-dynamic}. To elucidate this result, let $\tau_s^{\mathrm{fix}}$ denote a fixed static toll that maximizes expected revenue, i.e., attains $R_s^{\mathrm{fix}}$. Further, for each realization $\omega$, let $\mathrm{SC}_\omega(\tau_s^{\mathrm{fix}})$ denote the resulting system cost and let $\mathrm{SC}_\omega^*$ denote the minimum achievable system cost under that realization.
%In the following, we focus on the case when the bottleneck capacity is less than the desired arrival rate (i.e., $\mu_\omega<\lambda_\omega$) and note that the corresponding guarantee trivially extends to the case when $\mu_\omega \geq \lambda_\omega$.
\vspace{-5pt}
\begin{theorem}[System Cost Guarantee in Stochastic Bottleneck]
\label{thm:stochastic-system-cost}
Suppose $(\lambda_{\omega},\mu_{\omega})$ is drawn from an arbitrary joint distribution with $\lambda_{\omega}>0$, $\mu_{\omega}>0$, $\mathbb{E}[\lambda_{\omega}]<\infty$, and the realization $\omega$ is observed before users make their mode and departure-time choices. Moreover, suppose $z_T \geq z_C$ and, for every realization $\omega$, with $\lambda_{\omega} > \mu_{\omega}$, $z_T - z_C \leq T_{C,\omega} :=\frac{\Lambda_\omega eL}{\mu_\omega (e+L)},$  where $T_{C,\omega}$ is the maximum cost incurred by users in the car-only bottleneck setting under $\omega$. 
%Suppose %, for every realization $\omega$ with $\mu_\omega<\lambda_\omega$, it holds that 
%$0\leq \Delta \leq T_{C,\omega} :=\frac{\kappa\Lambda_\omega}{\mu_\omega},$ where $T_{C,\omega}$ is the maximum cost incurred by users in the car-only bottleneck setting under realization $\omega$ when $\lambda_{\omega} > \mu_{\omega}$. 
Then, %the revenue-optimal fixed static toll satisfies 
$\mathbb{E}\!\left[\mathrm{SC}_\omega(\tau_s^{\mathrm{fix}})\right] \leq 2\,\mathbb{E}\!\left[\mathrm{SC}_\omega^*\right].$ %No additional condition is required for realizations with $\mu_\omega\geq\lambda_\omega$.
\end{theorem}
\vspace{-6pt}

For a proof, see Appendix~\ref{apdx:sc-stochastic-extension}. Together, Theorems~\ref{thm:stochastic-extension} and~\ref{thm:stochastic-system-cost} show that the guarantees obtained in the deterministic model extend to the setting with stochastic demand and capacity, even under arbitrary joint distributions that allow these two quantities to be correlated.

\vspace{-3pt}

\section{Revenue-Optimal Tolling in Urban Systems} \label{sec:mfd}

\vspace{-3pt}

The previous section analyzed static and dynamic revenue-optimal tolling in the bottleneck model with an outside option, where the bottleneck represents a single facility (e.g., a bridge) with fixed capacity. We now extend this analysis to city-scale urban traffic systems, where congestion is system-wide rather than localized at a single facility. In such settings, the outflow capacity (or service rate) depends on vehicle accumulation in the system and decreases once it exceeds a critical threshold, e.g., due to queue spill back across intersections. We model these system-wide dynamics using the Macroscopic Fundamental Diagram (MFD), which provides a parsimonious yet empirically grounded representation of aggregate urban traffic behavior~\citep{DAGANZO200749}. While the MFD is commonly expressed as a relationship between average system-wide vehicle density and aggregate flow, when the average trip distance $D$ is constant across users, it can equivalently describe a state-dependent service rate as a function of total vehicle accumulation~\citep{GONZALES20121519}. Throughout this section, we adopt the standard accumulation-based (quasi-static) single-reservoir MFD representation~\citep{MARIOTTE2017245,MARIOTTE2019327}, under which vehicle accumulation evolves according to flow conservation and travel time is approximated using the contemporaneous average system speed.

%We now extend this analysis to urban systems, where the outflow capacity (or service rate) depends on the number of vehicles in the system and decreases once vehicle accumulation exceeds a critical threshold, due to queue spill back across intersections. We model these system-wide congestion dynamics using the Macroscopic Fundamental Diagram (MFD), which provides a parsimonious yet empirically grounded characterization of aggregate traffic behavior in urban systems~\cite{DAGANZO200749}. While the MFD is commonly expressed as a relationship between average system-wide vehicle density and aggregate flow, when the average trip distance $D$ is constant across users, it can be interpreted as describing a state-dependent service rate as a function of total vehicle accumulation~\cite{GONZALES20121519}. 

We focus on a well-established class of MFDs satisfying two standard properties: \vspace{-2pt}
\begin{itemize}
    \item \textbf{Property 1:} There exists a critical vehicle accumulation level $n_c$ beyond which additional vehicles reduce total outflow due to spill-back and gridlock effects.
    \item \textbf{Property 2:} For vehicle accumulations up to $n_c$, the throughput is linear, corresponding to free-flow conditions, achieving its maximum throughput $\mu_f$ at $n = n_c$.
\end{itemize}
Figure~\ref{fig:triangular-mfd} illustrates the canonical triangular MFD satisfying these properties~\citep{cao2016effects,ameli2025dynamic,MOSHAHEDI2022103911}. For vehicle accumulations below the threshold $n_c$, the system operates at free-flow with no delays. Once accumulation exceeds $n_c$, congestion sets in, reducing the effective system capacity, which drops to zero at the jam accumulation $n_j$. In Figure~\ref{fig:triangular-mfd} and the remainder of this work, we refer to the two labeled regimes as \emph{uncongested} and \emph{congested} for clarity, with the latter often termed \emph{hypercongestion} in the literature. This triangular MFD therefore abstracts from an intermediate regime in which speeds may decline before throughput falls for tractability. Conceptually, Vickrey’s bottleneck model can be represented by an MFD-like accumulation–throughput relation with a capacity plateau rather than a decline in throughput in the congested regime (see Appendix~\ref{apdx:conceptual-bottleneck-mfd}).

For urban systems governed by MFDs satisfying the above properties, Section~\ref{subsec:dynamic-tolls-mfd} first characterizes the dynamic revenue-optimal policy and shows that it maintains the system at the throughput-maximizing capacity $\mu_f$ achieved at $n_c$. Section~\ref{subsec:static-tolls-mfd} then derives a closed-form expression for the revenue generated by any static toll for triangular MFDs. Finally, while static revenue-optimal tolls generally achieve weaker guarantees than in the bottleneck setting (as the system service rate is state-dependent), in Section~\ref{subsec:static-v-dynamic-triangular-mfd}, we show that, across a broad and practically relevant range of parameters, static revenue-optimal tolls attain at least two-thirds of the dynamic optimal revenue while incurring at most twice the minimum attainable system cost.

%\vspace{-13pt}

\begin{figure}[tbh!]
      \centering
      \includegraphics[width=70mm]{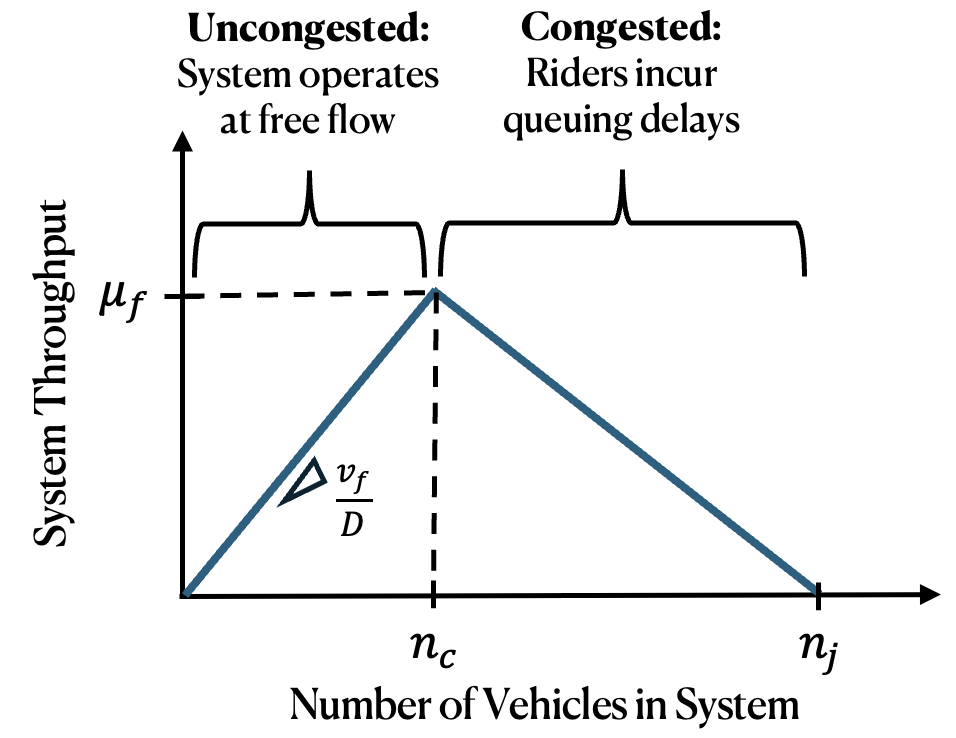}
      \vspace{-9pt}
      \caption{\small \sf Triangular MFD relating the system throughput to total vehicle accumulation. The slope of the segment connecting $(0,0)$ to any point on the curve determines the average system speed, equal to average trip distance $D$ times the slope. The MFD consists of an uncongested regime in which the system operates at free-flow with a speed $v_f$ up to the threshold $n_c$. Beyond $n_c$, the system enters a congested regime in which speeds decline, eventually dropping to zero at $n_j$.
      }
      \label{fig:triangular-mfd} 
   \end{figure} 

%\vspace{-20pt}

\vspace{-2pt}
\subsection{Revenue-Optimal Dynamic Tolling} \label{subsec:dynamic-tolls-mfd}
\vspace{-3pt}

We characterize the revenue-optimal dynamic tolling policy for urban systems governed by MFDs satisfying Properties 1 and 2, such as the triangular MFD in Figure~\ref{fig:triangular-mfd}, and show that it maintains the system at the critical vehicle accumulation level that maximizes system throughput. In proving this result, akin to Section~\ref{sec:revenue-optimal-tolling-classical-bottleneck}, we focus on the setting when the maximum achievable system throughput $\mu_f$ is strictly below the desired arrival rate $\lambda$, i.e., $\mu_f < \lambda$. We also maintain $z_T > z_C$, such that car travel at free flow is preferred to transit. For a depiction of the arrival and departure time profiles of users in the urban system under the revenue-maximizing policy, see Figure~\ref{fig:arrival_profile_dynamic_opt_combined} in Appendix~\ref{apdx:arrival-departure-opt}.

\vspace{-3pt}
\begin{theorem}[Revenue-Optimal Dynamic Tolls under MFD] \label{thm:rev-opt-dynamic-tolls-mfd}
Suppose $\mu_f < \lambda$ and users choose between two modes, traveling by car through an urban system characterized by an MFD satisfying properties 1 and 2, and using public transit with a fixed cost $z_T$, where $z_T > z_C$. %Furthermore, suppose that the cost of using public transit is at least the free-flow cost of using a car, i.e., $z_T \geq z_C$. 
Then, under the revenue-optimal dynamic tolling policy, the system operates at the critical vehicle accumulation level $n_c$ corresponding to a throughput-maximizing capacity $\mu_f$.
\end{theorem}
\vspace{-7pt}

%\proof{Proof (Sketch).}
\begin{hproof}
First, as with the bottleneck model, a necessary equilibrium condition is that $w'(t) + \tau'(t) = e$ for early arrivals and $w'(t) + \tau'(t) = -L$ for late arrivals, and that $w(t) + \tau(t) \leq z_T - z_C$ for all $t$, i.e., the sum of the waiting and toll costs is as depicted in Figure~\ref{fig:eq-waiting-tolls}. Here, the times $t_A', t_B', t_C', t_D'$ shown in the figure are determined endogenously by the tolling policy and the waiting time depends on the number of users $n(t)$ in the system at time $t$. Under this equilibrium condition, we compare two policies: (i) a candidate revenue-optimal policy $\tau^*(\cdot)$ with associated times $t_A^*, t_B^*, t_C^*, t_D^*$ in Figure~\ref{fig:eq-waiting-tolls}, and (ii) a policy $\Tilde{\tau}(\cdot)$, with associated times $\Tilde{t}_A, \Tilde{t}_B, \Tilde{t}_C, \Tilde{t}_D$, under which the system operates at the capacity $\mu_f$ with no waiting delays, where $\Tilde{t}_B = t_B^*$ and $\Tilde{t}_C = t_C^*$, and $\Tilde{\tau}(t) \! = \! z_T - z_C$ over $[t_B^*, t_C^*]$. Finally, using Properties 1 and 2, we show $R(\Tilde{\tau}(\cdot)) \! \geq \! R(\tau^*(\cdot))$. %\Halmos
%\endproof% to prove our result.
\end{hproof}
\vspace{-3pt}
For a complete proof, see Appendix~\ref{apdx:pf-mfd-dynamic-opt}.
The challenge in establishing this result and, in particular, proving $R(\Tilde{\tau}(\cdot)) \geq R(\tau^*(\cdot))$ lies in handling a key trade-off. Specifically, we compare the revenue of a candidate optimal policy $\tau^*(\cdot)$ that operates over a longer interval $[t_A^*, t_D^*]$ over which both throughput and tolls may be lower, with that of an alternative policy $\Tilde{\tau}(\cdot)$ that maintains the throughput-maximizing level $\mu_f$ by charging (weakly) higher tolls over a shorter interval $[\Tilde{t}_A,\Tilde{t}_D]$. A priori, it is not obvious which effect dominates: higher tolls over a shorter interval or lower tolls over a longer interval. Our analysis shows that despite this apparent trade-off between toll intensity and duration, revenue is maximized by maintaining the system at the throughput-maximizing operating point. We note that this trade-off does not arise in the bottleneck setting, where throughput is fixed at the bottleneck capacity; in contrast, under the MFD, throughput is endogenous and depends on vehicle accumulation in the system, making the analysis fundamentally more involved.

Theorem~\ref{thm:rev-opt-dynamic-tolls-mfd} establishes that the dynamic revenue-optimal toll eliminates queuing delays and maintains the system at the throughput-maximizing capacity $\mu_f$, a desirable property in real-world traffic systems. This result follows as operating away from $\mu_f$ either reduces throughput or generates waiting costs that could instead be extracted as toll revenue. Thus, revenue is maximized by operating at the throughput-maximizing capacity while eliminating delays. Consequently, by Theorems~\ref{thm:rev-opt-dynamic-tolls} and~\ref{thm:rev-opt-dynamic-tolls-mfd}, the dynamic revenue-optimal policy is akin to that depicted in Figure~\ref{fig:rev-opt-dynamic-tolls} (as it sets waiting times point-wise to zero) and the resulting revenue and system cost expressions coincide with those in the bottleneck setting in Section~\ref{sec:revenue-optimal-tolling-classical-bottleneck}, with the bottleneck capacity $\mu$ replaced by $\mu_f$. %Moreover, since the waiting time for all users is zero under the revenue-optimal dynamic toll, the equilibrium arrival and departure (crossing) time profiles for all users coincide.

Since the dynamic revenue-optimal toll eliminates queuing delays and sustains an equilibrium at the MFD's throughput-maximizing point, it closely resembles the system-cost-optimal policy in~\citet{Geroliminis2009}, which also maintains the system at the critical accumulation $n_c$. Nevertheless, the two policies differ in important respects, for the same reasons noted in Section~\ref{subsec:rev-opt-dynamic}.

More broadly, Theorem~\ref{thm:rev-opt-dynamic-tolls-mfd} can be interpreted through the lens of standard revenue-maximization trade-offs. In many settings, tolls or prices directly affect the number of users served, and the optimal revenue is achieved by balancing higher payments per user against reduced demand. In the setting studied here, Theorem~\ref{thm:rev-opt-dynamic-tolls-mfd} shows that this balance is achieved at the boundary of the feasible operating region, with the revenue-optimal policy maintaining the system at the critical accumulation $n_c$, rather than at higher accumulation levels on the congested branch of the MFD.

\vspace{-5pt}
\subsection{Static Tolling under Triangular MFD} \label{subsec:static-tolls-mfd}
\vspace{-2pt}

This section characterizes the revenue under static tolls for urban systems governed by a triangular MFD (see Figure~\ref{fig:triangular-mfd}). Unlike the dynamic tolling analysis in the previous section, which applies to general MFDs satisfying Properties 1 and 2, static tolling can induce operation on the congested branch of the MFD and thus requires additional structure for analytical tractability. Thus, focusing on triangular MFDs, we derive an expression for the revenue as a function of any static toll $\tau$.

\vspace{-3pt}
\begin{theorem}[Revenue of Static Tolls in MFD Framework] \label{thm:rev-static-tolls-mfd}
    Suppose $\mu_f < \lambda$ and users choose between traveling by car with a free-flow cost $z_C$ through an urban system governed by a triangular MFD (see Figure~\ref{fig:triangular-mfd}), and using public transit with a fixed cost $z_T$, where $z_T > z_C$ and all trips are of a fixed distance $D$. Moreover, let $\underline{\tau} \geq 0$ be the minimum toll at which some user is indifferent between car and transit. Then, the revenue under any static toll $\tau \in [\underline{\tau}, z_T - z_C]$ is: 
    {\setlength{\abovedisplayskip}{2pt}
\setlength{\belowdisplayskip}{2pt}
    \begin{align} \label{eq:rev-rlation-static-mfd}
        R(\tau) = \tau \left[ \frac{\Lambda}{\lambda} \frac{n_j}{\frac{n_j}{\mu_f} + z_T - z_C - \tau} + n_j \frac{e+L}{eL} \ln \left( 1 + \frac{(z_T - z_C - \tau) \mu_f}{n_j} \right) \left( 1 - \frac{\frac{n_j}{\frac{n_j}{\mu_f} + z_T - z_C - \tau}}{\lambda} \right) \right].
    \end{align}}
\end{theorem}

In the above expression, the revenue equals the static toll $\tau$ multiplied by the number of users who travel by car at equilibrium. The first component in brackets reflects the car throughput during the middle of the rush, while the logarithmic component captures the variation in throughput as congestion builds and dissipates during the early and late portions of the rush.

The proof of this result is significantly more involved than the characterization of static revenue-optimal tolls in the bottleneck setting because, under a static toll, the system may operate on the congested branch of the MFD, making throughput endogenous and time-varying through its dependence on vehicle accumulation. For a proof of Theorem~\ref{thm:rev-static-tolls-mfd}, see Appendix~\ref{apdx:mfd-statc-revenue-exp}. While Theorem~\ref{thm:rev-static-tolls-mfd} yields the revenue generated by a static toll $\tau$, its complexity makes a closed-form characterization of the revenue-optimal static toll challenging, unlike the bottleneck setting in Section~\ref{subsec:rev-opt-static}. This difficulty arises as static tolls can induce operation on the congested branch of the MFD. That said, since the revenue $R(\tau)=\tau\Lambda$ is increasing in the toll for $\tau<\underline{\tau}$, and the revenue in Equation~\eqref{eq:rev-rlation-static-mfd} is Lipschitz continuous over $\tau \in [\underline{\tau}, z_T - z_C]$, an approximately optimal toll can be computed by discretizing this interval and selecting the toll with the highest revenue from this set.

%That said, since the revenue function in Equation~\eqref{eq:rev-rlation-static-mfd} is Lipschitz continuous over $\tau \in [0, z_T - z_C]$, an approximately optimal toll can be computed by discretizing this interval and selecting the toll with the highest revenue from this set.

\vspace{-3pt}
\subsection{Static vs. Dynamic Tolling under Triangular MFD} \label{subsec:static-v-dynamic-triangular-mfd}
\vspace{-3pt}

Due to the state-dependent MFD capacities, static tolls admit weaker guarantees than in the bottleneck model. While static revenue-optimal tolls achieve at least half the dynamic optimal revenue in the bottleneck model (Theorem~\ref{thm:rev-comp-static-v-dynamic}), our numerical experiments in Section~\ref{sec:experiments} show that this revenue ratio may drop below half under a triangular MFD. Nevertheless, for such systems, when $z_T - z_C \leq \min \{\frac{\Lambda eL}{(\lambda - \mu_f)(e+L)}, \frac{\Lambda eL}{\mu_f(e+L)} \}$, an empirically relevant regime in which transit provides a viable alternative to cars, the static toll $\tau=z_T-z_C$ maintains the system at the throughput-maximizing capacity $\mu_f$ with no delays. In this case, the revenue and system cost then reduce to their bottleneck counterparts with $\mu$ replaced by $\mu_f$, implying that this toll achieves at least a $\frac{2}{3-\frac{\mu_f}{\lambda}}$ fraction of the optimal dynamic revenue and incurs at most twice the minimum achievable system cost. While these guarantees do not extend when $z_T - z_C > \min \{\frac{\Lambda eL}{(\lambda - \mu_f)(e+L)}, \frac{\Lambda eL}{\mu_f(e+L)} \}$, they show that static tolling can achieve robust performance across both revenue and system cost metrics over a broad and practically relevant range of parameters, even for urban systems governed by a triangular MFD. That said, due to the complexity of the revenue expression in Equation~\eqref{eq:rev-rlation-static-mfd}, a general characterization of the gap between static tolls and the corresponding dynamic optimal benchmarks is more challenging than in the bottleneck model. We thus provide a more comprehensive assessment of this gap through numerical experiments. % in the next section.

\vspace{-3pt}
\section{Numerical Experiments} \label{sec:experiments}
\vspace{-3pt}
This section presents numerical experiments comparing the efficacy of static revenue-optimal tolling against its dynamic revenue and system-cost-optimal counterparts in two real-world congestion pricing settings: (i) the SF–Oakland Bay Bridge for the bottleneck model and (ii) New York City’s congestion relief zone (CRZ) for the MFD framework. Our results, which are grounded in practical datasets, provide empirically grounded illustrations of our theoretical guarantees and demonstrate that simple static tolls can deliver robust performance on both revenue and system cost, underscoring their practical efficacy and shedding light on the strong empirical performance of the many static congestion pricing systems already in operation. In the following, we provide an overview of our experimental setup (Section~\ref{subsec:model-calibration}) and compare the relative efficacy of static and dynamic tolling in both settings (Section~\ref{subsec:numerical-results}). For details on our data sources and model calibration, see Appendix~\ref{apdx:additional-numerical-details}. %The code to generate the results are available at the following \href{{https://github.com/djalota/static_optimal_tolling}}{Link}.

\vspace{-3pt}
\subsection{Overview of Experimental Setup} \label{subsec:model-calibration}
\vspace{-4pt}
%For both the SF-Oakland Bay Bridge and New York City (NYC) case studies, we focus on the morning commuting period between 5-10 AM, when congestion is the most pronounced, and users choose between traveling by car or using a local public transit alterative, which we model as the outside option. In both settings, in practice, car travel is subject to a \emph{flat, time-invariant} congestion toll during this period, with a \$8.50 toll for westbound commuting trips on the Bay Bridge and a \$9 CRZ toll in NYC, making these settings natural empirical benchmarks for evaluating the performance of static tolling relative to its dynamic pricing benchmarks.

For both the SF-Oakland Bay Bridge and New York City (NYC) case studies, we focus on the weekday morning rush between 5-10 AM, when congestion is most pronounced and most cars are subject to a \emph{flat, time-invariant} toll of \$8.50 for westbound trips on the Bay Bridge and \$9 for entry into NYC’s congestion relief zone (CRZ). The use of simple static tolls in both settings makes them natural empirical settings for evaluating static tolling relative to its dynamic benchmarks. In each case, users choose between traveling by car or a local public transit alternative, Bay Area Rapid Transit (BART) in the Bay Bridge corridor or the NYC subway, modeled as the outside option.
To calibrate the bottleneck model for the Bay Bridge study and the MFD for the NYC study, we draw on several data sources, including BART and NYC subway ridership~\citep{BART_RidershipReports,MTA_OD_Ridership2025}, vehicular demand and flow (such as Bay Bridge traffic counts and NYC CRZ entries)~\citep{caltrans_pems,NY_CongestionRelief_Entries2025}, taxi and for-hire vehicle trip records~\citep{NYC_TLC_TripData}, and OpenStreetMap data~\citep{OpenStreetMap} to calibrate the MFD. We combine these data sources with value of time and schedule delay penalty estimates from the literature. %, as well as information on parking costs in both locations.

A key quantity we vary in our experiments is the cost difference between transit and car travel at free-flow, $z_T - z_C$, which plays a central role in our theoretical results. To do so, while we fix $z_C$ to reflect the costs of parking plus the average free-flow travel time for morning commutes in each setting, we vary $z_T$ through a \emph{discomfort multiplier} $\eta$. Specifically, we model $z_T$ as the sum of the transit fare and a time-based cost of walking to and from stations, waiting time on platforms, and in-vehicle BART or subway travel time. These time components are scaled by a parameter $\eta$, which we use as a reduced-form measure of the relative generalized cost of transit time compared with driving time. Since the discomfort multiplier can vary across settings, we conduct a sensitivity analysis over $\eta \in [1.5, 18]$, with values in the range $[1.5, 5]$ consistent with empirical estimates~\citep{Wardman2012}. In our calibration, substantially smaller values of $\eta$ would make transit comparatively more attractive than driving, so we focus on $\eta \geq 1.5$. Moreover, exploring a broader range of $\eta$ beyond empirical estimates allows us to capture heterogeneity in user preferences and assess the impact of technological or infrastructural changes, such as improvements in transit quality or increased attractiveness of car travel due to autonomous vehicles~\citep{ostrovsky2025congestion}. Most importantly, varying $\eta$ provides a natural proxy for spanning the full range of $z_T - z_C$ characterized in our theoretical results and evaluating static pricing across the corresponding regimes.

For further details on data sources and model calibration, see Appendix~\ref{apdx:model-calibration}. Parameter values for the Bay Bridge and NYC studies are summarized in Tables~\ref{tab:calibration_summary-bottleneck} (Appendix~\ref{apdx:bay_area_summary}) and~\ref{tab:calibration_summary_nyc_mfd} (Appendix~\ref{apdx:summary-values-nyc}). %, respectively.

\vspace{-3pt}
\subsection{Results} \label{subsec:numerical-results}
\vspace{-3pt}
This section compares static and dynamic tolls on revenue and system cost for the Bay Bridge and NYC studies. Figures~\ref{fig:bottleneck-ratios-bay-area} and~\ref{fig:mfd-ratios-nyc} plot the revenue and system cost ratios for a set of static and dynamic tolling policies, normalized by the corresponding dynamic revenue-optimal and system-cost-optimal benchmarks, as the discomfort multiplier $\eta$ is varied for the Bay Bridge and NYC settings, respectively. We denote the static revenue-optimal and system-cost-optimal policies by \emph{Static-RO} and \emph{Static-SO}, and their dynamic counterparts by \emph{Dynamic-RO} and \emph{Dynamic-SO}. The static and dynamic revenue-optimal policies are computed using the results in Sections~\ref{sec:revenue-optimal-tolling-classical-bottleneck} and~\ref{sec:mfd}. The dynamic system-cost-optimal policy is computed using established results from the literature~\citep{GONZALES20121519,Geroliminis2009}, while the static system-cost-optimal policies are computed based on the derivations in Appendices~\ref{apdx:sco-static-toll} and~\ref{apdx:derivation-tsc-triangular-mfd}.

Our results show that for practically relevant $\eta \in [1.5, 5]$, static revenue-optimal tolls incur at most a 10\% revenue loss in the Bay Bridge study and 20\% in the NYC study compared to the dynamic optimal revenue. As transit becomes less attractive ($\eta >> 5$), this gap widens to around 15\% in the Bay Bridge study and can exceed 50\% in NYC, consistent with our theory. While static revenue-optimal tolling achieves better revenue ratios in the Bay Bridge study, it performs worse on system cost. Specifically, in the NYC study, static revenue-optimal tolls incur at most an 8\% increase relative to the minimum achievable system cost, but, in the Bay Bridge study, the system cost increase reaches 25\% for $\eta \in [1.5, 5]$. This superior performance of static tolling on system cost in the NYC study is driven by how attractive transit remains as $\eta$ increases, which limits congestion costs under static tolling. In the Bay Bridge study, transit becomes less attractive more quickly, so static tolls induce larger system cost losses relative to the dynamic optimum. Further, when $\eta \gg 5$, the system cost ratio in the Bay Bridge study can even exceed two, i.e., the static policy incurs more than twice the optimal cost. Thus, static tolls in NYC sacrifice more revenue but remain closer to the system cost optimum, highlighting a trade-off between the revenue and system cost under static tolls.

Figure~\ref{fig:bottleneck-ratios-bay-area} also highlights a non-monotonic relationship between $\eta$ and the revenue gap between static and dynamic policies for the Bay Bridge study. Static revenue-optimal tolls perform well when transit is only slightly less desirable than driving at low $\eta$, with performance deteriorating at intermediate values of $\eta$, and improving again when transit becomes highly unattractive. Intuitively, as the cost difference $z_T - z_C$ initially increases with an increase in $\eta$, static tolling faces an increasing price-participation trade-off: raising the static toll extracts more from each car user but diverts more users to transit, whereas a dynamic toll can charge high tolls at the peak while using lower tolls at the shoulders to retain more car users, so its revenue advantage widens. However, at sufficiently large $z_T - z_C$, both policies keep essentially all users on the road, so further increases in $z_T - z_C$ (through increases in $\eta$) mainly increase the toll each policy can charge rather than changing who travels by car; hence, the advantage of time-varying tolls diminishes. This non-monotonic behavior is consistent with our theoretical bounds on the revenue ratio between static and dynamic revenue-optimal policies in Theorem~\ref{thm:rev-comp-static-v-dynamic}. Moreover, it contrasts the strictly monotonic pattern observed in the NYC study in Figure~\ref{fig:mfd-ratios-nyc}, where the static-to-dynamic revenue gap increases monotonically with $\eta$.

%performance gap of static revenue optimal tolls relative to its dynamic benchmarks worsens as $\eta$ increases.

Figure~\ref{fig:bottleneck-ratios-bay-area} further shows that the two static policies (Static-RO and Static-SO) exhibit nearly identical performance across most values of the discomfort multiplier $\eta$. Notably, in the practically relevant range $\eta \in [1.5, 5]$, the two policies coincide. In this regime, transit provides a strong outside option, so both objectives favor the highest feasible toll, $\tau = z_T - z_C$, which raises revenue while eliminating costly queuing. The policies diverge at intermediate values of $\eta$, when the revenue-optimal toll balances a higher toll against diversion to transit, while the system-cost-optimal policy may favor a lower toll to avoid sending users to an increasingly costly transit alternative. At sufficiently large values of $\eta$, however, transit becomes so unattractive that both static tolls induce the same all-car equilibrium, causing the two objectives to align. A similar pattern holds for the NYC study under the MFD framework, where the static system-cost-optimal and revenue-optimal tolls coincide for all $\eta$ shown in Figure~\ref{fig:mfd-ratios-nyc}; hence, we plot only the latter. Overall, these results suggest that if a central planner is required to set static tolls, as is common in practice, revenue and system cost objectives can be achieved simultaneously particularly over practically relevant ranges of $\eta$.

Beyond $\eta$, another parameter shaping performance is the ratio of the maximum system throughput to the desired arrival rate ($\frac{\mu}{\lambda}$ in the bottleneck model and $\frac{\mu_f}{\lambda}$ in the MFD). In NYC, where this calibrated ratio is small, around $0.2$, aligning with its high transit mode share, the relative car and transit costs over practically relevant $\eta$ result in a large share of users choosing transit under either objective, leaving limited scope to reallocate road users over time. In this setting, the dynamic revenue- and system-cost-optimal tolls nearly coincide, while static tolls achieve lower revenues and higher system costs than both dynamic benchmarks. Despite this Pareto dominance of dynamic tolling in the transit-intensive NYC study, the performance losses of static tolls remain small for most practical $\eta \in [1.5, 5]$. In contrast, when the ratio of the maximum throughput to the desired arrival rate is closer to one, as in the Bay Bridge morning peak where it is roughly $0.7$, more users can be accommodated by car and the revenue-optimal policy has more scope to adjust tolls to increase revenue, causing the two dynamic policies to differ considerably. In this case, while dynamic revenue-optimal tolls Pareto dominate both static policies for all $\eta$, dynamic system-cost-optimal tolls can generate up to 20\% less revenue than both static policies for practical $\eta$.

Finally, the static and dynamic revenue-optimal tolls in the experiments correspond to the current Bay Bridge and CRZ tolls within the practically relevant range of $\eta$. Notably, the static revenue-optimal toll equals the \$8.50 Bay Bridge toll at $\eta=1.76$ and the \$9 CRZ toll at $\eta=1.83$. We emphasize that $\eta$ is varied independently over the full sensitivity range and these values simply identify where the static revenue-optimal toll under our calibrated model coincides with current practice. %The lower value of $\eta$ in NYC is consistent with public transit being relatively more attractive, although we note that $\eta$ is a reduced-form parameter that may also reflect broader differences in fares, access times, parking costs, and demand. 
Conditional on these values of $\eta$, static tolls capture nearly all of the benefits of dynamic tolling, achieving roughly 97\% of dynamic optimal revenue and increasing total system cost by less than 1\% relative to the minimum achievable system cost in both settings.

%reflects its robust public transit system, which contrasts the dominance of car travel in the Bay Bridge corridor. %; importantly, the broader performance patterns reported above hold across the practically relevant range of $\eta$.

In Appendix~\ref{apdx:heterogeneous-static-v-dynamic}, we further show that the above results are robust to heterogeneity in users' values of time and schedule delay penalties in the bottleneck model. Overall, our results highlight the merits of static revenue-optimal tolling relative to its dynamic counterparts. % on both revenue and total system cost.

\vspace{-10pt}
\begin{figure*}[tbh!]
  \centering \hspace{-20pt}
  \begin{subfigure}[b]{0.4\columnwidth}
      \include{Fig/tex/wine_tex/bay_area_revenue_ratio}
  \end{subfigure} \hspace{30pt}
  \begin{subfigure}[b]{0.4\columnwidth}
      \include{Fig/tex/bay_area_sc_ratio}
  \end{subfigure} %\hspace{5pt}
     \vspace{-35pt}
    \caption{{\small \sf Depiction of the revenue (left) and system cost (right) ratios of a set of static and dynamic tolling policies, normalized by the corresponding dynamic revenue-optimal and system-cost-optimal benchmarks, as the discomfort multiplier $\eta$ is varied for the Bay Bridge case study. %The left panel depicts the revenue ratio of the (i) static revenue-optimal, (ii) static system-cost optimal, and (iii) dynamic system-cost optimal tolling policies, normalized by the maximum achievable dynamic tolling revenue. The right panel depicts the system cost ratio of the (i) static revenue-optimal, (ii) static system-cost optimal, and (iii) dynamic revenue-optimal tolling policies, normalized by the minimum achievable dynamic tolling system cost.
    %Variation of the revenue ratio of the (i) static revenue-optimal, (ii) static system-cost optimal, and (iii) dynamic system-cost optimal tolling policies, normalized by the maximum achievable dynamic tolling revenue, as a function of the discomfort multiplier $\eta$ for the Bay Area (left) and NYC (right) case studies.
    %revenue (left) and system cost (right) ratios of a set of static and dynamic tolling policies, normalized by the corresponding dynamic revenue-optimal and system-cost-optimal benchmarks, as the discomfort multiplier $\eta$ is varied for the Bay Area case study. In particular, we depict the revenue ratio corresponding to the static 
    }} 
    \label{fig:bottleneck-ratios-bay-area}
\end{figure*}
\vspace{-20pt}
\begin{figure*}[tbh!]
  \centering \hspace{-20pt}
  \begin{subfigure}[b]{0.4\columnwidth}
      \definecolor{mycolor1}{rgb}{0.00000,0.44700,0.74100}%
\definecolor{mycolor2}{rgb}{0.85000,0.32500,0.09800}%
\definecolor{mycolor3}{rgb}{0.46600,0.67400,0.18800}%
\definecolor{mycolor4}{rgb}{0.49400,0.18400,0.55600}%

\begin{tikzpicture}
\begin{axis}[
    width=2.5in,
    height=0.9in,
    at={(0in,0in)},
    scale only axis,
    xmin=1.5,
    xmax=18,
    ymin=0.25,
    ymax=1.05,
    xlabel={Discomfort Multiplier ($\eta$)},
    ylabel={Revenue Ratio},
    xlabel style={
        font=\color{white!15!black},
        font=\footnotesize,
        yshift=0.1cm
    },
    ylabel style={
        font=\color{white!15!black},
        font=\footnotesize
    },
    tick label style={
        font=\scriptsize,
        yshift=0.07cm
    },
    axis background/.style={fill=white},
    grid=both,
    grid style={dashed, gray!30},
    legend style={
        at={(0.65,0.21)},      % adjust if needed after adding shaded region
        anchor=east,
        legend cell align=left,
        align=left,
        draw=white!0!black,
        inner xsep=1pt,
        inner ysep=0pt,
        font=\tiny,
        row sep=-3pt
    }
]

% --- Practical range shading + dashed boundaries ---
% Shade the vertical band from x=1.5 to x=3
\addplot[
    draw=none,
    fill=gray,
    fill opacity=0.20,
    legend image code/.code={
        \draw[fill=gray, fill opacity=0.20, draw=none]
            (0cm,-0.1cm) rectangle (0.6cm,0.1cm);
    }
] coordinates {(1.5,0) (5,0) (5,2.1) (1.5,2.1)} \closedcycle;

% Legend entry for the shaded band
\addlegendentry{Practical Range ($\eta \in [1.5, 5]$)}

% % Draw dashed vertical boundary lines at x=1.5 and x=3
% \addplot[
%     black,
%     dashed,
%     thick,
%     forget plot
% ] coordinates {(1.5,0) (1.5,2.1)};

% \addplot[
%     black,
%     dashed,
%     thick,
%     forget plot
% ] coordinates {(3,0) (3,2.1)};

% ---------- Line 1 ----------
\addplot[color = mycolor1, thick] table[row sep=\\] {
x y \\
1.5 0.99568183 \\
2.47058824 0.93555859 \\
3.44117647 0.88228283 \\
4.41176471 0.83474779 \\
5.38235294 0.79207300 \\
6.35294118 0.75354933 \\
7.32352941 0.71859917 \\
8.29411765 0.68674733 \\
9.26470588 0.65759930 \\
10.23529412 0.63082483 \\
11.20588235 0.60614533 \\
12.17647059 0.58332418 \\
13.14705882 0.56215909 \\
14.11764706 0.54247612 \\
15.08823529 0.52412485 \\
16.05882353 0.50697455 \\
17.02941176 0.49091106 \\
18.00000000 0.47583426 \\
};
\addlegendentry{Static-RO}

\begin{comment}
% ---------- Line 1 ----------
\addplot[color = mycolor2, thick, dashed] table[row sep=\\] {
x y \\
1.5 0.99568183 \\
2.47058824 0.93555859 \\
3.44117647 0.88228283 \\
4.41176471 0.83474779 \\
5.38235294 0.79207300 \\
6.35294118 0.75354933 \\
7.32352941 0.71859917 \\
8.29411765 0.68674733 \\
9.26470588 0.65759930 \\
10.23529412 0.63082483 \\
11.20588235 0.60614533 \\
12.17647059 0.58332418 \\
13.14705882 0.56215909 \\
14.11764706 0.54247612 \\
15.08823529 0.52412485 \\
16.05882353 0.50697455 \\
17.02941176 0.49091106 \\
18.00000000 0.47583426 \\
};
\addlegendentry{Static SC Opt}
\end{comment}

% ---------- Line 2 ----------
\addplot[color = mycolor3, thick, mark=none] table[row sep=\\] {
% (second y-array, same x)
x y \\
1.5 0.99952020 \\
2.47058824 0.99283984 \\
3.44117647 0.98692031 \\
4.41176471 0.98163864 \\
5.38235294 0.97689700 \\
6.35294118 0.97261659 \\
7.32352941 0.96873324 \\
8.29411765 0.96519415 \\
9.26470588 0.96195548 \\
10.23529412 0.95898054 \\
11.20588235 0.95623837 \\
12.17647059 0.95370269 \\
13.14705882 0.95135101 \\
14.11764706 0.94916401 \\
15.08823529 0.94712498 \\
16.05882353 0.94521939 \\
17.02941176 0.94343456 \\
18.00000000 0.94175936 \\
};
\addlegendentry{Dynamic-SO}

\end{axis}
\end{tikzpicture}
  \end{subfigure} \hspace{30pt}
  \begin{subfigure}[b]{0.4\columnwidth}
      \definecolor{mycolor1}{rgb}{0.00000,0.44700,0.74100}%
\definecolor{mycolor2}{rgb}{0.85000,0.32500,0.09800}%
\definecolor{mycolor3}{rgb}{0.46600,0.67400,0.18800}%
\definecolor{mycolor4}{rgb}{0.49400,0.18400,0.55600}%

\begin{tikzpicture}
\begin{axis}[
    width=2.5in,
    height=0.9in,
    at={(0in,0in)},
    scale only axis,
    xmin=1.5,
    xmax=18,
    ymin=0.9,
    ymax=1.9,
    xlabel={Discomfort Multiplier ($\eta$)},
    ylabel={System Cost Ratio},
    xlabel style={
        font=\color{white!15!black},
        font=\footnotesize,
        yshift=0.1cm
    },
    ylabel style={
        font=\color{white!15!black},
        font=\footnotesize
    },
    tick label style={
        font=\scriptsize,
        yshift=0.07cm
    },
    axis background/.style={fill=white},
    grid=both,
    grid style={dashed, gray!30},
    legend style={
        at={(0.6,0.82)},      % adjust if needed after adding shaded region
        anchor=east,
        legend cell align=left,
        align=left,
        draw=white!0!black,
        inner xsep=1pt,
        inner ysep=0pt,
        font=\tiny,
        row sep=-3pt
    }
]

% --- Practical range shading + dashed boundaries ---
% Shade the vertical band from x=1.5 to x=3

% Legend entry for the shaded band
%\addlegendentry{Practical Range ($\eta \in [1.5, 5]$)}

% % Draw dashed vertical boundary lines at x=1.5 and x=3
% \addplot[
%     black,
%     dashed,
%     thick,
%     forget plot
% ] coordinates {(1.5,0) (1.5,2.1)};

% \addplot[
%     black,
%     dashed,
%     thick,
%     forget plot
% ] coordinates {(3,0) (3,2.1)};

% ---------- Line 1 ----------
\addplot[color = mycolor1, thick] table[row sep=\\] {
x y \\
1.5 1.00005841 \\
2.47058824 1.01025849 \\
3.44117647 1.02993623 \\
4.41176471 1.05454599 \\
5.38235294 1.08262837 \\
6.35294118 1.11365068 \\
7.32352941 1.14744963 \\
8.29411765 1.18404422 \\
9.26470588 1.22356277 \\
10.23529412 1.26621278 \\
11.20588235 1.31226962 \\
12.17647059 1.36207506 \\
13.14705882 1.41604187 \\
14.11764706 1.47466324 \\
15.08823529 1.53852676 \\
16.05882353 1.60833364 \\
17.02941176 1.68492426 \\
18.00000000 1.76931204 \\
};
\addlegendentry{Static-RO}

\begin{comment}
\addplot[color = mycolor2, thick, mark=none, dashed] table[row sep=\\] {
% (second y-array, same x)
x y \\
1.5 1.00005841 \\
2.47058824 1.01025849 \\
3.44117647 1.02993623 \\
4.41176471 1.05454599 \\
5.38235294 1.08262837 \\
6.35294118 1.11365068 \\
7.32352941 1.14744963 \\
8.29411765 1.18404422 \\
9.26470588 1.22356277 \\
10.23529412 1.26621278 \\
11.20588235 1.31226962 \\
12.17647059 1.36207506 \\
13.14705882 1.41604187 \\
14.11764706 1.47466324 \\
15.08823529 1.53852676 \\
16.05882353 1.60833364 \\
17.02941176 1.68492426 \\
18.00000000 1.76931204 \\
};
\addlegendentry{Static SC Opt}
\end{comment}

% ---------- Line 2 ----------
\addplot[color = mycolor4, thick, mark=none] table[row sep=\\] {
% (second y-array, same x)
x y \\
1.5 1.00000365 \\
2.47058824 1.00064116 \\
3.44117647 1.00187101 \\
4.41176471 1.00340912 \\
5.38235294 1.00516427 \\
6.35294118 1.00710317 \\
7.32352941 1.00921560 \\
8.29411765 1.01150276 \\
9.26470588 1.01397267 \\
10.23529412 1.01663830 \\
11.20588235 1.01951685 \\
12.17647059 1.02262969 \\
13.14705882 1.02600262 \\
14.11764706 1.02966645 \\
15.08823529 1.03365792 \\
16.05882353 1.03802085 \\
17.02941176 1.04280777 \\
18.00000000 1.04808200 \\
};
\addlegendentry{Dynamic-RO}

\addplot[
    draw=none,
    fill=gray,
    fill opacity=0.20,
    % legend image code/.code={
    %     \draw[fill=gray, fill opacity=0.20, draw=none]
    %         (0cm,-0.1cm) rectangle (0.6cm,0.1cm);
    % }
] coordinates {(1.5,0) (5,0) (5,2.1) (1.5,2.1)} \closedcycle;

\end{axis}
\end{tikzpicture}
  \end{subfigure} %\hspace{5pt}
     \vspace{-35pt}
    \caption{{\small \sf Depiction of the revenue (left) and system cost (right) ratios of a set of static and dynamic tolling policies, normalized by the corresponding dynamic revenue-optimal and system-cost-optimal benchmarks, as the discomfort multiplier $\eta$ is varied for the New York City case study.
    %Variation in the system cost ratio of the (i) static revenue-optimal, (ii) static system-cost optimal, and (iii) dynamic revenue-optimal tolling policies, normalized by the minimum achievable dynamic tolling system cost, as a function of the discomfort multiplier $\eta$ for the Bay Area (left) and NYC (right) case studies.
    }} 
    \label{fig:mfd-ratios-nyc}
\end{figure*}
\vspace{-15pt}
%This regime may arise, for instance, in settings with autonomous driving, where car travel at free flow becomes substantially more attractive than public transit~\cite{ostrovsky2025congestion}.

\begin{comment}
\begin{figure}[tbh!]
      \centering
      \includegraphics[width=80mm]{Fig/experiments/bay_bridge_toll_comparison_to_true_value.png}
      \vspace{-10pt}
      \caption{\small \sf Comparison between the Bay Bridge toll of $\$8$ and the static revenue optimal tolls for varying values of the discomfort multiplier for using transit.}
      \label{fig:toll-comparison-bay-bridge} 
   \end{figure}     
\end{comment}

\vspace{-3pt}

\section{Conclusion and Future Work} \label{sec:conclusion}
\vspace{-4pt}

This work analyzed the performance gap between simple static and optimal dynamic congestion pricing in two canonical models capturing the dominant real-world use cases of congestion pricing. In both models, we analyzed revenue-optimal static and dynamic tolling policies and showed that static tolls deliver robust performance, achieving constant-factor guarantees, on both revenue and system cost for a wide range of practically relevant parameters. We further illustrated our theoretical findings using two case studies. Overall, our results demonstrate the potential practical efficacy of static tolls and help shed light on the attractiveness of static tolling as a design choice in real-world congestion pricing systems. Moreover, our work abstracts from the substantial informational and computational burdens dynamic tolling places on transportation planners and users, and accounting for such implementation frictions may further strengthen the case for static tolling.

%our results can be interpreted as worst-case guarantees. Incorporating realistic behavioral or operational constraints would strengthen the case for static tolling.

There are several future research directions. First, it would be valuable to study the performance gap between static and dynamic tolling under standard extensions of the bottleneck model, including non-uniform desired bottleneck departure distributions, stochastic costs, or alternative information structures beyond the stochastic demand and capacity setting considered in Section~\ref{sec:stochastic-extension}. It would also be worthwhile to examine tolling policies for MFDs, such as Greenshield's relation~\citep{greenshields1947potential}, that do not satisfy Properties 1 and 2 or are non-triangular for static tolling. Furthermore, it would be valuable to extend our analysis to intermediate tolling policies, such as step tolls between fully static and fully dynamic tolling, and characterize their revenue and system cost performance. Finally, there is scope to study objectives beyond revenue and system cost, and differentiated tolling schemes (e.g., carpool or taxi discounts) that better reflect real-world congestion pricing structures. 

\section*{Acknowledgments}

This work was supported by the Data Science Institute Postdoctoral Fellowship at Columbia University. We also thank Michael Ostrovsky for insightful discussions and Luyang Han for his support with the code to process the OpenStreetMap and NYC subway origin-destination data, required for our experiments in Section~\ref{sec:experiments}. 

%\clearpage
\vspace{-4pt}
\bibliographystyle{unsrtnat} % outcomment this and next line in Case 1
\bibliography{EC_26/sample_bibliography_ts} % if more than one, comma separated

\clearpage

\appendix

\section{Arrival and Departure Time Profiles} \label{apdx:arrival-departure-overall}
\vspace{-4pt}
\subsection{Equilibrium Arrival and Departure Time Profiles} \label{apdx:arrival-departure-profiles}
\vspace{-4pt}

Figure~\ref{fig:arrival_profile_bottlenck_combined} depicts the equilibrium arrival and departure time profiles in the bottleneck model without an outside option and with an outside option in a mixed-mode equilibrium where both car and transit are used in the setting when $\mu < \lambda$. %We denote the equilibrium intervals of crossing the bottleneck without an outside option by \([t_A^{(\mathrm{No})},t_D^{(\mathrm{No})}]\), and that with an outside option by \([t_A,t_D]\), since these equilibrium endpoints generally differ across the two settings. For more details on these equilibrium arrival and departure time profiles, see~\citet{GONZALES20121519}. 
The waiting time profiles depicted in Figure~\ref{fig:equilibrium_both_models} for car users are given by the horizontal separation between the corresponding cumulative arrival and departure curves. In the bottleneck model, the arrival time is the time at which users join the bottleneck queue, while the departure time is the time at which they pass through the bottleneck.
%is given by $w(t) = D(t) - A(t)$ for the bottleneck model without an outside option and is $w(t) = D_C(t) - A_C(t)$ for the bottleneck model with an outside option. 
\vspace{-10pt}
\begin{figure}[tbh!]
      \centering
      \includegraphics[width=150mm]{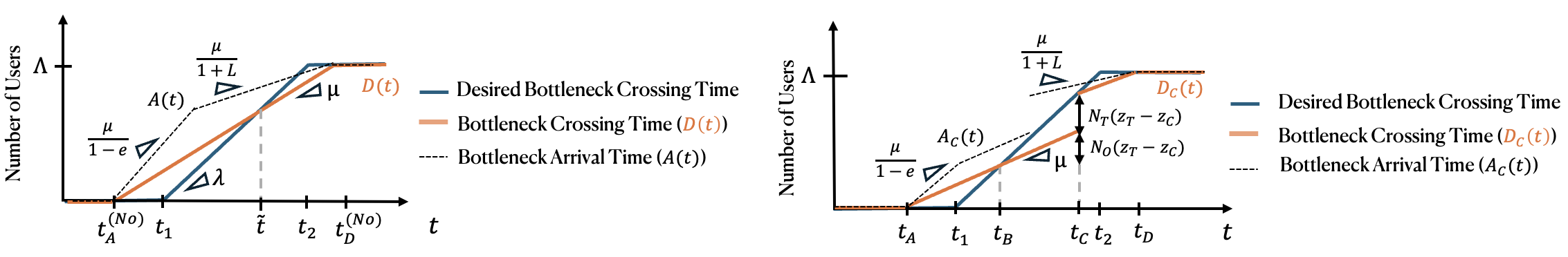}
      \vspace{-8pt}
      \caption{\footnotesize \sf Arrival and departure time profiles in the bottleneck model without (left) and with (right) an outside option when $\mu < \lambda$ and $z_T > z_C$. $A(t)$ and $A_C(t)$ are the arrival time distribution of car users at the bottleneck while $D(t)$ and $D_C(t)$ are the cumulative bottleneck departure/crossing-time curves. Moreover, in the bottleneck model with an outside option, $N_T(z_T - z_C)$ is the number of transit users and $N_O(z_T - z_C)$ is the number of on-time car users. The curve in blue is the distribution of desired bottleneck crossing times. 
      %The slopes $e$ and $L$ denote the normalized schedule-delay penalties for early and late arrivals, respectively, and $[t_1,t_2]$ represents users’ desired bottleneck crossing times. %All users passing the bottleneck at $t< \Tilde{t}$ pass the bottleneck earlier than their desired time while those passing the bottleneck at $t > \Tilde{t}$ arrive later than their desired depart
      }
      \label{fig:arrival_profile_bottlenck_combined} 
   \end{figure}

\vspace{-26pt}

\subsection{Arrival and Departure Time Profiles under Dynamic Revenue Optimal Toll} \label{apdx:arrival-departure-opt}
\vspace{-4pt}

Figure~\ref{fig:arrival_profile_dynamic_opt_combined} depicts the equilibrium arrival and departure profiles under the dynamic revenue-optimal toll in the bottleneck model (left) and the MFD framework (right). In the bottleneck model, when $\mu < \lambda$ and $z_T > z_C$, the toll in Figure~\ref{fig:rev-opt-dynamic-tolls} eliminates waiting delays, so the arrival and departure profiles coincide, i.e., $D_C(t)-A_C(t)=0$ for all $t$. In the MFD model, when $\mu_f<\lambda$ and $z_T\geq z_C$, queuing is likewise eliminated, but the profiles are offset by the free-flow travel time $t_f=D/v_f$, corresponding to travel over the fixed distance $D$. The difference arises because the bottleneck model has no notion of distance. If it is instead modeled as an uncongested segment of length $D$ followed by a point queue (see Figure~\ref{fig:conceptual-bottleneck-mfd}), its arrival and departure profiles are similar to that in the MFD framework.

\vspace{-10pt}
\begin{figure}[tbh!]
      \centering
      \includegraphics[width=150mm]{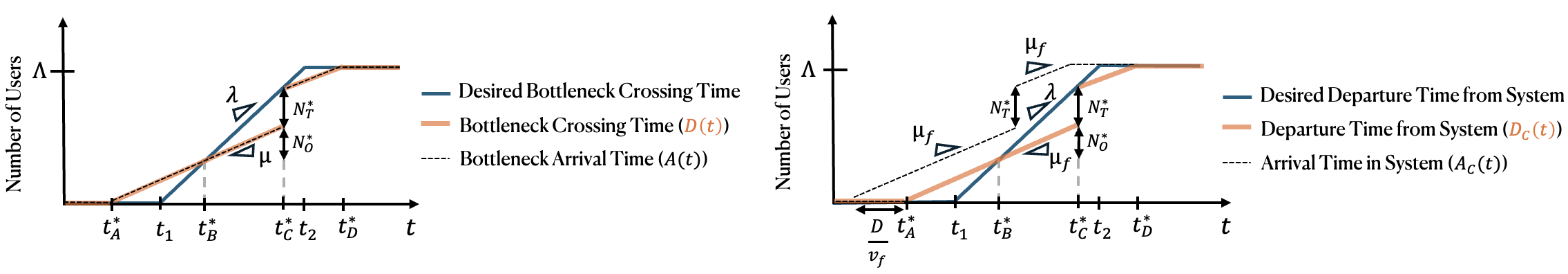}
      \vspace{-10pt}
      \caption{\footnotesize \sf Arrival and departure time profiles for car users under the dynamic revenue-optimal toll in the bottleneck (left) and MFD (right) models with an outside option when $\mu < \lambda$ and $z_T > z_C$. In both models, queuing is eliminated under the dynamic revenue-optimal toll. Hence, both arrival and departure time profiles for car users coincide, i.e., $D_C(t) = A_C(t)$, for the bottleneck model and are offset by the free-flow travel time $t_f=D/v_f$ for the MFD. Here, $N_T^*$ is the number of transit users and $N_O^*$ is the number of on-time car users.
      %Since waiting or queuing delays are eliminated under the dynamic revenue-optimal toll, both arrival and departure time profiles for car users coincide, i.e., $D_C(t) = A_C(t)$. Moreover, $N_T^*$ is the number of transit users and $N_O^*$ is the number of on-time car users.
      }
      \label{fig:arrival_profile_dynamic_opt_combined} 
   \end{figure}  

\vspace{-25pt}

\vspace{-3pt}
\section{Conceptual Relationship Between Bottleneck and MFD Models} \label{apdx:conceptual-bottleneck-mfd}
\vspace{-3pt}

Figure~\ref{fig:conceptual-bottleneck-mfd} describes a high-level conceptual connection between the bottleneck and MFD models. Specifically, it depicts a bottleneck system as an uncongestible physical section followed by a downstream point queue, and shows the resulting MFD-like flow–accumulation relationship. Akin to the triangular MFD in Figure~\ref{fig:triangular-mfd}, throughput in the bottleneck setting increases linearly with vehicle accumulation up to the critical threshold $n_c$. Beyond $n_c$, the throughput remains at the bottleneck capacity $\mu$ rather than declining in the congested regime as in the case of the triangular MFD.

%Unlike the triangular MFD in Figure~\ref{fig:triangular-mfd}, the throughput plateaus at the bottleneck capacity $\mu$. Note that unlike the triangular MFD in Figure~\ref{fig:triangular-mfd}, the relationship between system throughput and the number of vehicles in the system is linear up to a threshold $n_c$ and then plataeus at the maximum flow or capacity of the bottleneck rather than declines as in the case of the triangular MFD.

\vspace{-17pt}

\begin{figure}[tbh!]
      \centering
      \includegraphics[width=90mm]{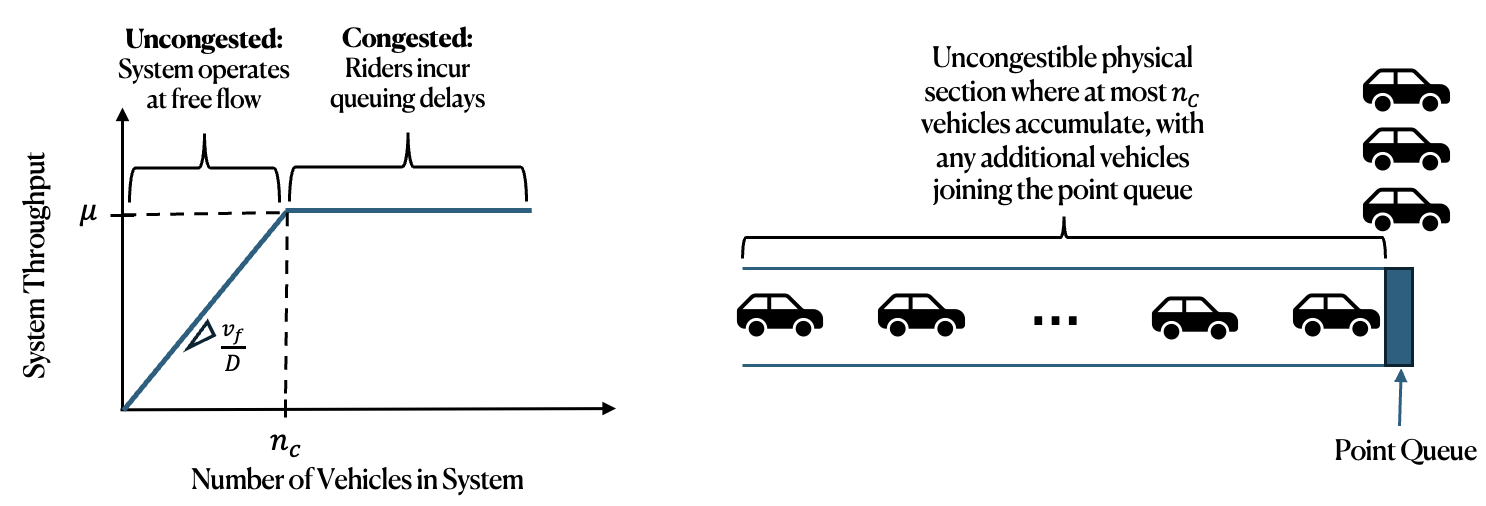}
      \vspace{-10pt}
      \caption{\footnotesize \sf Bottleneck with a point queue and its implied throughput-accumulation relation. The bottleneck consists of an uncongestible physical section (e.g., a bridge of length $D$) followed by a downstream point queue that occupies no physical space. The physical section holds at most $n_c$ vehicles and any additional vehicles join the point queue. When $n \leq n_c$, the system operates in free flow with a speed $v_f$. Beyond $n_c$, the system throughput remains at the bottleneck capacity $\mu$ and vehicles incur delays. %The slope of the segment from $(0,0)$ to any point on the curve determines average system speed, equal to trip distance $D$ times the slope. %, which declines when $n > n_c$.
      }
      \label{fig:conceptual-bottleneck-mfd} 
   \end{figure}  

\vspace{-22pt}

\section{Proofs} \label{apdx:proofs}

\vspace{-3pt}
\subsection{Proof of Theorem~\ref{thm:rev-comp-static-v-dynamic}} \label{apdx:rev-comp-stativ-v-dynamic}
\vspace{-3pt}

We analyze three regimes for $z_T - z_C$, each corresponding to a distinct revenue under the static revenue-optimal toll:
%, given by: %which is obtained by substituting the optimal toll from Proposition~\ref{prop:rev-opt-static-tolls} into the revenue expression in Equation~\eqref{eq:rev-opt-static-bottleneck}:
%in Equation~\eqref{eq:revenue-static-opt-toll-expression}.
%First, substituting the static revenue-optimal tolls from Proposition~\ref{prop:rev-opt-static-tolls} into the revenue expression in Equation~\eqref{eq:rev-opt-static-bottleneck}, we obtain: %First, the revenue under static revenue-optimal tolls derived in Proposition~\ref{prop:rev-opt-static-tolls} is given by:
{\setlength{\abovedisplayskip}{1pt}
\setlength{\belowdisplayskip}{1pt}
\begin{align*} %\label{eq:revenue-static-opt-toll-expression}
    R(\tau_s^*) = 
    \begin{cases}
    (z_T - z_C) \Lambda \frac{\mu}{\lambda}, & \text{if } 0 \leq z_T - z_C < \frac{\Lambda eL}{(\lambda - \mu)(e+L)} \\ 
    \frac{\mu\, eL}{4\,(e+L)\left(1-\frac{\mu}{\lambda}\right)} \left[ \frac{\Lambda}{\lambda} \! + \! \frac{e+L}{eL}\left(1 \! - \! \frac{\mu}{\lambda}\right) (z_T - z_C) \right]^2, & \text{if } z_T - z_C \! \in \! \left[ \frac{\Lambda eL}{(\lambda - \mu)(e+L)}, \frac{\Lambda eL}{(e+L)} \left( \frac{1}{\lambda - \mu} \! + \! \frac{2}{\mu} \right) \right] %\tau^* = \frac{z_T - z_C}{2} + \frac{\Lambda eL}{2(\lambda - \mu)(e+L)} 
    \\ 
    (z_T - z_C) \Lambda - \frac{\Lambda^2 eL}{\mu(e+L)} , & \text{if } z_T - z_C > \frac{\Lambda eL}{(e+L)} \left( \frac{1}{\lambda - \mu} \! + \! \frac{2}{\mu} \right).
    %\tau^* = (z_T - z_C) - \frac{\Lambda eL}{\mu(e+L)}.
\end{cases}
\end{align*}}

%Then, we establish our result by analyzing three regimes for $z_T - z_C$, each corresponding to a distinct revenue expression under the static revenue-optimal toll in Equation~\eqref{eq:revenue-static-opt-toll-expression}.

%Then, we establish our result by considering three cases corresponding to the different parameter regimes of $z_T - z_C$, each corresponding to a distinct revenue expression under the static revenue-optimal tolling policy in Equation~\eqref{eq:revenue-static-opt-toll-expression}.

\textbf{Case i $\left(z_T - z_C < \frac{\Lambda eL}{(\lambda - \mu)(e+L)} \right)$:} Note $z_T - z_C < \frac{\Lambda eL}{(\lambda - \mu)(e+L)} \frac{\lambda}{\mu}$ as $\lambda > \mu$; hence, from Theorem~\ref{thm:rev-opt-dynamic-tolls}, the optimal dynamic revenue is $(z_T - z_C) \Lambda \frac{\mu}{\lambda} + \frac{(z_T - z_C)^2 \mu (e+L)}{2 e L} \left( 1 - \frac{\mu}{\lambda} \right)^2$. If \(z_T=z_C\), both the optimal static and dynamic revenues are zero, so the claimed guarantee holds trivially. Hence, let $z_T > z_C$ and let $s \geq 1$ be such that $s (z_T - z_C) = \frac{\Lambda eL}{(\lambda - \mu)(e+L)}$. Then:
{\setlength{\abovedisplayskip}{1pt}
\setlength{\belowdisplayskip}{1pt}
\begin{align*}
    R(\tau_d^*(\cdot)) &= (z_T - z_C) \Lambda \frac{\mu}{\lambda} + \frac{(z_T - z_C)^2 \mu (e+L)}{2 e L} \left( 1 - \frac{\mu}{\lambda} \right)^2 \stackrel{(a)}{=} R(\tau_s^*) + \frac{(z_T - z_C)^2 \mu (e+L)}{2 e L} \left( 1 - \frac{\mu}{\lambda} \right)^2, \\
    &= R(\tau_s^*) \! + \! (z_T - z_C)\frac{\Lambda eL}{s(\lambda - \mu)(e+L)} \frac{\mu (e+L)}{2eL} \left( \frac{\lambda - \mu}{\lambda} \right)^2 \! = R(\tau_s^*) \! + \! \frac{1}{2s} (z_T - z_C) \Lambda\frac{\mu}{\lambda} \left( 1 - \frac{\mu}{\lambda} \right), \\
    &= R(\tau_s^*) + R(\tau_s^*) \frac{1}{2s} \left( 1 - \frac{\mu}{\lambda} \right) = R(\tau_s^*) \left( 1 + \frac{1}{2s} \left( 1 - \frac{\mu}{\lambda} \right) \right) 
\end{align*}}
where (a) follows as $R(\tau_s^*) = (z_T - z_C) \Lambda \frac{\mu}{\lambda}$. Thus, we obtain the $\frac{2s}{2s + 1 - \frac{\mu}{\lambda}} \geq \frac{2}{3 - \frac{\mu}{\lambda}}$ bound on the revenue ratio.

%the optimal static toll achieves at least $\frac{2s}{2s + 1 - \frac{\mu}{\lambda}} \geq \frac{2}{3 - \frac{\mu}{\lambda}}$ fraction of the optimal dynamic revenue.

\textbf{Case ii $\left(z_T - z_C \in \left[ \frac{\Lambda eL}{(\lambda - \mu)(e+L)}, \frac{\Lambda eL}{(e+L)} \left( \frac{1}{\lambda - \mu} + \frac{2}{\mu} \right) \right] \right)$:} Letting $s \leq 1$ be such that $s(z_T - z_C) = \frac{\Lambda eL}{(\lambda - \mu)(e+L)}$, we have
{\setlength{\abovedisplayskip}{2pt}
\setlength{\belowdisplayskip}{2pt}
\begin{align*}
    R(\tau_s^*) &= \frac{\mu eL}{4 (e+L)\left(1-\frac{\mu}{\lambda}\right)} \left[ \frac{\Lambda}{\lambda} + \frac{e+L}{eL}\left(1-\frac{\mu}{\lambda}\right) (z_T - z_C) \right]^2
    %&= \frac{\mu eL}{4(e+L)(\lambda - \mu)} \frac{\Lambda^2}{\lambda} + \frac{\Lambda}{2} \frac{\mu}{\lambda} (z_T - z_C) + \frac{(z_T - z_C)^2 \mu (e+L)}{4 eL} \left( 1 - \frac{\mu}{\lambda} \right), \\
    \stackrel{(a)}{=} \frac{\Lambda s}{4} \frac{\mu}{\lambda} (z_T - z_C) + \frac{\Lambda}{2} \frac{\mu}{\lambda} (z_T - z_C) + \frac{(z_T - z_C)^2 \mu (e+L)}{4 eL} \left( 1 - \frac{\mu}{\lambda} \right), \\
    &\geq \min \left\{ \frac{1}{2} + \frac{s}{4}, \frac{1}{2(1 - \frac{\mu}{\lambda})} \right\} \left[ \Lambda \frac{\mu}{\lambda} (z_T - z_C) + \frac{(z_T - z_C)^2 \mu (e+L)}{2 eL} \left( 1 - \frac{\mu}{\lambda} \right)^2 \right]  \stackrel{(b)}{\geq} \min \left\{ \frac{2+s}{4}, \frac{1}{2(1 - \frac{\mu}{\lambda})} \right\} R(\tau_d^* (\cdot)),
\end{align*}}
where (a) uses $s(z_T - z_C) = \frac{\Lambda eL}{(\lambda - \mu)(e+L)}$ and (b) follows as the term in the square brackets is the unconstrained optimum of Problem~\eqref{eq:constrained-dynamic-rev-opt}. Thus, optimal static tolls achieve at least $\min \left\{ \frac{2+s}{4}, \frac{1}{2(1 - \frac{\mu}{\lambda})} \right\}$ fraction of the optimal dynamic revenue. %, where $s \leq 1$.

%$z_T - z_C \geq \frac{\Lambda eL}{(\lambda - \mu)(e+L)}$. In this setting, we analyze two sub-cases depending on whether the static optimal toll is (a) $\frac{z_T - z_C}{2} + \frac{\Lambda eL}{2(\lambda - \mu)(e+L)}$ or (b) $z_T - z_C - \frac{\Lambda e L}{\mu(e+L)}$.

%\textbf{Case (iia):} We let $s \leq 1$ be a constant such that $s(z_T - z_C) = \frac{\Lambda eL}{(\lambda - \mu)(e+L)}$. In this setting, it holds that the optimal static revenue is given by:

\textbf{Case iii $\left( z_T - z_C > \frac{\Lambda eL}{e+L} \left( \frac{1}{\lambda - \mu} + \frac{2}{\mu} \right) \right)$:} In this case, %the static revenue-optimal toll 
$\tau_s^* = z_T - z_C - \frac{\Lambda e L}{\mu(e+L)}$, which implies: $z_T - z_C - \frac{\Lambda e L}{\mu(e+L)} > \frac{z_T - z_C}{2} + \frac{\Lambda eL}{2(\lambda - \mu)(e+L)}$. This, in turn, implies: $\! z_T - z_C \! > \! \frac{\Lambda eL}{(\lambda - \mu)(e+L)} \! + \! 2 \frac{\Lambda eL}{\mu(e+L)} \! = \! \frac{\Lambda eL}{e+L} \left( \frac{1}{\lambda - \mu} \! + \! \frac{2}{\mu} \right) \! > \! \frac{\Lambda eL}{e \!+ \! L} \left( \frac{1}{\lambda - \mu} \! + \! \frac{1}{\mu} \right) \! = \! \frac{\Lambda eL}{(\lambda \! - \! \mu)(e\!+ \! L)} \frac{\lambda}{\mu}$.

When $z_T - z_C > \frac{\Lambda eL}{(\lambda - \mu)(e+L)} \frac{\lambda}{\mu}$, %by Theorem~\ref{thm:rev-opt-dynamic-tolls}, 
the dynamic optimal revenue is given by $(z_T - z_C) \Lambda - \frac{\Lambda^2}{2 \mu} \frac{eL}{e+L}$. Then:
{\setlength{\abovedisplayskip}{1pt}
\setlength{\belowdisplayskip}{0pt}
\begin{align*}
    \frac{R(\tau_s^*)}{R_d^*} \! & %= \! \frac{(z_T - z_C) \Lambda - \frac{\Lambda^2}{\mu} \frac{eL}{e+L}}{(z_T - z_C) \Lambda - \frac{\Lambda^2}{2 \mu} \frac{eL}{e+L}} \! 
    = \! 1 \! - \! \frac{\frac{\Lambda^2}{2 \mu} \frac{eL}{e+L}}{(z_T - z_C) \Lambda - \frac{\Lambda^2}{2 \mu} \frac{eL}{e+L}} \! 
    %= \! 1 \! - \! \frac{1}{\frac{(z_T - z_C) \Lambda}{\frac{\Lambda^2}{2 \mu} \frac{eL}{e+L}} - 1} \! 
    = \! 1 \! - \! \frac{1}{\frac{(z_T - z_C)(e+L)}{\Lambda eL}(2 \mu) \! - \! 1} 
    &\stackrel{(a)}{\geq} 1 - \frac{1}{ \left( \frac{1}{\lambda - \mu} + \frac{2}{\mu} \right) (2 \mu) - 1} %= 1 - \frac{1}{\frac{2(2 \lambda - \mu)}{\lambda - \mu} - 1} 
    = 1 - \frac{\lambda - \mu}{3 \lambda - \mu} \geq %1 - \frac{\lambda - \mu}{3(\lambda - \mu)} = 
    \frac{2}{3}
\end{align*}}
where (a) follows as $z_T - z_C > \frac{\Lambda eL}{e+L} \left( \frac{1}{\lambda - \mu} + \frac{2}{\mu} \right)$. This establishes our claim.

\vspace{-8pt}

\subsection{Proof of Theorem~\ref{thm:sc-comp-static-v-dynamic}} \label{apdx:pf-sc-comp}

\vspace{-5pt}

We first derive the total system cost under the static and dynamic revenue-optimal policies and the optimal cost $SC^*$.

\emph{Total System Cost of Dynamic Revenue-Optimal Tolls:} Under the dynamic revenue-optimal policy, recall from the proof of Theorem~\ref{thm:rev-opt-dynamic-tolls} that there are no queuing delays. Letting, $N_e$ be the number of users that pass the bottleneck early (i.e., between  $[t_A^*, t_B^*]$ in Figure~\ref{fig:rev-opt-dynamic-tolls}) and $N_L$ be the number of users that pass the bottleneck late (i.e., between  $[t_C^*, t_D^*]$ in Figure~\ref{fig:rev-opt-dynamic-tolls}), the total system cost under the dynamic revenue-optimal policy for the fraction $f^*$ defined in Theorem~\ref{thm:rev-opt-dynamic-tolls} is:
{\setlength{\abovedisplayskip}{1pt}
\setlength{\belowdisplayskip}{1pt}
\begin{align*}
    SC(\tau^*_{d}(\cdot)) = z_T f^* \Lambda \left( 1 - \frac{\mu}{\lambda} \right) + z_C \left( f^* \Lambda \frac{\mu}{\lambda} + (1 - f^*) \Lambda \right) + e N_e \frac{t_1 - t_A^*}{2} + L N_L \frac{t_D^* - t_2}{2},
\end{align*}}
which is the sum of (i) the cost of using transit, (ii) the free-flow cost of using a car, and (iii) schedule delay costs. Then, noting that $N_e = \frac{L}{e+L} (1 - f^*)\Lambda$, $N_L = \frac{e}{e+L} (1 - f^*)\Lambda$, and substituting the relations for $t_1 - t_A^*$ and $t_D^* - t_2$ based on the analysis in the proof of Theorem~\ref{thm:rev-opt-dynamic-tolls}, the total system cost of the revenue-optimal dynamic tolling policy is:
{\setlength{\abovedisplayskip}{1pt}
\setlength{\belowdisplayskip}{1pt}
\begin{align*}
    SC(\tau^*_{d}(\cdot)) &= z_T f^* \Lambda \left( 1 - \frac{\mu}{\lambda} \right) + z_C \left( f^* \Lambda \frac{\mu}{\lambda} + (1 - f^*) \Lambda \right) + \frac{\Lambda^2 eL}{2 \mu (e+L)} (1 - f^*)^2 \left( 1 - \frac{\mu}{\lambda} \right). %\\
    %&= ,
\end{align*}}
When $z_T - z_C \! \leq \! \frac{\Lambda eL}{\mu (e+L)}$, we have $z_T - z_C \! \leq \! \frac{\Lambda eL}{\mu (e+L)} \frac{\lambda}{\lambda - \mu}$; hence, Theorem~\ref{thm:rev-opt-dynamic-tolls} implies $f^* \! = \! 1 \! - \! \frac{(z_T - z_C) \mu (e+L)}{\Lambda e L} \left( 1 - \frac{\mu}{\lambda} \right)$. Then:
{\setlength{\abovedisplayskip}{1pt}
\setlength{\belowdisplayskip}{1pt}
\begin{align} \label{eq:dynamic-rev-opt-sc-expression}
    SC(\tau^*_{d}(\cdot)) = z_C \Lambda \frac{\mu}{\lambda} + z_T \left( 1 - \frac{\mu}{\lambda} \right) \Lambda - \frac{(z_T-z_C)^2\mu(e+L)}{2eL} \left(1-\frac{\mu}{\lambda}\right)^2 \left(1+\frac{\mu}{\lambda}\right).
\end{align}}

\emph{Total System Cost of Static Revenue-Optimal Tolls:} For a static toll $\tau$, we show in Appendix~\ref{apdx:sco-static-toll} that: 
{\setlength{\abovedisplayskip}{1pt}
\setlength{\belowdisplayskip}{1pt}
\begin{align} \label{eq:sc-expression-toll-static}
    SC(\tau) &= z_T \left[ 1 \! - \! \frac{z_T - z_C - \tau}{\frac{\Lambda e L}{\mu(e+L)}} \right] \Lambda \left( 1 \! - \! \frac{\mu}{\lambda}\right) + z_C \left[ \mu (z_T - z_C - \tau) \left( \frac{1}{e} \! + \! \frac{1}{L} \right) \! + \! \left( 1 \! - \! \frac{z_T - z_C - \tau}{\frac{\Lambda e L}{\mu(e+L)}} \right) \frac{\Lambda \mu}{\lambda} \right] \\
    &\! \! \! \! \! \! \! \! \! + \frac{\mu (z_T  \!- \! z_C \! - \tau)^2}{2} \left(1 \! - \! \frac{\mu}{\lambda} \right) \left( \frac{1}{e} \! + \! \frac{1}{L} \right) \! + \! (z_T - z_C - \tau)\left( \left( 1 \! - \! \frac{z_T \! - \! z_C \! - \! \tau}{\frac{\Lambda e L}{\mu(e+L)}} \right) \Lambda \frac{\mu}{\lambda} \! + \! \frac{\mu (z_T \! - \! z_C \! - \! \tau)}{2} \left( \frac{1}{e} \! + \! \frac{1}{L} \right)  \right), \nonumber
\end{align}}
which is the sum of the (i) cost of using transit, (ii) the free-flow cost of using a car, (iii) schedule delay costs, and (iv) waiting time delays. From this expression, in the regime $z_T - z_C \leq \frac{\Lambda eL}{\mu (e+L)}$, we derive the total system cost under the two feasible static revenue-optimal tolls in Equation~\eqref{eq:optimal-static-toll}: (i) $\tau_s^* = z_T - z_C$ and (ii) $\tau_s^* = \frac{z_T - z_C}{2} + \frac{\Lambda eL}{2(\lambda - \mu)(e+L)}$: % Specifically, we obtain:
%\small
{\setlength{\abovedisplayskip}{2pt}
\setlength{\belowdisplayskip}{2pt}
\begin{align*} %\label{eq:sc-static-rev-opt-toll-expression}
    SC(\tau^*_{s}) \! = \! 
    \begin{cases}
        z_C \Lambda \frac{\mu}{\lambda} + z_T \left( 1 - \frac{\mu}{\lambda} \right) \Lambda, & \text{if }  \tau_s^* \! = \! z_T - z_C \\
        z_C \Lambda \frac{\mu}{\lambda} \! + \! z_T \left( 1 \! - \! \frac{\mu}{\lambda} \right) \! \Lambda \! - \! \frac{\mu (2\lambda \! - \! \mu)}{8 \lambda eL (e+L) (\lambda - \mu)^2} \left( (z_T \! - \! z_C)(e \! + \! L)(\lambda \! - \! \mu) \! - \! \Lambda eL  \right)^2\!, \! \! \! \!& \text{if } \tau_s^* \! = \! \frac{z_T - z_C}{2} \! + \! \frac{\Lambda eL}{2(\lambda  - \mu) \! (e+L)} \\
    \end{cases}
\end{align*} }
%\normalsize
We note that the third candidate revenue-optimal toll, $\tau_s^* = (z_T - z_C) - \frac{\Lambda eL}{\mu(e+L)}$ (see Proposition~\ref{prop:rev-opt-static-tolls}), is non-positive in the parameter regime of interest and thus is not relevant for our analysis here.

\emph{Optimal System Cost:} From~\citet{GONZALES20121519}, the minimum achievable system cost is $SC^* = z_C \Lambda + \left( 1 - \frac{\mu}{\lambda} \right) \Lambda (z_T - z_C) - \left( 1 - \frac{\mu}{\lambda} \right) \frac{\mu (e+L)}{2eL} (z_T - z_C)^2$ in the regime when $0 < z_T - z_C \leq \frac{\Lambda eL}{\mu(e+L)}$.

\emph{System Cost Comparisons:} From the above derived relations, $SC(\tau_s^*), SC(\tau_d^*(\cdot)) \leq z_C \Lambda \frac{\mu}{\lambda} + z_T \left( 1 - \frac{\mu}{\lambda} \right) \Lambda$. Then, the ratio of the system costs under these two policies satisfies (which we present here for the static revenue-optimal policy):
{\setlength{\abovedisplayskip}{2pt}
\setlength{\belowdisplayskip}{2pt}
\begin{align*}
    \frac{SC(\tau^*_s)}{SC^*} &\leq \frac{z_C \Lambda \frac{\mu}{\lambda} + z_T \left( 1 - \frac{\mu}{\lambda} \right) \Lambda}{z_C \Lambda + \left( 1 - \frac{\mu}{\lambda} \right) \Lambda (z_T - z_C) - \left( 1 - \frac{\mu}{\lambda} \right) \frac{\mu (e+L)}{2eL} (z_T - z_C)^2} \\ %= 1 + \frac{\left( 1 - \frac{\mu}{\lambda} \right) \frac{\mu (e+L)}{2eL} (z_T - z_C)^2}{z_T \Lambda - (z_T - z_C) \frac{\mu}{\lambda} \Lambda - (z_T - z_C) \frac{\Lambda}{2} \left( 1 - \frac{\mu}{\lambda} \right)}, \\
    &\stackrel{(a)}{\leq} 1 + \frac{(z_T - z_C) \frac{\Lambda}{2} \left( 1 - \frac{\mu}{\lambda} \right)}{z_T \Lambda - (z_T - z_C) \frac{\mu}{\lambda} \Lambda - (z_T - z_C) \frac{\Lambda}{2} \left( 1 - \frac{\mu}{\lambda} \right)} \stackrel{(b)}{\leq} %1 + \frac{\frac{(z_T - z_C)}{2} \left( 1 - \frac{\mu}{\lambda} \right)}{\frac{(z_T - z_C)}{2} \left( 1 - \frac{\mu}{\lambda} \right)} = 
    2,
\end{align*}}
where (a) follows as $z_T - z_C < \frac{\Lambda eL}{\mu (e+L)}$, and (b) follows by subtracting $z_C \Lambda$ in the denominator. %This establish our desired system cost ratio. % for static revenue-optimal tolling. Note that the same set of inequalities applies for the dynamic revenue-optimal tolling policy, as its total system cost also satisfies $SC(\tau_s^*), SC(\tau_d^*(\cdot)) \leq z_C \Lambda \frac{\mu}{\lambda} + z_T \left( 1 - \frac{\mu}{\lambda} \right) \Lambda$. This establishes our claim.

\vspace{-3pt}
\subsection{Proof of Proposition~\ref{prop:unbounded-sc-ratios}} \label{apdx:unbounded-sc-ratios}
\vspace{-3pt}

In the regime when $z_T - z_C > \frac{\Lambda e L}{\mu (e+L)}$, the minimum achievable total system cost reduces to that in the classical bottleneck model without an outside option and is given by $SC^* = z_C \Lambda + \frac{eL}{2(e+L)} \Lambda^2 \left( \frac{1}{\mu} - \frac{1}{\lambda} \right)$. To illustrate the implications of this regime, consider an instance in which $z_C = 0$ (the argument extends directly to settings in which $z_C$ is a small positive constant) and $z_T > \frac{\Lambda eL}{(e+L)(\lambda - \mu)} \left( \frac{2\lambda - \mu}{\mu} \right)$. 
In this setting, the optimal static toll is $\tau_s^* = (z_T - z_C) - \frac{\Lambda eL}{\mu(e+L)}$ and its total system cost is given by $z_C \Lambda + \frac{\Lambda^2 eL}{2 \mu (e+L)} \left( 1 - \frac{\mu}{\lambda} \right) + \frac{\Lambda^2 eL}{2 \mu (e+L)}$. Then, we have: $\frac{SC(\tau_s^*)}{SC^*} = \frac{z_C \Lambda + \frac{\Lambda^2 eL}{2 \mu (e+L)} \left( 1 - \frac{\mu}{\lambda} \right) + \frac{\Lambda^2 eL}{2 \mu (e+L)}}{z_C \Lambda + \frac{\Lambda^2 eL}{2 \mu (e+L)} \left( 1 - \frac{\mu}{\lambda} \right)} \stackrel{(a)}{=} 1 + \frac{1}{1 - \frac{\mu}{\lambda}},$
%Then, combining this expression for the minimum total system cost to the corresponding expression for the total system cost induced by the static revenue-optimal toll derived in the proof of Theorem~\ref{thm:sc-comp-static-v-dynamic}, we get:
where the final equality follows from letting $z_C = 0$. Consequently, if $\frac{\mu}{\lambda} \rightarrow 1$, the ratio of the system cost of the static revenue-optimal policy to the minimum achievable system cost can be unbounded in the regime when $\frac{\mu}{\lambda} \rightarrow 1$.

\vspace{-6pt}
\subsection{Proof of Theorem~\ref{thm:stochastic-extension}} \label{apdx:rev-stochastic-extension}
\vspace{-5pt}

Let $\kappa:=\frac{eL}{e+L},$ $\Lambda_\omega:=\lambda_\omega(t_2-t_1)$, and $\Delta:=z_T-z_C>0$. For a realization $\omega$, let $D_\omega(\tau)$ be the mass of users who choose car travel under a static toll $\tau\in[0,\Delta]$. %We first characterize $D_\omega(\tau)$ and relate it to the optimal dynamic revenue in that realization.
If $\mu_\omega\geq\lambda_\omega$, all users can travel at their desired times without delays. Hence, $D_\omega(\tau)=\Lambda_\omega$ for all $\tau \in [0, \Delta]$. %, where at $\tau=\Delta$ we use revenue-favoring tie-breaking. 
Thus, the optimal dynamic revenue is $R_{d,\omega}^*=\Delta\Lambda_\omega
=\int_0^\Delta D_\omega(\tau)\,d\tau.$

Now suppose $\mu_\omega<\lambda_\omega$, and define $\rho_\omega:=\frac{\mu_\omega}{\lambda_\omega}\in(0,1)$, $q_\omega:=1-\rho_\omega$, and $H_\omega:=\frac{\kappa\Lambda_\omega}{\mu_\omega}.$ Then, under a static toll $\tau$, using the relations derived in the proof of Proposition~\ref{prop:rev-opt-static-tolls}, the number of car users at equilibrium is
{\setlength{\abovedisplayskip}{1pt}
\setlength{\belowdisplayskip}{1pt}
\begin{equation}\label{eq:Domega}
D_\omega(\tau)
=
\min\left\{
\Lambda_\omega,
\rho_\omega\Lambda_\omega
+\frac{\mu_\omega q_\omega}{\kappa}(\Delta-\tau)
\right\}.
\end{equation}}
Hence $D_\omega(\tau)$ is non-negative, non-increasing, and concave in $\tau$. %, as it is constant at $\Lambda_\omega$ when the cap in \eqref{eq:Domega} binds and affine thereafter.

For the adaptive dynamic benchmark, conditional on $\omega$, the revenue optimization problem is exactly that in the deterministic bottleneck model in Theorem~\ref{thm:rev-opt-dynamic-tolls}. Applying that result with $(\lambda,\mu,\Lambda)$ replaced by $(\lambda_\omega,\mu_\omega,\Lambda_\omega)$ gives
{\setlength{\abovedisplayskip}{1pt}
\setlength{\belowdisplayskip}{1pt}
\begin{equation}\label{eq:dynamic-statewise}
R_{d,\omega}^*
=
\begin{cases}
\Delta\rho_\omega\Lambda_\omega
+\dfrac{\mu_\omega q_\omega^2}{2\kappa}\Delta^2,
& \displaystyle \Delta\leq \frac{H_\omega}{q_\omega},\\[8pt]
\Delta\Lambda_\omega
-\dfrac{\kappa\Lambda_\omega^2}{2\mu_\omega},
& \displaystyle \Delta>\frac{H_\omega}{q_\omega}.
\end{cases}
\end{equation}}

We now show that $\tau=\Delta/2$ obtains at least half of \eqref{eq:dynamic-statewise} in every realization. Letting $y_\omega \!:= \!\frac{\Delta}{H_\omega}$ and using \eqref{eq:Domega}:
{\setlength{\abovedisplayskip}{1pt}
\setlength{\belowdisplayskip}{1pt}
\begin{equation}\label{eq:midpoint-static}
\frac{2R_\omega(\Delta/2)}{\Delta\Lambda_\omega}
=
\min\left\{
1,\,
\rho_\omega+\frac{q_\omega y_\omega}{2}
\right\}.
\end{equation}}
On the other hand, \eqref{eq:dynamic-statewise} can be written as
{\setlength{\abovedisplayskip}{0pt}
\setlength{\belowdisplayskip}{0pt}
\begin{equation}\label{eq:dynamic-normalized}
\frac{R_{d,\omega}^*}{\Delta\Lambda_\omega}
=
\begin{cases}
\displaystyle \rho_\omega+\frac{q_\omega^2y_\omega}{2},
& y_\omega\leq 1/q_\omega,\\[7pt]
\displaystyle 1-\frac{1}{2y_\omega},
& y_\omega>1/q_\omega.
\end{cases}
\end{equation}}
Since $q_\omega^2\leq q_\omega$, in the regime when $y_\omega\leq 1/q_\omega$, it follows that (i) $\rho_\omega+\frac{q_\omega^2y_\omega}{2} \leq \rho_\omega+\frac{q_\omega y_\omega}{2}$ and (ii) $\rho_\omega+\frac{q_\omega^2y_\omega}{2} \leq \rho_\omega+\frac{q_\omega}{2}\leq 1.$
Thus, the right-hand side of \eqref{eq:midpoint-static} is at least the first expression in \eqref{eq:dynamic-normalized}. Moreover, note that the second expression in \eqref{eq:dynamic-normalized} is at most $1$ and in the regime that $y_\omega>1/q_\omega$: $\rho_\omega+\frac{q_\omega y_\omega}{2}
-\left(1-\frac{1}{2y_\omega}\right)
=
\frac{q_\omega(y_\omega-1)^2+\rho_\omega}{2y_\omega}
\geq 0.$

Hence, in either regime, $R_\omega\!\left(\frac{\Delta}{2}\right)
\geq \frac12 R_{d,\omega}^*.$ Together with the uncongested case (when $\mu_{\omega} \geq \lambda_{\omega}$), for which it is easy to check that the toll $\frac{\Delta}{2}$ is also half of the optimal dynamic revenue, this inequality holds for all $\omega$. Finally, $\Delta/2$ is a feasible fixed static toll and is the same across all $\omega$. Thus, $R_s^{\mathrm{fix}}
=\max_{0\leq\tau\leq\Delta}\mathbb{E}[R_\omega(\tau)]
\geq
\mathbb{E}\!\left[R_\omega\!\left(\frac{\Delta}{2}\right)\right]
\geq
\frac12\mathbb{E}[R_{d,\omega}^*] =\frac12 R_d^{\mathrm{adapt}}.$ % which proves the claim.

\vspace{-3pt}
\subsection{Proof of Theorem~\ref{thm:stochastic-system-cost}} \label{apdx:sc-stochastic-extension}
\vspace{-3pt}
Let $\kappa:=\frac{eL}{e+L},$ $\Lambda_\omega:=\lambda_\omega(t_2-t_1)$, and $\Delta:=z_T-z_C>0$. Fix $\omega$, and let $N_\omega$, $R_\omega$, and $C_\omega$ denote, respectively, the mass of car users, toll revenue, and aggregate waiting and schedule delay cost induced by the fixed static toll. Since every car user weakly prefers driving to transit, each user’s waiting and schedule delay cost plus toll payment is at most $\Delta$; summing over all car users therefore gives $C_\omega+R_\omega\leq \Delta N_\omega.$ Then, the corresponding system cost satisfies
{\setlength{\abovedisplayskip}{1pt}
\setlength{\belowdisplayskip}{1pt}
\begin{align}
\mathrm{SC}_\omega(\tau_s^{\mathrm{fix}})
&=z_C N_\omega+z_T(\Lambda_\omega-N_\omega)+C_\omega =z_C\Lambda_\omega+\Delta(\Lambda_\omega-N_\omega)+C_\omega  \leq z_C\Lambda_\omega+\Delta\Lambda_\omega-R_\omega.
\label{eq:stochastic-sc-upper}
\end{align}}

Next, define $\bar\rho_\omega:=\min\left\{1,\frac{\mu_\omega}{\lambda_\omega}\right\}$ and $q_\omega:=1-\bar\rho_\omega.$ Then, since the toll $\tau=\Delta$ is feasible, under the revenue-favoring tie-breaking convention, it serves $\bar\rho_\omega\Lambda_\omega$ car users in every realization, and therefore earns $\Delta\bar\rho_\omega\Lambda_\omega$ revenue. Since $\tau_s^{\mathrm{fix}}$ maximizes expected revenue, $R_s^{\mathrm{fix}} \geq \Delta\,\mathbb{E}[\bar\rho_\omega\Lambda_\omega]$, and taking expectations in Equation~\eqref{eq:stochastic-sc-upper} yields
{\setlength{\abovedisplayskip}{1pt}
\setlength{\belowdisplayskip}{1pt}
\begin{equation}
\mathbb{E}\!\left[\mathrm{SC}_\omega(\tau_s^{\mathrm{fix}})\right]
\leq
\mathbb{E}[B_\omega],
\qquad
B_\omega:=z_C\Lambda_\omega+\bar q_\omega\Delta\Lambda_\omega.
\label{eq:stochastic-sc-B}
\end{equation}}

It remains to compare $B_\omega$ with the minimum system cost. If $\mu_\omega\geq\lambda_\omega$, all users can travel by car at their desired times without queuing or schedule delay. Hence $\mathrm{SC}_\omega^*=z_C\Lambda_\omega=B_\omega.$

Now consider a realization with $\mu_\omega<\lambda_\omega$, and write $\rho_\omega:=\mu_\omega/\lambda_\omega$ and $q_\omega:=1-\rho_\omega$. Conditional on any realization $\omega$, the system planner faces the same optimization problem as in the deterministic model, with $(\lambda,\mu,\Lambda)$ replaced by $(\lambda_\omega,\mu_\omega,\Lambda_\omega)$. Then, following the proof of Theorem~\ref{thm:sc-comp-static-v-dynamic} and using $\mu_\omega\Delta/\kappa\leq\Lambda_\omega$, the optimal system cost satisfies:
{\setlength{\abovedisplayskip}{1pt}
\setlength{\belowdisplayskip}{1pt}
\begin{align*}
\mathrm{SC}_\omega^* &=z_C\Lambda_\omega +q_\omega\left( \Delta\Lambda_\omega-\frac{\mu_\omega\Delta^2}{2\kappa} \right)  \geq z_C\Lambda_\omega+\frac{1}{2}q_\omega\Delta\Lambda_\omega \geq \frac{1}{2}B_\omega, %\label{eq:stochastic-sc-lower}
\end{align*}}
%where we use $\mu_\omega\Delta/\kappa\leq\Lambda_\omega$. 
Thus $B_\omega\leq2\mathrm{SC}_\omega^*$ for all $\omega$. Combining this bound with \eqref{eq:stochastic-sc-B} and taking expectations proves our result.

\vspace{-5pt}
\subsection{Proof of Theorem~\ref{thm:rev-opt-dynamic-tolls-mfd}} \label{apdx:pf-mfd-dynamic-opt}
\vspace{-5pt}

As with the bottleneck model, a necessary equilibrium condition is $w'(t) + \tau'(t) = e$ for early arrivals and $w'(t) + \tau'(t) = -L$ for late arrivals, and $w(t) + \tau(t) \leq z_T - z_C$, i.e., the sum of the waiting and toll costs is as depicted in Figure~\ref{fig:eq-waiting-tolls}. Note here that the waiting time $w(t)$ can be expressed as a function of the number of users $n(t)$ in the system at time $t$, i.e., $w(t) = w(n(t))$, and the exact times $t_A', t_B', t_C', t_D'$ shown in Figure~\ref{fig:eq-waiting-tolls} are determined endogenously by the tolling policy. We now show that even for urban systems governed by an MFD, the dynamic revenue-optimal tolls correspond to operating on the throughput-maximizing point of the MFD where the waiting times are set point-wise to zero. %, i.e., the system always operates at free flow.

To establish this claim, we consider two tolling policies corresponding to possibly different values of the endogenously determined time points $t_A', t_B', t_C', t_D'$ in Figure~\ref{fig:eq-waiting-tolls}. First, consider the dynamic revenue-optimal policy $\tau^*(\cdot)$ with associated times $t_A^*, t_B^*, t_C^*, t_D^*$ in Figure~\ref{fig:eq-waiting-tolls}. %, which may differ from those under the static optimal toll.
Moreover, let $f$ be the fraction of time corresponding to the horizontal portion of the curve in Figure~\ref{fig:eq-waiting-tolls} under the optimal dynamic policy. Additionally, let $f_1$ be the fraction of users with desired arrival times in $[t_1, t_B^*]$ and $f_2$ be the fraction of users with desired arrival times in $[t_C^*, t_2]$, where $f_1 + f_2 = 1 - f$.

We define the second tolling policy as $\Tilde{\tau}(\cdot)$, with no waiting time delays (i.e., $w(t) = 0$ for all $t$ and the system operates at maximum throughput $\mu_f$) and $\Tilde{\tau}(t) = z_T - z_C$ for the period $[t_B^*, t_C^*]$. In this case, we define the associated times in Figure~\ref{fig:eq-waiting-tolls} as $\Tilde{t}_A, \Tilde{t}_B, \Tilde{t}_C, \Tilde{t}_D$, where $\Tilde{t}_B = t_B^*$ and $\Tilde{t}_C = t_C^*$. Under the quasi-static reservoir representation in Section~\ref{sec:mfd}, this construction is feasible. In particular, the inflow can be chosen so that the system reaches \(n_c\) and then equals the outflow \(\mu_f\), maintaining \(n(t)=n_c\). Also, as in the bottleneck model, users are ordered by desired crossing times. Then, letting the system throughput $\mu(t) = \mu(n(t))$ be time-varying and depend on the vehicle accumulation in the system, the optimal revenue $R^* = \int_{t_A^{*}}^{t_D^*} \mu(n(t)) \tau^*(t) dt = \int_{t_A^{*}}^{t_B^*} \mu(n(t)) \tau^*(t) dt + \int_{t_B^{*}}^{t_C^*} \mu(n(t)) \tau^*(t) dt + \int_{t_C^{*}}^{t_D^*} \mu(n(t)) \tau^*(t) dt,$ and the revenue of the second policy is: $\Tilde{R} = \int_{\Tilde{t}_A}^{\Tilde{t}_D} \mu_f \Tilde{\tau}(t) dt = \int_{\Tilde{t}_A}^{t_B^*} \mu_f \Tilde{\tau}(t) dt + \int_{t_B^{*}}^{t_C^*} \mu_f \Tilde{\tau}(t) dt + \int_{t_C^{*}}^{\Tilde{t}_D} \mu_f \Tilde{\tau}(t) dt.$

To establish our claim, we show $\Tilde{R} \geq R^*$ by proving three inequalities: (i) $\int_{t_B^{*}}^{t_C^*} \mu(n(t)) \tau^*(t) dt \leq \int_{t_B^{*}}^{t_C^*} \mu_f \Tilde{\tau}(t) dt$, (ii) $\int_{t_A^{*}}^{t_B^*} \mu(n(t)) \tau^*(t) dt \leq \int_{\Tilde{t}_A}^{t_B^*} \mu_f \Tilde{\tau}(t) dt$, and (iii) $\int_{t_C^{*}}^{t_D^*} \mu(n(t)) \tau^*(t) dt \leq \int_{t_C^{*}}^{\Tilde{t}_D} \mu_f \Tilde{\tau}(t) dt$.

\textbf{Proof of (i):} Note $\mu(n(t)) \tau^*(t) \leq \mu_f (z_T - z_C)$ for all $t$, as $\mu(n(t)) \leq \mu_f$ and $w(n(t)) + \tau^*(t) \leq z_T - z_C$ $\forall t$. Then:
{\setlength{\abovedisplayskip}{2pt}
\setlength{\belowdisplayskip}{2pt}
\begin{align*}
    \int_{t_B^{*}}^{t_C^*} \! \! \mu(n(t)) \tau^*(t) dt &\leq \! \! \int_{t_B^{*}}^{t_C^*} \! \! \mu_f (z_T - z_C) dt = \mu_f (z_T - z_C) (t_C^* - t_B^*) = \mu_f (z_T - z_C) f \frac{\Lambda}{\lambda} = \mu_f \! \! \int_{t_B^*}^{t_C^*} \! \! \Tilde{\tau}(t) dt.
\end{align*}}
%i.e., the second policy achieves a higher revenue in the period $[t_B^*, t_C^*]$ compared to the first.
%\vspace{-5pt}
\textbf{Proof of (ii):} First note that the revenue under the policy $\Tilde{\tau}(\cdot)$ in the range $[\Tilde{t}_A, t_B^*]$ is given by:
{\setlength{\abovedisplayskip}{1pt}
\setlength{\belowdisplayskip}{1pt}
\begin{align}
    \int_{\Tilde{t}_A}^{t_B^*} \! \! \!\! \mu_f \Tilde{\tau}(t) dt &\stackrel{(a)}{=} \int_{0}^{t_B^* - \Tilde{t}_A} \! \! \!\! \mu_f (\Tilde{\tau}(t_B^*) - e \Delta) d \Delta %= \mu_f \left[ \Tilde{\tau}(t_B^*) (t_B^* - \Tilde{t}_A) - \frac{e (t_B^* - \Tilde{t}_A)^2}{2} \right], \nonumber \\
    = \mu_f \left[ \Tilde{\tau}(t_B^*) - \frac{e (t_B^* - \Tilde{t}_A)}{2} \right] (t_B^* - \Tilde{t}_A) \stackrel{(b)}{=} \left[ z_T - z_C - \frac{e (t_B^* - \Tilde{t}_A)}{2} \right] f_1 \Lambda, \label{eq:early-dynamic-optimal-helper} %\\
    %&= \frac{\Tilde{\tau}(\Tilde{t}_A) + \Tilde{\tau}(t_B^*)}{2} f_1 \Lambda \stackrel{(c)}{=} \Tilde{\tau}\left(\frac{\Tilde{t}_A + t_B^*}{2} \right) f_1 \Lambda, %\\
    %&= \left( z_T - z_C - \frac{e (t_B^* - \Tilde{t}_A)}{2} \right) f_1 \Lambda \stackrel{(b)}{=}  \left( z_T - z_C - \frac{e f_1 \Lambda}{2 \mu_f} \right) f_1 \Lambda,
\end{align}}
where (a) follows from the structure of the tolling policy in Figure~\ref{fig:eq-waiting-tolls} and (b) follows as $f_1 \Lambda = \lambda (t_B^* - t_1) = \mu_f(t_B^* - \Tilde{t}_{A})$. Then, to establish (ii), we use a variable transformation $\Delta = t_B^* - t$ to get:
{\setlength{\abovedisplayskip}{1pt}
\setlength{\belowdisplayskip}{1pt}
\begin{align}
    \int_{t_A^*}^{t_B^*} \mu(t) \tau^*(t) dt &= \int_{0}^{t_B^* - t_A^*} \mu(\Delta) \tau^*(\Delta) d \Delta = \int_{0}^{t_B^* - t_A^*} \mu(\Delta) (z_T - z_C - e \Delta - w(\Delta)) d \Delta, \nonumber \\
    &\stackrel{(a)}{\leq} \int_{0}^{t_B^* - t_A^*} \mu(\Delta) (z_T - z_C - e \Delta) d \Delta, \nonumber \\
    &= \int_{0}^{t_B^* - \Tilde{t}_A} \mu(\Delta) (z_T - z_C - e \Delta) d \Delta + \int_{t_B^* - \Tilde{t}_A}^{t_B^* - t_A^*} \mu(\Delta) (z_T - z_C - e \Delta) d \Delta, \label{eq:ub1-earlyrevenue}
\end{align}}
where (a) follows as $w(\Delta) \geq 0$ for all $\Delta \in [t_A^*, t_B^*]$. Furthermore, for some $g \in [0, 1]$, we define: $\int_{0}^{t_B^* - \Tilde{t}_A} \mu(\Delta) d \Delta = g f_1 \Lambda$ and $\int_{t_B^* - \Tilde{t}_A}^{t_B^* - t_A^*} \mu(\Delta) d \Delta = (1-g) f_1 \Lambda.$ We use these expressions to upper bound the two terms in Equation~\eqref{eq:ub1-earlyrevenue}. Note that:
{\setlength{\abovedisplayskip}{2pt}
\setlength{\belowdisplayskip}{2pt}
\begin{align} \label{eq:ub-term1-earlyrevenue}
    \int_{t_B^* - \Tilde{t}_A}^{t_B^* - t_A^*} \mu(\Delta) (z_T - z_C - e \Delta) d \Delta \leq (1 - g) f_1 \Lambda (z_T - z_C - e (t_B^* - \Tilde{t}_A)),
\end{align}}
which follows as $\Delta \geq t_B^* - \Tilde{t}_A$ for all $\Delta \in [t_B^* - \Tilde{t}_A, t_B^* - t_A^*]$. Next, for the other term in Equation~\eqref{eq:ub1-earlyrevenue}: 
{\setlength{\abovedisplayskip}{2pt}
\setlength{\belowdisplayskip}{2pt}
\begin{align}
    \int_{0}^{t_B^* - \Tilde{t}_A} \! \! \! \! \! \mu(\Delta) (z_T - z_C - e \Delta) d \Delta &= \int_{0}^{t_B^* - \Tilde{t}_A} \! \! \! \! \! \mu(\Delta) (z_T - z_C) d \Delta - e \int_{0}^{t_B^* - \Tilde{t}_A} \! \! \! \! \! \mu(\Delta) \Delta d \Delta, \nonumber \\
    &\stackrel{(a)}{=} g f_1 \Lambda (z_T - z_C) - e \int_{0}^{t_B^* - \Tilde{t}_A} \! \! \! \! \! \mu(\Delta) \Delta d \Delta, \label{eq:ub-term2-earlyrevenue}
\end{align}}
where (a) follows as $\int_{0}^{t_B^* - \Tilde{t}_A} \mu(\Delta) d \Delta = g f_1 \Lambda$ and $z_T - z_C$ is a constant.

Next, from Equations~\eqref{eq:ub-term1-earlyrevenue} and~\eqref{eq:ub-term2-earlyrevenue}, we obtain the following upper bound on the revenue:
{\setlength{\abovedisplayskip}{1pt}
\setlength{\belowdisplayskip}{1pt}
\setlength{\jot}{-3pt}
\begin{align}
    \int_{t_A^*}^{t_B^*} \mu(t) \tau^*(t) dt &\leq g f_1 \Lambda (z_T - z_C) - e \int_{0}^{t_B^* - \Tilde{t}_A} \mu(\Delta) \Delta d \Delta + (1 - g) f_1 \Lambda (z_T - z_C - e (t_B^* - \Tilde{t}_A)), \nonumber \\
    &= f_1 \Lambda (z_T - z_C) - (1-g) f_1 \Lambda e (t_B^* - \Tilde{t}_A) - e \int_{0}^{t_B^* - \Tilde{t}_A} \mu(\Delta) \Delta d \Delta. \label{eq:ub2-earlyrevenue}
\end{align}}

Now, to bound the term $\int_{0}^{t_B^* - \Tilde{t}_A} \mu(\Delta) \Delta d \Delta$, we define $F(\Delta) = \int_0^{\Delta} \mu(x) dx$. Note that $F$ is absolutely continuous with $F'(\Delta) = \mu(\Delta)$. Consequently, using integration by parts, it follows that:
{\setlength{\abovedisplayskip}{1pt}
\setlength{\belowdisplayskip}{1pt}
\setlength{\jot}{0pt}
\begin{align}
    \int_{0}^{t_B^* - \Tilde{t}_A} \! \! \! \! \! \mu(\Delta) \Delta d \Delta &= \int_{0}^{t_B^* - \Tilde{t}_A} F'(\Delta) \Delta d \Delta = [\Delta F(\Delta)]_{0}^{t_B^* - \Tilde{t}_A} - \int_{0}^{t_B^* - \Tilde{t}_A} F(\Delta) d \Delta, \nonumber \\
    &= (t_B^* - \Tilde{t}_A) \int_{0}^{t_B^* - \Tilde{t}_A} \! \! \! \! \! \mu(\Delta) d \Delta - \int_{0}^{t_B^* - \Tilde{t}_A} \left( \int_0^{\Delta} \! \! \mu(x) dx \right) d \Delta, \nonumber \\
    &= (t_B^* - \Tilde{t}_A) g f_1 \Lambda - \int_{0}^{t_B^* - \Tilde{t}_A} \! \left( \int_0^{\Delta} \! \! \mu(x) dx \right) d \Delta. \label{eq:side-helper1}
\end{align}}
%where (a) follows from integration by parts.

Next, substituting Equation~\eqref{eq:side-helper1} in Equation~\eqref{eq:ub2-earlyrevenue}, we get:
{\setlength{\abovedisplayskip}{1pt}
\setlength{\belowdisplayskip}{1pt}
\setlength{\jot}{-1pt}
\begin{align}
    \int_{t_A^*}^{t_B^*} \! \! \! \mu(t) \tau^*(t) dt \! &\leq \! f_1 \Lambda (z_T - z_C) \! - \! (1-g) f_1 \Lambda e (t_B^* - \Tilde{t}_A) \! - \! e \bigg[ (t_B^* - \Tilde{t}_A) g f_1 \Lambda - \! \int_{0}^{t_B^* - \Tilde{t}_A} \! \left( \int_0^{\Delta} \! \! \! \mu(x) dx \right) d \Delta \bigg], \nonumber \\
    %&= f_1 \Lambda (z_T - z_C) - f_1 \Lambda e (t_B^* - \Tilde{t}_A) + e \int_{0}^{t_B^* - \Tilde{t}_A} \left( \int_0^{\Delta} \mu(x) dx \right) d \Delta, \nonumber\\
    &\stackrel{(a)}{\leq} f_1 \Lambda (z_T - z_C) - f_1 \Lambda e (t_B^* - \Tilde{t}_A) + e \int_{0}^{t_B^* - \Tilde{t}_A} \left( \int_0^{\Delta} \mu_f dx \right) d \Delta, \nonumber \\
    %&= f_1 \Lambda (z_T - z_C) - f_1 \Lambda e (t_B^* - \Tilde{t}_A) + e \int_{0}^{t_B^* - \Tilde{t}_A}  \mu_f \Delta d \Delta, \nonumber \\
    %&= f_1 \Lambda (z_T - z_C) - f_1 \Lambda e (t_B^* - \Tilde{t}_A) + e \mu_f \frac{(t_B^* - \Tilde{t}_A)^2}{2} 
    &\stackrel{(b)}{=} f_1 \Lambda (z_T - z_C) - f_1 \Lambda e (t_B^* - \Tilde{t}_A) + e f_1 \Lambda \frac{t_B^* - \Tilde{t}_A}{2} = f_1 \Lambda (z_T - z_C) - f_1 \Lambda e \frac{t_B^* - \Tilde{t}_A}{2}, \nonumber \\
    &\stackrel{(c)}{=} \int_{\Tilde{t}_A}^{t_B^*} \mu_f \Tilde{\tau}(t) dt, \nonumber
\end{align}}
where (a) follows as $\mu(x) \leq \mu_f$ for all $x \in [0, t_B^* - \Tilde{t}_A]$, (b) follows as $\mu_f (t_B^* - \Tilde{t}_A) = f_1 \Lambda$, and (c) follows by our derived relation for $\int_{\Tilde{t}_A}^{t_B^*} \mu_f \Tilde{\tau}(t) dt$ in Equation~\eqref{eq:early-dynamic-optimal-helper}. This establishes inequality (ii).

\textbf{Proof of (iii):} We prove this relation similarly to inequality (ii), and we omit details for brevity.
%we can show $\int_{t_C^{*}}^{\Tilde{t}_D} \mu_f \Tilde{\tau}(t) dt \geq \int_{t_C^{*}}^{t_D^*} \mu(n(t)) \tau^*(t) dt.$ We omit the details for brevity.

%\textbf{Concluding the Proof:}
Combining (i), (ii), and (iii), $R^* \leq \Tilde{R}$. %, i.e., for any tolling policy $\tau^*(\cdot)$ there exists another policy $\Tilde{\tau}(\cdot)$ that operates the system at the throughput-maximizing capacity at all periods and achieves a weakly higher revenue. 
Note that if the throughput induced by $\tau^*(\cdot)$ is strictly below $\mu_f$ over any sub-interval, then at least one of inequalities (i)–(iii) is strict, implying $R^* \! < \! \Tilde{R}$, a contradiction to the optimality of $\tau^*(\cdot)$. %, thus establishing our claim. 

\vspace{-3pt}
\subsection{Proof of Theorem~\ref{thm:rev-static-tolls-mfd}} \label{apdx:mfd-statc-revenue-exp}
\vspace{-3pt}
We prove this claim in two steps. First, under a triangular MFD, we characterize the equilibrium system throughput $\mu(n(t))$, which is time-varying and depends on the system's vehicle accumulation, unlike in the bottleneck setting, under a static toll $\tau$. Using this relation for the throughput, we then derive the expression in the theorem statement.

To calculate $\mu(n(t))$, first note that the static toll $\tau \geq 0$ must be such that there is some user with a cost of $z_T$ at equilibrium. If not, all users would take the car and thus the toll can be increased without reducing car ridership or revenues. Thus, let $\underline{\tau} \geq 0$ be the minimum toll at which there is some user indifferent between car and transit.

%first note that a static toll $\tau \geq 0$ must be such that $z_T - z_C - \tau \leq T_C$. Note that if $z_T - z_C - \tau > T_C$, then $\tau$ can be increased to $z_T - z_C - T_C$ without changing the number of car users, resulting in a higher revenue. Next, for any $\tau \geq \max \{ 0, z_T - z_C - T_C \}$, note that the equilibrium is as specified in Proposition~\ref{prop:rev-opt-static-tolls} other than that $T = z_T - z_C - \tau$.

Next, recall that under a static toll $\tau \geq \underline{\tau}$, the equilibrium waiting time is akin to the right of Figure~\ref{fig:equilibrium_both_models}, with the time points $t_A, t_B, t_C, t_D$ determined endogenously by the tolling policy and where the peak waiting time is $z_T - z_C - \tau$ during the interval $[t_B, t_C]$. Moreover, note that between $[t_A, t_B]$, the waiting time is given by $z_T - z_C - \tau - e \Delta$ where $\Delta \in [0, t_B - t_A]$. Furthermore, between $[t_C, t_D]$, the waiting time is given by $z_T - z_C - \tau - L \Delta$ where $\Delta \in [0, t_D - t_C]$.

Next, let $v_f$ be the free-flow speed in the uncongested part of the MFD, so free-flow travel time is $\frac{D}{v_f}$. Moreover, letting $v(n(t))$ be the system speed at time $t$, under the quasi-static representation, the travel time is $\frac{D}{v(n(t))}$. %the travel time with vehicle accumulation $n(t)$ is $\frac{D}{v(n(t))}$. 
Here, $v(n(t)) = D \frac{\mu(n(t))}{n(t)}$, since system speed equals $D$ times the slope from the origin to any point on the triangular MFD in Figure~\ref{fig:triangular-mfd}. Consequently, with slight abuse of notation, the waiting time satisfies: $w(n(t)) = \frac{D}{v(n(t))} - \frac{D}{v_f} = \frac{n(t)}{\mu(n(t))} - \frac{n_c}{\mu_f}.$ Moreover, by our triangular MFD assumption, $\frac{\mu_f}{\mu(n(t))} = \frac{n_j - n_c}{n_j - n(t)}$, which implies $n(t) = n_j - \frac{\mu(n(t))}{\mu_f} (n_j - n_c)$. Substituting this relation into the waiting time equation and rearranging, $\mu(n(t)) = \frac{n_j}{\frac{n_j}{\mu_f} + w(n(t))}$.

Let $\mu_{\tau}$ be the fixed system throughput during the period $[t_B, t_C]$ when the waiting time function is flat and at its peak under the static toll $\tau$. Then, using the above relation for $\mu(n(t))$, we have: %the system throughput as a function of $t$, we have the following expression for the revenue given a fixed toll $\tau$:
{\setlength{\abovedisplayskip}{2pt}
\setlength{\belowdisplayskip}{2pt}
\begin{align*}
    R(\tau) &= \tau \int_{t_A}^{t_D} \mu(n(t)) dt = \tau \left[ \int_{t_A}^{t_B} \mu(n(t)) dt + \int_{t_B}^{t_C} \mu(n(t)) dt + \int_{t_C}^{t_D} \mu(n(t)) dt \right], \\
    &\stackrel{(a)}{=} \! \tau \bigg[ \int_{0}^{t_B - t_A} \! \! \! \! \frac{n_j}{\frac{n_j}{\mu_f} + (z_T - z_C - \tau - \Delta e)} d \Delta + \int_{0}^{t_D - t_C} \! \! \! \! \frac{n_j}{\frac{n_j}{\mu_f} + (z_T - z_C - \tau - \Delta L)} d \Delta + \mu_{\tau} (t_C - t_B)  \bigg], \\
    &\stackrel{(b)}{=} \! \tau \bigg[ \frac{n_j}{e} \ln \left( 1 \! + \! \frac{(z_T - z_C - \tau) \mu_f}{n_j} \right) \! + \! \frac{n_j}{L} \ln \left( 1 \! + \! \frac{(z_T - z_C - \tau) \mu_f}{n_j} \right) \! + \! \frac{n_j}{\frac{n_j}{\mu_f} \! + z_T - z_C - \tau} (t_C - t_B) \bigg], \\
    &= \tau \bigg[ n_j \frac{e+L}{eL} \ln \left( 1 + \frac{(z_T - z_C - \tau) \mu_f}{n_j} \right) + \frac{n_j}{\frac{n_j}{\mu_f} + z_T - z_C - \tau} (t_C - t_B) \bigg],
\end{align*}}
where (a) follows by the variable transformation $\Delta = t_B - t$ for the integral between $[t_A, t_B]$ and the variable transformation $\Delta = t - t_C$ for the integral between $[t_C, t_D]$, and (b) follows by evaluating the integral and noting that the waiting time at $t_A$ and $t_D$ is zero. Next, noting that $\int_{t_A}^{t_B} \mu(n(t)) dt = \lambda (t_B - t_1)$ and $\int_{t_C}^{t_D} \mu(n(t)) dt = \lambda (t_2 - t_C)$, we have:
{\setlength{\abovedisplayskip}{2pt}
\setlength{\belowdisplayskip}{2pt}
\begin{align*}
    t_C - t_B = \frac{\Lambda}{\lambda} - \frac{1}{\lambda} \left[ \int_{t_A}^{t_B} \mu(n(t)) dt + \int_{t_C}^{t_D} \mu(n(t)) dt \right] = \frac{\Lambda}{\lambda} - \frac{1}{\lambda} n_j \frac{e+L}{eL} \ln \left( 1 + \frac{(z_T - z_C - \tau) \mu_f}{n_j} \right).
\end{align*}}
Substituting this into the above relation for the revenue and simplifying, we obtain our desired relation for the revenue. % as a function of the static toll $\tau$. %\Halmos
\begin{comment}
{\setlength{\abovedisplayskip}{2pt}
\setlength{\belowdisplayskip}{2pt}
\begin{align*}
    R(\tau) &= \! \tau \left[ n_j \frac{e+L}{eL} \log \left( 1 \! + \! \frac{(z_T - z_C - \tau) \mu_f}{n_j} \right) \! + \! \frac{n_j}{\frac{n_j}{\mu_f} + z_T - z_C - \tau} \left( \frac{\Lambda}{\lambda} - \frac{1}{\lambda} n_j \frac{e+L}{eL} \log \left( 1 \! + \! \frac{(z_T - z_C - \tau) \mu_f}{n_j} \right) \right) \right], \\
    &= \tau \left[ \frac{\Lambda}{\lambda} \frac{n_j}{\frac{n_j}{\mu_f} + z_T - z_C - \tau} + n_j \frac{e+L}{eL} \log \left( 1 + \frac{(z_T - z_C - \tau) \mu_f}{n_j} \right) \left( 1 - \frac{\frac{n_j}{\frac{n_j}{\mu_f} + z_T - z_C - \tau}}{\lambda} \right) \right],
\end{align*}}
which establishes our desired relation for the revenue as a function of the static toll $\tau$.
\end{comment}

\vspace{-3pt}
\section{Additional Details on Numerical Experiments} \label{apdx:additional-numerical-details}

\vspace{-3pt}

\subsection{Model Calibration Details} \label{apdx:model-calibration}

\vspace{-3pt}

This section describes our assumptions and methodology to calibrate our model parameters for the bottleneck and MFD frameworks based on the SF-Oakland Bay Bridge and New York City's CRZ case studies, respectively. Throughout the calibration, we treat each observed road vehicle as one user, corresponding to an assumed vehicle occupancy of one.

\textbf{Calibration of Bottleneck Model for Bay Area Study:} We consider weekday morning westbound trips from the East Bay to San Francisco between 5:00 AM and 10:00 AM. Commuters choose between driving across the SF-Oakland Bay Bridge and using BART, which serves as the outside option. All westbound vehicles using the SF–Oakland Bay Bridge are subject to a fixed (static) toll of \$8.50 for most commuter vehicles, and empirical evidence indicates that the SF–Oakland Bay Bridge operates as a true bottleneck~\citep{GONZALES2015267}. %, making this corridor a natural setting for studying the performance gap between static and dynamic tolling with an outside option.

We calibrate demand $\Lambda$ using average weekday westbound car and BART trips between 5:00 AM and 10:00 AM in August 2025. BART ridership data imply 27,132 westbound trips between 5:00 AM and 10:00 AM~\citep{BART_RidershipReports}, while PeMS data from a sensor (VDS 426389), located upstream of the Bay Bridge, records an average of 41,369 vehicles over the same period. This yields a total demand of 68,501, which we round to $\Lambda=70{,}000$, corresponding to $\lambda=14{,}000$ users per hour. Next, using the maximum observed flow at the same PeMS sensor, we set the bottleneck service rate to $\mu=9{,}600$ vehicles per hour, consistent with standard capacity ranges for a five-lane highway such as the SF–Oakland Bay Bridge. Under these calibrated parameters, the bottleneck service rate during the morning commute is strictly lower than the desired arrival rate (i.e., $\mu < \lambda$), implying congestion delays. % on the Bay Bridge.

%Next, to calibrate the bottleneck service rate, we use the observed maximum flow at detector VDS 426389 on the SF–Oakland Bay Bridge. The peak flow at this location is approximately $\mu = 9,600$ vehicles per hour, which is consistent with standard capacity ranges for a five-lane highway such as the SF–Oakland Bay Bridge. Under these calibrated parameters, observe that the bottleneck service rate during the morning commute is strictly lower than the desired arrival rate (i.e., $\mu < \lambda$), implying congestion delays on the Bay Bridge. These delays will, in turn, affect commuters’ departure-time decisions and may induce mode shifts toward the outside option.

Next, we calibrate the user cost parameters in Equation~\eqref{eq:userCostNormalized}. We set the value of waiting time to $c_W = \$22$, based on an inflation-adjusted estimate of the average value of travel time in the Bay Area~\citep{MTC2015}. Following~\cite{Small1982}, we set the schedule delay parameters to $e=0.61$ for earliness and $L=2.4$ for lateness. Moreover, we model the free-flow cost of car travel, $z_C$, as the sum of (i) the average free-flow travel time for westbound trips from the East Bay into San Francisco and (ii) the daily parking cost in San Francisco, converted into time-equivalent units through a normalization with the value of waiting time $c_W$. Daily parking costs are assumed to be \$30, consistent with weekday parking fees in SF. Using origin–destination shares and distances for westbound AM trips into San Francisco, we obtain a weighted average free-flow travel time of around 21 minutes. Details of the free-flow travel time calibration are in Appendix~\ref{subsec:fftt-calibration-bay-area}. The resulting free-flow cost of car travel is $z_C = 1.714$ hrs.

Finally, we calibrate the average cost $z_T$ of westbound BART trips from the East Bay into San Francisco as the sum of a monetary fare and a generalized time cost comprising (i) walking time to and from the BART station, (ii) travel time in the BART, and (iii) waiting time at the station. %In line with the empirical evidence that time spent using public transit is perceived as more onerous than time spent driving~\citep{Wardman2012}, 
We weight these time components by a multiplicative discomfort factor $\eta$. To estimate these components, we use BART GTFS data to compute ridership volume-weighted averages for east–west Bay trips during the morning rush, yielding an average fare of \$6.14 and an average in-vehicle travel time of approximately 32 minutes~\citep{BART_GTFS}. Walking time is assumed to be 10 minutes each way, resulting in 20 minutes of total access and egress time~\citep{BART_AccessGuidelines}. The waiting time is taken to be half of the average train headway of 20 minutes~\citep{BART_FY26Budget}. %Overall, $z_T$ equals the sum of the fare, converted to time-equivalent units using $c_W$, and the generalized travel time (walking, waiting, and travel time on BART) weighted by $\eta$.

\textbf{Calibration of MFD Framework for NYC Study:}
%We now calibrate the MFD framework to NYC's congestion pricing program. 
Since January 5, 2025, vehicles are tolled when entering NYC's Congestion Relief Zone (CRZ), covering Manhattan south of 60th Street. While rates vary by vehicle class and other factors, most commuter vehicles face a flat toll of \$9 on weekdays between 5:00 AM and 9:00 PM. For our empirical calibration, we consider commuter trips that either originate or terminate within the CRZ, so each trip occurs at least partly inside the tolled region. Consistent with the aggregate MFD framework, we toll all car trips involving the CRZ, interpreting the toll as a stylized area-use charge rather than an exact replication of the CRZ’s entry-based charging rule. Travelers choose between car, incurring a toll when using the CRZ, and the subway separated from the road network. To isolate the mechanisms of interest, we do not model bus or other transit modes in NYC.

As with the Bay Area study, we focus on the weekday morning rush between 5:00 AM and 10:00 AM. During this period, weekday averages for August 2025 indicate around 700,000 subway trips and 200,000 vehicle trips occurring at least partially within the CRZ, yielding $\Lambda=900{,}000$ and $\lambda=180{,}000$ users per hour. For details, see Appendix~\ref{apdx:calibration-travel-demand-nyc}.

To calibrate the triangular MFD, we use NYC taxi and for-hire vehicle data~\citep{NYC_TLC_TripData}, which imply an average vehicle trip length within the CRZ of approximately 3.7 miles (5.95 km). We thus set the trip distance to $D=6$ km and the free-flow speed to 25 mph, consistent with speed limits in Manhattan.

%for the New York City case study. To estimate the average trip distance $D$, we use New York City taxi and for-hire vehicle trip data~\citep{NYC_TLC_TripData} and find that the average trip length for vehicle trips within the CRZ is approximately 3.7 miles (5.95 km). We round this value to $D = 6$ km and use it as the representative trip distance for all vehicle trips in the zone. Moreover, we set the free-flow speed in Manhattan to $25mph$, consistent with posted speed limits on local Manhattan streets.

For jam density, assuming an average vehicle length of 5 meters and approximately 7 meters of road space per vehicle under congested conditions yields 142 vehicles/km/lane. We round this to 140 vehicles/km/lane, consistent with estimates in~\cite{Knoop2021}, and OpenStreetMap data imply approximately 1,000 lane-kilometers of roads within the CRZ~\citep{OpenStreetMap}, implying a system-level jam accumulation of roughly 140,000 vehicles. Because the effective jam accumulation may be lower, we conduct sensitivity analyses ranging from 14,000 to 140,000 vehicles. The revenue and system cost ratios are unchanged across these values, so the results in Figure~\ref{fig:mfd-ratios-nyc} are insensitive to jam accumulation. We calibrate the throughput-maximizing flow analogously to the arrival rate, obtaining $\mu_f=45{,}000$ vehicles per hour (see Appendix~\ref{apdx:mu_f_calibration_nyc} for details). Our results are robust to moderate variations in $\mu_f$.

%We find that our results are not sensitive to moderate variations in $\mu_f$, indicating that this value captures the key effects for our analysis.

For car travel costs, we set the value of waiting time to $c_W=\$40$, following~\citet{cook2025short} and consistent with mean hourly wages in New York City~\citep{BLS_OEWS_NY_May2024}. As in the Bay Area study, we set the schedule delay parameters to $e=0.61$ and $L=2.4$. The free-flow car cost $z_C$ combines parking costs and travel time (of an uncongested trip) in time-equivalent units. We assume a daily parking fee of \$30, corresponding to approximately \$3.50 per hour over an eight-hour workday, and set the free-flow travel time to $\frac{D}{v_f} = \frac{6 km}{40 km/h} = 0.15$ hours.

We calibrate the subway cost $z_T$ using the same formulation as in the Bay Area study. The trip fare is \$3 (as of January 2026), and weekday headways of less than five minutes imply an average waiting time of about 2.5 minutes~\citep{NYC_Comptroller_SubwayFrequencies}. We assume 10 minutes of walking per access and egress leg, resulting in a total of 20 minutes. %~\citep{StreetEasy_SubwayAccess}. 
Finally, we assume an average in-vehicle subway time within the CRZ of 12 minutes, corresponding to roughly half the travel time between Columbus Circle and South Ferry along the subway’s 1 train, assuming spatially uniform trips. % within the CRZ. 

%We summarize the above parameter values in the New York City case study in Table~\ref{tab:calibration_summary_nyc_mfd} in Appendix~\ref{apdx:summary-values-nyc}.

\vspace{-6pt}
\subsection{Free-flow Travel Time Calibration for Bay Area Case Study} \label{subsec:fftt-calibration-bay-area}
\vspace{-5pt}

We calibrate the average free-flow travel time for westbound trips from East Bay to San Francisco using origin–destination (OD) data from~\citet{MTC2015}, which reports distances and OD shares for major East Bay sub-regions traveling to SF (see Table~\ref{tab:freeflow_od}). The remaining share comprises trips from other East Bay locations for which these data are unavailable, so we use the weighted average over the reported origins as a proxy for the full East Bay average. For the origins in Table~\ref{tab:freeflow_od}, we compute free-flow travel times assuming a constant speed of $50$ mph, consistent with the posted speed limit on the SF-Oakland Bay Bridge, and aggregate them using the reported westbound AM OD shares. This yields a weighted average free-flow travel time of \textbf{0.35 hours (21 minutes)}. %, which we use to calibrate $z_C$.

%the free-flow travel times are computed assuming a constant speed of $50$mph, consistent with the posted speed limit on the SF–Oakland Bay Bridge. These travel time estimates are then aggregated using the reported OD shares for westbound AM trips to obtain a weighted average free-flow travel time of \textbf{0.35 hours (21 minutes)}. This weighted average provides a parsimonious and empirically grounded estimate of the free-flow travel time for westbound trips from East Bay into San Francisco, required to calibrate the parameter $z_C$.
\vspace{-5pt}
\begin{table}[h!]
\centering
\caption{\footnotesize \sf Calibration of free-Flow Travel Times for Westbound AM Trips into San Francisco. %Free-flow times are computed assuming a travel speed of 50 mph. OD shares and distances are based on westbound AM trips reported in~\citet{MTC2015}.
}
\label{tab:freeflow_od}
\small
\begin{tabular}{lccc}
\hline
\textbf{Origin} & \textbf{Distance (miles)} & \textbf{Free-Flow Time (hours)} & \textbf{OD Share (\%)} \\
\hline
Oakland        & 12.6 & 0.252 & 29.8 \\
Richmond       & 17.9 & 0.359 & 12.4 \\
Berkeley       & 13.8 & 0.276 & 12.4 \\
Hayward        & 27.2 & 0.544 & 8.6  \\
Walnut Creek   & 25.3 & 0.506 & 5.0  \\
Concord        & 31.1 & 0.622 & 4.8  \\
\hline
\end{tabular} \normalsize
\end{table}
\vspace{-15pt}

%\Needspace{10\baselineskip}

\clearpage

\subsection{Summary of Parameter Values for Bay Area Case Study} \label{apdx:bay_area_summary}
%\vspace{-3pt}

\begin{comment}
\begin{table}[h!]
\centering
\caption{\footnotesize \sf Calibrated bottleneck model parameter values for the SF-Oakland Bay Bridge case study.}
\label{tab:calibration_summary-bottleneck}
\small
\begin{tabular}{ll}
\hline
\textbf{Parameter} & \textbf{Value} \\
\hline
\multicolumn{2}{l}{\textit{Demand, Arrival Rate, and Bottleneck Service Rate}} \\
\hline
$\Lambda$ & $70{,}000$ \\
$\lambda$ & $14{,}000$ users/hr \\
$\mu$ & $9{,}600$ vehicles/hr \\
\hline
\multicolumn{2}{l}{\textit{Components of User Cost of Driving}} \\
\hline
$c_W$ & \$22/hr \\
$e$ & $0.61$ \\
$L$ & $2.4$ \\
Parking fee & \$30 \\
Free-flow travel time & $21$ min \\
$z_C$ & $1.714$ hrs \\
\hline
\multicolumn{2}{l}{\textit{Transit (BART) Parameters to Calibrate $z_T$}} \\
\hline
Fare & \$6.14 \\
Walking time & $20$ min \\
Waiting time & $10$ min \\
On-board BART travel time & $32$ min \\
Discomfort Multiplier ($\eta$) & $[1.5, 18]$ \\
Range of $z_T$ values & [$1.829$ hrs, $18.879$ hrs] \\
\hline
\end{tabular} \normalsize
\end{table}    
\end{comment}

\vspace{-10pt}
\begin{table}[!ht]
\centering
\caption{\footnotesize \sf Calibrated bottleneck model parameter values for the SF--Oakland Bay Bridge case study.}
\label{tab:calibration_summary-bottleneck}

\small

\begin{tabular}{llll}
\toprule
\textbf{Parameter} & \textbf{Value} & \textbf{Parameter} & \textbf{Value} \\
\midrule

\multicolumn{4}{l}{\textbf{Demand, Arrival Rate, and Bottleneck Service Rate}} \\
$\Lambda$ & $70{,}000$
& $\lambda$ & $14{,}000$ users/hr \\
$\mu$ & $9{,}600$ vehicles/hr
& & \\

\midrule
\multicolumn{4}{l}{\textbf{Components of User Cost of Driving}} \\
$c_W$ & \$22/hr
& $e$ & $0.61$ \\
$L$ & $2.4$
& Parking fee & \$30 \\
Free-flow travel time & $21$ min
& $z_C$ & $1.714$ hrs \\

\midrule
\multicolumn{4}{l}{\textbf{Transit (BART) Parameters to Calibrate $z_T$}} \\
Fare & \$6.14
& Walking time & $20$ min \\
Waiting time & $10$ min
& On-board BART travel time & $32$ min \\
Discomfort multiplier $\eta$ & $[1.5,18]$
& Range of $z_T$ & $[1.829,18.879]$ hrs \\

\bottomrule
\end{tabular}

\normalsize
\end{table}

\vspace{-15pt}

\subsection{Calibration of Total Travel Demand for New York Case Study} \label{apdx:calibration-travel-demand-nyc}
\vspace{-5pt}

We calibrate total weekday morning travel demand within NYC's CRZ as the sum of subway and vehicle trips for which at least part of the trip occurs within the zone between 5-10 AM, averaged over weekdays in August 2025.

For subway demand, we classify stations as inside or outside the CRZ and aggregate trips across three OD categories: (i) trips from inside to outside the CRZ, (ii) trips entirely within the CRZ, and (iii) trips from outside to inside the CRZ. These categories capture a majority of subway trips in the zone. Using the subway OD ridership data~\citep{MTA_OD_Ridership2025}, we estimate 684,306 such trips during the morning peak and round this to 700,000 to account conservatively for any remaining trips. For vehicle demand, CRZ entry data~\citep{NY_CongestionRelief_Entries2025} indicate an average of approximately 128,000 vehicle entries between 5-10 AM on weekdays in August 2025, including about 35,000 taxi and for-hire vehicle (FHV) entries. Because these data exclude trips occurring entirely within the CRZ, we use NYC taxi and FHV trip data~\citep{NYC_TLC_TripData}, including yellow taxis, green taxis, and FHVs, apply the same three OD categories used for subway demand. This yields about 50,000 taxi and FHV trips traversing the CRZ during the morning peak, compared with 35,000 entries recorded in the CRZ data. We therefore scale total vehicle entries by the corresponding factor, yielding an adjusted estimate of approximately $128,000 \times \frac{50,000}{35,000} \approx 183,000$ vehicle trips, which we round to 200,000 to conservatively account for remaining trips.

\vspace{-3pt}
\subsection{Calibration of Throughput Maximizing Capacity for New York Case Study} \label{apdx:mu_f_calibration_nyc}
\vspace{-3pt}

We calibrate the throughput-maximizing capacity $\mu_f$ within NYC's CRZ using the same procedure as the vehicle demand calibration in Appendix~\ref{apdx:calibration-travel-demand-nyc}, using peak rather than average demand as a proxy for maximum throughput. CRZ entry data~\citep{NY_CongestionRelief_Entries2025} show a maximum of around 137,000 vehicle entries between 5-10 AM on weekdays in August 2025, including about 36,400 taxi and for-hire vehicle (FHV) entries. As these data exclude trips entirely within the CRZ, we supplement them with NYC taxi and FHV trip data~\citep{NYC_TLC_TripData}. Aggregating across the same three OD categories as in Appendix~\ref{apdx:calibration-travel-demand-nyc} gives 54,600 taxi and FHV trips traversing the CRZ, compared with 36,400 recorded entries. Scaling total entries by this ratio gives approximately $137{,}000\times\frac{54{,}600}{36{,}400}\approx205{,}000$ vehicle trips. Applying the same conservative adjustment used in Appendix~\ref{apdx:calibration-travel-demand-nyc} gives $205{,}000\times\frac{200{,}000}{183{,}000}\approx225{,}000$ trips over the five-hour morning period, corresponding to $\mu_f=45{,}000$ vehicles per hour.

\clearpage

%\vspace{-3pt}
\subsection{Summary of Parameter Values for New York Case Study} \label{apdx:summary-values-nyc}
%\vspace{-10pt}
\begin{comment}
\begin{table}[h!]
\centering
\caption{\footnotesize \sf Calibrated parameter values for the MFD framework in the New York City Case Study}
\label{tab:calibration_summary_nyc_mfd}
\small
\begin{tabular}{ll}
\hline
\textbf{Parameter} & \textbf{Value} \\
\hline
\multicolumn{2}{l}{\textit{Demand and Arrival Rates}} \\
\hline
$\Lambda$ & $900{,}000$ users \\
$\lambda$ & $180{,}000$ users/hr \\
\hline
\multicolumn{2}{l}{\textit{MFD Parameters}} \\
\hline
Jam density & $140$ veh/km/lane \\
Total road length in CRZ & $1{,}000$ lane-km \\
Jam accumulation level ($n_j$) & $\{14,000, 42,000, 70,000, 140{,}000 \}$ vehicles \\
Throughput-maximizing capacity $\mu_f$ & $45{,}000$ veh/hr \\
Average trip distance $D$ & $6$ km \\
\hline
\multicolumn{2}{l}{\textit{Car Travel Cost Parameters}} \\
\hline
$c_W$ & \$40/hr \\
$e$ & $0.61$ \\
$L$ & $2.4$ \\
Parking fee & \$30 \\
Free-flow speed $v_f$ & $40$ km/hr \\
$z_C$ & 0.9 \\
\hline
\multicolumn{2}{l}{\textit{Transit (Subway) Parameters to Calibrate $z_T$}} \\
\hline
Fare & \$3 \\
Walking time & $20$ min \\
Waiting time & $2.5$ min \\
On-board subway time & $12$ min \\
Discomfort multiplier $\eta$ & $[1.5,18]$ \\
Range of $z_T$ values & [$0.938$ hrs, $10.425$ hrs] \\
\hline
\end{tabular} \normalsize
\end{table}    
\end{comment}

\vspace{-12pt}
\begin{table}[!ht]
\centering
\caption{\footnotesize \sf Calibrated parameter values for the MFD framework in the New York City Case Study}
\label{tab:calibration_summary_nyc_mfd}

\small

\begin{tabular}{llll}
\toprule
\textbf{Parameter} & \textbf{Value} & \textbf{Parameter} & \textbf{Value} \\
\midrule

\multicolumn{4}{l}{\textbf{Demand and Arrival Rates}} \\
$\Lambda$ & $900{,}000$ users
& $\lambda$ & $180{,}000$ users/hr \\

\midrule
\multicolumn{4}{l}{\textbf{MFD Parameters}} \\
Jam density & $140$ veh/km/lane
& Road length in CRZ & $1{,}000$ lane-km \\
Jam accumulation ($n_j$) & $\{14,000,42,000,70,000,140,000\}$
& $\mu_f$ & $45{,}000$ veh/hr \\
Average trip distance $D$ & $6$ km
& Free-flow speed $v_f$ & $40$ km/hr \\

\midrule
\multicolumn{4}{l}{\textbf{Car Travel Cost Parameters}} \\
$c_W$ & \$40/hr
& $e$ & $0.61$ \\
$L$ & $2.4$
& Parking fee & \$30 \\
$z_C$ & $0.9$
& & \\

\midrule
\multicolumn{4}{l}{\textbf{Transit (Subway) Parameters to Calibrate $z_T$}} \\
Fare & \$3
& Walking time & $20$ min \\
Waiting time & $2.5$ min
& On-board subway time & $12$ min \\
Discomfort multiplier $\eta$ & $[1.5,18]$
& Range of $z_T$ & $[0.938,10.425]$ hrs \\

\bottomrule
\end{tabular}

\normalsize
\end{table}

\vspace{-15pt}

\subsection{Static vs. Dynamic Tolling Comparisons with Heterogeneous User Types} \label{apdx:heterogeneous-static-v-dynamic}

This section extends our SF-Oakland Bay Bridge numerical study to allow for heterogeneity in users' values of time and schedule delay penalties. We consider four equally sized user types that differ along these two dimensions, with parameter values chosen symmetrically around the corresponding values in Section~\ref{sec:experiments} so that population averages remain unchanged. In this setting, the desired arrival times of users continue to be uniformly distributed over the same interval as in the baseline model with homogeneous cost parameters and are independent of the user type. This specification allows us to isolate the effect of preference heterogeneity on the relative performance of static and dynamic tolling.

Let $h_V$ and $h_S$ denote the proportional spreads in the value of time and schedule delay penalties around their population means. Specifically, the values of time $\alpha\in \{(1-h_V)\bar\alpha,(1+h_V)\bar\alpha\},$ while the earliness and lateness penalties are scaled by $1-h_S$ or $1+h_S$. We preserve the SF-Oakland Bay Bridge benchmark means $\bar\alpha=22$ dollars/hour, $\bar\beta^E=0.61\bar\alpha$, and $\bar\beta^L=2.4\bar\alpha$. 
Thus, $h_V$ and $h_S$ control the degree of heterogeneity while leaving population means fixed.
For each type \(g\), we construct the free-flow car and transit costs using the Bay Bridge calibration in Section~\ref{sec:experiments}, with travel-time components of those costs weighted using the type-specific value of time \(\alpha_g\), yielding costs \(z_{C,g}\) and \(z_{T,g}\).

We first describe the numerical procedure to compute the approximate revenue-optimal static toll and obtain bounds on the dynamic optimal toll revenue and system cost, before presenting the resulting numerical comparisons.

\emph{Numerical Computation of Static Revenue-Optimal Tolls.}
We discretize desired and actual bottleneck crossing times into intervals of length $\delta$. Then, each group $g=(k,j)$ corresponds to preference type $k$ and desired time bin $j$, with mass $q_g$ determined by the uniform desired-time distribution. Let $d_{gs}$ be the schedule delay cost from crossing in slot $s$, and $\Delta_g=z_{T,g}-z_{C,g}$. For any fixed toll vector $p$, we numerically compute an equilibrium via the linear program:
{\setlength{\abovedisplayskip}{2pt}
\setlength{\belowdisplayskip}{2pt}
\begin{equation}
\begin{aligned}
\min_{x,y\geq0}\quad &
\sum_{g,s}\frac{d_{gs}+p_s}{\alpha_g}x_{gs}
+\sum_g\frac{\Delta_g}{\alpha_g}y_g \qquad \text{s.t.}\qquad \sum_s x_{gs}+y_g=q_g \quad \forall g, \quad \sum_g x_{gs}\leq\mu\delta \quad \forall s.
\end{aligned}
\label{eq:heterogeneous-equilibrium-lp}
\end{equation}}
The dual variables on the slot capacity constraints are the equilibrium wait times, and the KKT conditions imply every used car slot (and transit, when used) attains the minimum cost for that group. This LP is the discretized counterpart of the heterogeneous bottleneck equilibrium LP in \citet{AKAMATSU2021152}. For a static toll, we set $p_s\equiv p$, solve \eqref{eq:heterogeneous-equilibrium-lp}, and maximize the revenue $p\sum_{g,s}x_{gs}$ over $p\in[0,\max_g\Delta_g]$. We first evaluate static tolls on a coarse grid with \$0.10 increments and then locally around the best toll with \$0.005 increments. The baseline discretization uses 15-minute desired arrival and crossing time intervals, with a 10-minute discretization yielding nearly identical results.
%which we repeat with 10-minute intervals as a sensitivity check, yielding nearly identical results. 
The bounds reported below nevertheless apply to the discretized model rather than to the unrestricted continuous-time problem. %; results using a 10-minute grid are nearly identical, providing a sensitivity check on grid resolution.

\emph{Dynamic revenue and system-cost benchmarks.}
Rather than solve the unrestricted heterogeneous dynamic-toll bi-level problem, we bound its value. If group $g$ uses a car in slot $s$ in any dynamic toll equilibrium, $p_s+\alpha_g w_s\leq \Delta_g-d_{gs}=:v_{gs},$ where $w_s$ is the waiting time in slot $s$ and is the dual variable of that slot's capacity constraint in the above LP. Since $w_s\geq0$, $p_s\leq v_{gs}$ when $x_{gs}>0$. We then compute an upper bound on dynamic toll revenue by maximizing revenue subject only to bottleneck capacity, available mass of each user group, and that any user assigned to a crossing time is willing to pay its toll. As we do not require the resulting assignments to satisfy all cross-time equilibrium conditions or queue dynamics, this problem can only overstate the revenue achievable by a dynamic toll. On the finite grid, it suffices to consider tolls in each slot $s$ equal to one of the positive willingness-to-pay values $(v_{gs})_g$. This yields a MILP whose optimal value $\overline R_d$ upper bounds the dynamic toll revenue, i.e., $R_d^* \leq \overline R_d.$ Thus, $R_s^*/R_d^*\geq R_s^*/\overline R_d$.

For system cost, we compute a first-best allocation by minimizing free-flow travel and schedule delay costs subject to group mass and bottleneck capacity constraints. As this allocation may not be implementable by an anonymous dynamic toll, it lower bounds the minimum possible system cost ($SC^{\mathrm{FB}} \!\leq \! SC_d^*$); thus, $SC_s/SC^{\mathrm{FB}}$ upper bounds $SC_s/SC_d^*$.

\textbf{Results.}
Tables~\ref{tab:heterogeneous-bottleneck} and~\ref{tab:heterogeneous-bottleneck-50} compare the static revenue optimal toll to its dynamic benchmarks for moderate and larger preference heterogeneity, respectively. For $h_V=h_S=0.20$, the static-to-dynamic revenue bound decreases from $100\%$ at $\eta=1.5$ to $71.5\%$ at $\eta=5$, while the corresponding system cost ratio remains between $1.000$ and $1.434$. Interestingly, the reported guarantees improve with greater heterogeneity, i.e., $h_V=h_S=0.50$, where the revenue bound ranges from $100\%$ to $88.5\%$, and the system cost ratio remains below $1.202$. A possible explanation is that greater preference heterogeneity may create stronger separation in users' willingness to pay, allowing a static toll to screen lower value users while retaining higher value users. That said, as the reported revenue and system cost bounds are, respectively, lower and upper bounds evaluated on a discretized grid, these comparisons do not by themselves establish that greater heterogeneity narrows the true performance gap.
Overall, these results reinforce the findings in Section~\ref{subsec:numerical-results} that static tolling retains strong performance relative to dynamic tolling, even under preference heterogeneity across users.
%when users differ substantially in their values of time and schedule delay penalties.
%Thus, there is less scope for dynamic tolls to obtain additional revenue through variation in prices across time. 

%the results indicate that the relative performance of static tolling remains strong under substantial heterogeneity in user preferences.

\begin{table}[t]
\centering

\begin{minipage}[t]{0.48\textwidth}
\centering
\captionof{table}{\footnotesize \sf Four-type robustness %check for the SF--Oakland Bay Bridge study
with $h_V=h_S=0.20$.}
\label{tab:heterogeneous-bottleneck}
\small
\begin{tabular}{cccc}
\toprule
$\eta$ & $\tau_s^*$ & $R_s^*/\overline R_d$ & $SC_s/SC^{\mathrm{FB}}$ \\
\midrule
1.5 & \$7.82   & 100.0\% & 1.000 \\
%2.1 & \$24.19  & 88.0\%  & 1.041 \\
2.5 & \$35.10  & 83.7\%  & 1.081 \\
3.0 & \$48.74  & 80.2\%  & 1.137 \\
4.0 & \$76.02  & 75.3\%  & 1.268 \\
5.0 & \$103.30 & 71.5\%  & 1.434 \\
\bottomrule
\end{tabular}
\end{minipage}
\hfill
\begin{minipage}[t]{0.48\textwidth}
\centering
\captionof{table}{\footnotesize \sf Four-type robustness %check for the SF--Oakland Bay Bridge study 
with $h_V=h_S=0.50$.}
\label{tab:heterogeneous-bottleneck-50}
\small
\begin{tabular}{cccc}
\toprule
$\eta$ & $\tau_s^*$ & $R_s^*/\overline R_d$ & $SC_s/SC^{\mathrm{FB}}$ \\
\midrule
1.5 & \$15.74  & 100.0\% & 1.000 \\
%2.1 & \$36.20  & 100.0\% & 1.000 \\
2.5 & \$49.84  & 99.5\%  & 1.003 \\
3.0 & \$66.89  & 96.2\%  & 1.033 \\
4.0 & \$100.99 & 91.7\%  & 1.108 \\
5.0 & \$135.09 & 88.5\%  & 1.202 \\
\bottomrule
\end{tabular}
\end{minipage}

\end{table}

\vspace{-3pt}
\section{Additional Theoretical Results} \label{apdx:sco-static-toll-overall}

\vspace{-3pt}

\subsection{Robustness to General Transit Crowding Costs} \label{apdx:general-transit-cost}

\vspace{-3pt}

We extend the bottleneck model in Section~\ref{sec:revenue-optimal-tolling-classical-bottleneck} to allow transit costs to depend on ridership. We first show that the one-half revenue guarantee in Theorem~\ref{thm:rev-comp-static-v-dynamic} holds for any non-decreasing transit cost function. We then establish a general bound on the ratio between the system cost induced by the static revenue-optimal toll and the minimum achievable system cost, expressed in terms of the shape of the transit cost function. As a special case, linear transit costs yield a closed-form factor-three guarantee, which reduces to the factor-two bound in Theorem~\ref{thm:sc-comp-static-v-dynamic} when transit costs are constant. %independent of ridership.

%As a special case, for linear transit costs, this bound admits a closed-form characterization that recovers the corresponding constant-factor guarantee in Theorem~\ref{thm:sc-comp-static-v-dynamic}.

We first introduce some notation. Consider the congested bottleneck regime $\mu<\lambda$, define $q:=1-\frac{\mu}{\lambda},$ $H:=\Lambda q$, and recall that $T_C:=\frac{\Lambda eL}{\mu(e+L)}.$ Thus, $H$ is the mass of transit users required to eliminate the road queue and $T_C$ is the maximum bottleneck cost in the car-only equilibrium. Moreover, let $z_T:[0,\Lambda]\rightarrow\mathbb{R}_+$ denote the transit cost as a function of transit ridership $n$, and assume that $z_T(n)\geq z_C$ for all $n$ and that $z_T(\cdot)$ is non-decreasing. 

To express the revenue functions compactly, first consider transit ridership up to the queue-eliminating level $H$. In this case, transit ridership is $n=Hf$ for $f\in[0,1]$ and we define the normalized transit-car cost gap
{\setlength{\abovedisplayskip}{0pt}
\setlength{\belowdisplayskip}{0pt}
\begin{equation}
g(f):=\frac{z_T(Hf)-z_C}{T_C}.
\label{eq:general-g}
\end{equation}}
Thus, $f$ denotes transit ridership relative to the level $H$ required to eliminate queuing, while $g(f)$ measures the transit-car cost difference relative to the maximum car-only bottleneck cost. For a given $f\in[0,1]$, transit ridership is $Hf=\Lambda qf$, so the remaining mass $\Lambda(1-qf)$ travels by car. Under a static toll, the equilibrium bottleneck cost is $T_C(1-f)$, and hence the toll that supports transit ridership $Hf$ is $T_C[g(f)-(1-f)]$. Under dynamic tolling, the revenue characterization in the proof of Theorem~\ref{thm:rev-opt-dynamic-tolls} applies after replacing the fixed transit--car cost gap $z_T-z_C$ by the ridership-dependent gap $z_T(Hf)-z_C=T_Cg(f)$. Thus, normalizing by $\Lambda T_C$, the static and dynamic revenues are
{\setlength{\abovedisplayskip}{0pt}
\setlength{\belowdisplayskip}{0pt}
\begin{equation}
S_g(f)
=
\bigl[g(f)-(1-f)\bigr](1-qf),
\qquad
D_g(f)
=
g(f)(1-qf)-\frac{1}{2}(1-f)^2.
\label{eq:general-revenues}
\end{equation}}
Since a static toll of $T_C[g(f)-(1-f)]$ induces transit ridership $Hf$, this ridership level can be induced by a non-negative toll if and only if
{\setlength{\abovedisplayskip}{0pt}
\setlength{\belowdisplayskip}{0pt}
\begin{equation} \label{eq:feasibility-g}
g(f)\geq 1-f.
\end{equation}}
For dynamic tolling, substituting the ridership-dependent transit-car cost gap $z_T(Hf)-z_C=T_Cg(f)$ into the lower-bound constraint on $f$ in Equation~\eqref{eq:constrained-dynamic-rev-opt}, derived in the proof of Theorem~\ref{thm:rev-opt-dynamic-tolls}, yields the same feasibility condition.

When transit ridership exceeds $H$, the mass of car users is small enough that all drivers can travel at their desired times without delay. Let $y:=n/\Lambda\in[q,1]$ denote the fraction of users taking transit and let $\widetilde g(y):=\frac{z_T(\Lambda y)-z_C}{T_C}.$ Since car users incur only the free-flow cost $z_C$, the maximum toll supporting transit ridership $\Lambda y$ is $z_T(\Lambda y)-z_C=T_C\widetilde g(y)$. Moreover, as there is no time-varying congestion or schedule delay cost to exploit, dynamic tolling provides no additional revenue advantage over a static toll. Thus, static and dynamic tolling coincide, and their normalized revenue is
{\setlength{\abovedisplayskip}{0pt}
\setlength{\belowdisplayskip}{0pt}
\begin{equation}
U(y)=\widetilde g(y)(1-y).
\label{eq:underutilized-revenue}
\end{equation}}
%Hence, the comparison below accounts for both transit ridership below and above the queue-eliminating level $H$.

We now show that the one-half revenue guarantee continues to hold for any non-decreasing transit cost function. %, as elucidated by the following theorem.

\vspace{-3pt}
\begin{theorem}%[Revenue Guarantee under General Transit Costs] 
\label{thm:general-transit-revenue}
Suppose $\mu<\lambda$ and let $z_T(\cdot)$ be a non-decreasing transit cost function that satisfies $z_T(n)\geq z_C$ for all $n\in[0,\Lambda]$. Further, let $R_s^*$ and $R_d^*$ denote the optimal revenues under static and dynamic tolling. Then, $R_s^*\ge \frac{1}{2}R_d^*$.
\end{theorem}
\vspace{-3pt}

\begin{proof}
It suffices to establish this guarantee for transit ridership at most $H$, since when transit ridership exceeds $H$, static and dynamic tolling attain the same revenue in Equation~\eqref{eq:underutilized-revenue}. Then, fix any $f$ inducing a transit ridership $Hf$, i.e., the feasibility condition $g(f)\geq 1-f$ in Equation~\eqref{eq:feasibility-g} is satisfied, and set $\widehat f:=\frac{1+f}{2},$ $a:=1-f,$ and $\rho:=1-q=\frac{\mu}{\lambda}.$

Since $\widehat f\ge f$ and $g$ is non-decreasing, $g(\widehat f)\ge g(f)\ge a\ge \frac{a}{2}=1-\widehat f$; hence, by Equation~\eqref{eq:feasibility-g}, the transit ridership level $H\widehat f$ can be induced by a non-negative static road toll. Then, using 
{\setlength{\abovedisplayskip}{1pt}
\setlength{\belowdisplayskip}{1pt}
\eqref{eq:general-revenues}
\begin{align}
S_g(\widehat f)-\frac{1}{2}D_g(f)
&=\bigl[g(\widehat f)-g(f)\bigr](1-q\widehat f)
  +\frac{\rho}{4}\bigl[a^2-2a+2g(f)\bigr] \ge 0,
\end{align}}
where the first term is non-negative by monotonicity of $g(\cdot)$ and the second is non-negative because $g(f)\ge a$. Thus, for any transit ridership level $Hf\leq H$ under dynamic tolling, the construction above yields a transit ridership level $H\widehat f$ attainable under a non-negative static toll that generates at least one-half as much revenue. Applying this construction to a dynamic revenue maximizer establishes the guarantee whenever optimal transit ridership is at most $H$; if it exceeds $H$, static and dynamic tolling attain the same revenue by Equation~\eqref{eq:underutilized-revenue}. Hence, $R_s^*\geq \frac{1}{2}R_d^*$. 
\end{proof}

We next turn to system cost. While the revenue guarantee requires only that transit costs be non-decreasing, the corresponding system cost bound comparing the system cost of the static revenue-optimal toll to the minimum achievable system cost depends on the shape of the transit cost function, since the minimum achievable system cost can vary with how rapidly transit costs increase with ridership. To this end, suppose additionally that $z_T(\cdot)$ is differentiable and convex on $[0,H]$, and define $h_g(f):=f g(f)+\frac{1}{2}(1-f)^2,$ $m_g:=\min_{f\in[0,1]} h_g(f),$ and $G:=g(1).$ %The following theorem provides the resulting system-cost guarantee.

\vspace{-5pt}
\begin{theorem}[System Cost Guarantee with General Transit Costs]
\label{thm:general-transit-sc}
Suppose $\mu<\lambda$ and let $z_T:[0,\Lambda]\rightarrow\mathbb{R}_+$ be non-decreasing with $z_T(n)\geq z_C$ for all $n\in[0,\Lambda]$. Suppose additionally that $z_T(\cdot)$ is differentiable and convex on $[0,H]$, $G=g(1)\leq1$, and $U(y)< U(q)$ for all $y\in[q,1]$. Considering non-negative road tolls, let $SC_s$ denote the system cost induced by revenue-optimal static toll, and let $SC^*$ denote the minimum achievable system cost. Then,
{\setlength{\abovedisplayskip}{2pt}
\setlength{\belowdisplayskip}{2pt}
\begin{equation}
SC_s \leq z_C\Lambda+\Lambda T_C q G,
\qquad
\frac{SC_s}{SC^*} \leq \Gamma_g:=\frac{G}{m_g} = \frac{G}{\min_{f\in[0,1]}\left\{f g(f)+\frac{1}{2}(1-f)^2\right\}},
\label{eq:sc-ratio-transit-costs}
\end{equation}}
for $G>0$. When $G=0$, the corresponding ratio is one. If the transit cost function is linear, $z_T(n)=z_T(0)+\beta n$ with $\beta\geq0$, and $\delta:=(z_T(0)-z_C)/T_C$ and $\eta:=\beta H/T_C$, then $g(f)=\delta+\eta f$ and $\frac{SC_s}{SC^*} \leq \Gamma_g = \frac{(\delta+\eta)(1+2\eta)}{\delta+\eta-\delta^2/2} \leq 3.$
\end{theorem}
\vspace{-3pt}

\begin{proof}
For any $f\in[0,1]$, transit ridership is $Hf=\Lambda qf$. Relative to the free-flow cost $z_C\Lambda$, transit users incur an additional aggregate cost $Hf\,[z_T(Hf)-z_C] = \Lambda T_C q\, f g(f).$ Following the proof of Theorem~\ref{thm:sc-comp-static-v-dynamic}, the remaining car users incur aggregate schedule delay costs equal to $\Lambda T_C q\,\frac{1}{2}(1-f)^2.$ Hence, for any \(f\in[0,1]\), the corresponding system cost is at least $z_C\Lambda + \Lambda T_C q \left[ f g(f)+\frac{1}{2}(1-f)^2 \right] \geq z_C\Lambda+\Lambda T_C q\,m_g.$ For transit ridership above \(H\), let \(y=n/\Lambda\in[q,1]\). Then, the remaining car users can travel without delay, so the system cost is $z_C\Lambda+\Lambda y\bigl[z_T(\Lambda y)-z_C\bigr] \geq z_C\Lambda+\Lambda T_C q\,G \geq z_C\Lambda+\Lambda T_C q\,m_g,$ by the monotonicity of \(z_T(\cdot)\) and \(m_g\leq G\). Thus, $SC^* \geq  z_C\Lambda+\Lambda T_C q\,m_g.$
Next, the condition $U(y)< U(q)$ for all $y\in[q,1]$ rules out revenue-optimal transit ridership above $H$.

%ensures that increasing transit ridership beyond $H$ cannot increase revenue, so a revenue-optimal policy can be selected with transit ridership at most $H$. 

%Hence, the minimum achievable system cost associated with a given $f$ satisfies: $SC^* \geq z_C\Lambda + \Lambda T_C q \left[ f g(f)+\frac{1}{2}(1-f)^2 \right] =z_C\Lambda+\Lambda T_C q\,m_g.$

For the static revenue-optimal toll, the same system cost analysis as in Theorem~\ref{thm:sc-comp-static-v-dynamic} gives $SC_s(f)=z_C\Lambda+\Lambda T_C C_s(f)$ and $C_s(f):=q f g(f)+(1-q)f(1-f)+\frac{1+q}{2}(1-f)^2.$ At $f=1$, $C_s(1)=qG$. Otherwise, let $f<1$ be an interior static revenue maximizer and let $a:=1-f,$ $c:=1-qf,$ and $u:=g(f)-a\ge0.$ Then, the first-order condition for $S_g(f)=uc$ is $\bigl[g'(f)+1\bigr]c=qu.$ Also, convexity of $z_T(\cdot)$ implies the convexity of $g(\cdot)$, yielding $G\ge g(f)+a g'(f)$. 

Next, since $z_T(\cdot)$ is non-decreasing, $g'(f)\ge0$. Moreover, since $c=(1-q)+qa$, we have
{\setlength{\abovedisplayskip}{1pt}
\setlength{\belowdisplayskip}{1pt}
\begin{align*}
qG-C_s(f)
&\ge a\left[qu\left(1+\frac{q}{c}\right)-1+\frac{(1-q)a}{2}\right] \ge a\left[c+q-1+\frac{(1-q)a}{2}\right]
 =\frac{1+q}{2}a^2\ge0.
\end{align*}}
Finally, the lower bound on $SC^*$ and $m_g\le h_g(1)=G$ imply $\frac{z_C\Lambda+\Lambda T_C qG}{z_C\Lambda+\Lambda T_C q m_g} \le \frac{G}{m_g},$ which proves \eqref{eq:sc-ratio-transit-costs}.

For linear transit costs, suppose $z_T(n)=z_T(0)+\beta n,$ with $\beta\ge0.$ Then $g(f)=\delta+\eta f$, where $\delta:=\frac{z_T(0)-z_C}{T_C},$ $\eta:=\frac{\beta H}{T_C},$ and $G=\delta+\eta.$ Then, the value of $f$ minimizing $m_g$ is $\frac{1-\delta}{1+2\eta}$ and the resulting $m_g =\frac{1}{2}-\frac{(1-\delta)^2}{2(1+2\eta)} =\frac{\delta+\eta-\delta^2/2}{1+2\eta}$. Hence, $\Gamma_{\mathrm{lin}} =\frac{(\delta+\eta)(1+2\eta)}{\delta+\eta-\delta^2/2} \leq 3.$ This bound reduces to the factor-two bound in Theorem~\ref{thm:sc-comp-static-v-dynamic} with constant costs (i.e., $\eta = 0$). %\Halmos
\end{proof}

The assumptions required for this result are natural. First, convexity is a standard way to capture increasing transit crowding costs, while $G\leq1$ parallels the condition in Theorem~\ref{thm:sc-comp-static-v-dynamic} that transit remains sufficiently attractive compared to car travel at free flow. The condition $U(y)< U(q)$ simply rules revenue-optimal policies that shift more than $H$ users to transit and leaves road capacity under-utilized to raise toll revenue. Under these conditions, Theorem~\ref{thm:general-transit-sc} shows that ridership-dependent transit costs can weaken the factor-two system-cost guarantee, but the loss remains bounded: linear transit costs yield a factor-$3$ guarantee, with the factor-$2$ bound recovered when transit costs are ridership-independent.

\vspace{-4pt}
\subsection{Derivation of System Cost Optimal Static Toll for Bottleneck Model} \label{apdx:sco-static-toll}
\vspace{-3pt}

We derive the system-cost-optimal static toll in the bottleneck model with an outside option. To this end, we first present the expression for the total system cost as a function of the static toll $\tau$, which is the sum of the cost of using transit and the generalized cost of using the car plus scheduling and queuing delays. In presenting this expression, we focus on the case when $\tau \in [\max \{z_T - z_C - T_C, 0 \}, z_T - z_C]$ and $z_T - z_C \geq 0$. Note that in this regime, we have a mixed-mode equilibrium as characterized in Proposition~\ref{prop:user-eq-two-modes}. Specifically, letting $\Bar{w} = z_T - z_C - \tau$, we have: (i) $N_e = \frac{\mu \Bar{w}}{e}$ early arrivals, (ii) $N_L = \frac{\mu \Bar{w}}{L}$ late arrivals, (iii) $N_o = (1-\frac{\Bar{w}}{T_C})\frac{\Lambda}{\lambda}\mu$ on-time car arrivals, and (iv) $N_T = (1-\frac{\Bar{w}}{T_C})\Lambda \left( 1 - \frac{\mu}{\lambda}\right)$ transit arrivals.

Then, we obtain the following expression for the total system cost as a function of the static toll $\tau$:
{\setlength{\abovedisplayskip}{2pt}
\setlength{\belowdisplayskip}{2pt}
\begin{align*}
    SC(\tau) &= z_T \left[ \Lambda - \frac{\mu (e+L)(z_T - z_C - \tau)}{eL} \right] \left( 1 - \frac{\mu}{\lambda}\right) + z_C \mu (z_T - z_C - \tau) \frac{(e+L)}{eL} + z_C \left[ \Lambda - \frac{\mu (e+L)(z_T - z_C - \tau)}{eL} \right] \frac{\mu}{\lambda} \\
    &+ \frac{\mu (z_T - z_C - \tau)^2}{2} \frac{(e+L)}{eL} \left( 1 - \frac{\mu}{\lambda} \right) + (z_T - z_C - \tau) \left( \Lambda - \frac{\mu (e+L) (z_T - z_C - \tau)}{eL} \right) \frac{\mu}{\lambda} + \frac{\mu (z_T - z_C - \tau)^2}{2} \frac{(e+L)}{eL}.
\end{align*}}
To compute the optimal uniform toll, first note: $SC'(\tau) = \tau \frac{\mu (e+L)}{eL} \left( 2 - 3 \frac{\mu}{\lambda} \right) - \frac{\Lambda \mu}{\lambda} + (z_T - z_C) \frac{\mu (e+L)}{eL} \left( \frac{2 \mu}{\lambda} - 1 \right).$

%we first compute the derivative $SC'(\tau)$ of the total system cost, which is:
%{\setlength{\abovedisplayskip}{2pt}
%\setlength{\belowdisplayskip}{2pt}
%\begin{align*}
%    SC'(\tau) = \tau \frac{\mu (e+L)}{eL} \left( 2 - 3 \frac{\mu}{\lambda} \right) - \frac{\Lambda \mu}{\lambda} + (z_T - z_C) \frac{\mu (e+L)}{eL} \left( \frac{2 \mu}{\lambda} - 1 \right).  
%\end{align*}}

We now compute the system-cost-optimal static toll in three regimes: (i) $\frac{\mu}{\lambda} = \frac{2}{3}$, (ii) $\frac{\mu}{\lambda} > \frac{2}{3}$, and (iii) $\frac{\mu}{\lambda} < \frac{2}{3}$. %Note that the toll $\tau \in [0, z_T - z_C]$.

\textbf{Case (i):} When $\frac{\mu}{\lambda} = \frac{2}{3}$, the total system cost is linear with $SC'(\tau) = - \frac{2\Lambda}{3} + \frac{1}{3}(z_T - z_C) \frac{\mu (e+L)}{eL}$. Thus, we have:
{\setlength{\abovedisplayskip}{2pt}
\setlength{\belowdisplayskip}{2pt}
\begin{align}
\tau^* = 
    \begin{cases}
    \max \{ 0, z_T - z_C - T_C \}, & \text{if } z_T - z_C > \frac{2 \Lambda e L}{\mu(e+L)}  \\
    z_T - z_C, & \text{if } z_T - z_C < \frac{2 \Lambda e L}{\mu(e+L)} \\
    [\max \{ 0, z_T - z_C - T_C \}, z_T - z_C], & \text{if } z_T - z_C = \frac{2 \Lambda e L}{\mu(e+L)}.
\end{cases}
\end{align}}
Note that in the regime when $z_T - z_C > \frac{2 \Lambda e L}{\mu(e+L)} = 2 T_C$, it follows that $\max \{ 0, z_T - z_C - T_C \} = z_T - z_C - T_C$.

\textbf{Case (ii):} When $\frac{\mu}{\lambda} > \frac{2}{3}$, the total system cost is concave over the interval \([\max\{0, z_T - z_C -T_C\},z_T - z_C]\); hence, the minimum occurs at one of the end-points. When $z_T - z_C \leq T_C$, $SC(z_T - z_C) - SC(0) = \int_{0}^{z_T - z_C} SC'(\tau) d \tau = \frac{z_T - z_C}{2} (SC'(0) + SC'(z_T - z_C)) = \frac{z_T - z_C}{2} \frac{\mu}{\lambda} \left( (z_T - z_C) \mu \frac{e+L}{eL} - 2 \Lambda \right) < 0$, so $\tau^* = z_T - z_C$. When $z_T - z_C > T_C$, $SC(z_T - z_C) - SC(z_T - z_C -T_C) = T_C SC'(z_T - z_C - \frac{T_C}{2}) = \frac{\Lambda}{2} \left[ 2\left(1-\frac{\mu}{\lambda}\right) (z_T - z_C) - \left(2-\frac{\mu}{\lambda}\right)T_C \right].$ Let $\Hat{\Delta} = T_C\frac{2-\frac{\mu}{\lambda}}
{2\left(1-\frac{\mu}{\lambda}\right)}$. Then:

%{\setlength{\abovedisplayskip}{1pt}
%\setlength{\belowdisplayskip}{1pt}
\begin{align}
\tau^* = 
    \begin{cases}
    z_T - z_C - T_C, & \text{if } z_T-z_C > \Hat{\Delta}, \\
    z_T - z_C, & \text{if } z_T-z_C < \Hat{\Delta}, \\
    [z_T - z_C - T_C, z_T - z_C], & \text{if } z_T-z_C = \Hat{\Delta}.
\end{cases}
\end{align}%}

%; hence, the minimum must occur at one of the end-points $0$ or $z_T - z_C$. Leveraging the fundamental theorem of calculus and the linearity of the derivative $SC'(\tau)$, it follows that: $SC(z_T - z_C) - SC(0) = \int_{0}^{z_T - z_C} SC'(\tau) d \tau = \frac{z_T - z_C}{2} (SC'(0) + SC'(z_T - z_C)) = \frac{z_T - z_C}{2} \frac{\mu}{\lambda} \left( (z_T - z_C) \mu \frac{e+L}{eL} - 2 \Lambda \right).$ Thus:

\textbf{Case (iii):} In the regime when $\frac{\mu}{\lambda} < \frac{2}{3}$, the total system cost is convex and hence there is a possibility that the minimum occurs between the range $(0, z_T - z_C)$. In this case, note that the unique unconstrained minimizer satisfies $SC'(\underline{\tau}) = 0$, i.e., $\underline{\tau} = \frac{\frac{\Lambda \mu el}{\lambda \mu (e+L)} + (z_T - z_C) \left(1 - \frac{2 \mu}{\lambda} \right))}{2 - 3 \frac{\mu}{\lambda}}$. We have that if $\frac{\Lambda \mu el}{\lambda \mu (e+L)} \leq (z_T - z_C)(\frac{2\mu}{\lambda} - 1)$, then $\underline{\tau} \leq 0$, i.e., the minimum occurs at $0$, if $\frac{\Lambda \mu el}{\lambda \mu (e+L)} > (z_T - z_C) (1 - \frac{\mu}{\lambda})$, then $\underline{\tau} \geq z_T - z_C$. Then, we have:
{\setlength{\abovedisplayskip}{2pt}
\setlength{\belowdisplayskip}{2pt}
\begin{align}
\tau^* = 
    \begin{cases}
    \max \{ 0, z_T - z_C - T_C \}, & \text{if } \frac{\Lambda \mu el}{\lambda \mu (e+L)} \leq (z_T - z_C)(\frac{2\mu}{\lambda} - 1) \\
    z_T - z_C, & \text{if } \frac{\Lambda \mu el}{\lambda \mu (e+L)} \geq (z_T - z_C) (1 - \frac{\mu}{\lambda}) \\
    \max \left\{ \frac{\frac{\Lambda \mu el}{\lambda \mu (e+L)} + (z_T - z_C) \left(1 - \frac{2 \mu}{\lambda} \right))}{2 - 3 \frac{\mu}{\lambda}}, z_T - z_C - T_C \right\}, & \text{otherwise}.
\end{cases}
\end{align}}

\subsection{Derivation of Total System Cost Under Static Tolling for Triangular MFD} \label{apdx:derivation-tsc-triangular-mfd}
\vspace{-3pt}
Under the quasi-static MFD representation in Section~\ref{sec:mfd}, we derive the total system cost of any static toll $\tau$ under a triangular MFD, which consists of the (i) cost of using transit, (ii) cost of using a car at free-flow, (iii) waiting time or queuing delays when using the car, and (iv) schedule delays. We now derive expressions for each of these four terms in the regime when $\tau \geq \underline{\tau}$, where $\underline{\tau} \geq 0$ is the minimum toll at which some user is indifferent between car and transit. 

\textbf{Transit Costs:} We have the following expression for the cost of using transit:
\begin{align*}
    C_T(\tau) &= z_T \left(\Lambda - \int_{t_A}^{t_D} \mu(n(t)) dt \right), \\
    &= z_T \Big[ \Lambda - \frac{\Lambda}{\lambda} \frac{n_j}{\frac{n_j}{\mu_f} + z_T - z_C - \tau} - n_j \frac{e+L}{eL} \ln \left( 1 + \frac{(z_T - z_C - \tau) \mu_f}{n_j} \right) \left( 1 - \frac{\frac{n_j}{\frac{n_j}{\mu_f} + z_T - z_C - \tau}}{\lambda} \right) \Big],
\end{align*} 
which follows from the analysis in Theorem~\ref{thm:rev-static-tolls-mfd}.

\textbf{Free-flow Car Costs:} We have the following expression for the cost of using car:
{\setlength{\abovedisplayskip}{0pt}
\setlength{\belowdisplayskip}{1pt}
\begin{align} \label{eq:car-cost-ff}
    C_F(\tau) &= z_C \int_{t_A}^{t_D} \mu(n(t)) dt, \nonumber \\ 
    &= z_C \left[ \frac{\Lambda}{\lambda} \frac{n_j}{\frac{n_j}{\mu_f} + z_T - z_C - \tau} + n_j \frac{e+L}{eL} \ln \left( 1 + \frac{(z_T - z_C - \tau) \mu_f}{n_j} \right) \left( 1 - \frac{\frac{n_j}{\frac{n_j}{\mu_f} + z_T - z_C - \tau}}{\lambda} \right) \right].
\end{align}}

\textbf{Waiting and Queuing Delays:} The total waiting and queuing delays are given by the sum of the corresponding delays for the on-time car users, early car users, and late car users. For on-time car users, the queuing delay satisfies: % the following expression for the queuing delay:
{\setlength{\abovedisplayskip}{1pt}
\setlength{\belowdisplayskip}{1pt}
\begin{align} \label{eq:on-time-users}
    C_Q^O(\tau) &= (t_C - t_B) \mu_{\tau} (z_T - z_C - \tau), \nonumber \\
    &= \left( \frac{\Lambda}{\lambda} - \frac{1}{\lambda} n_j \frac{e+L}{eL} \ln \left( 1 + \frac{(z_T - z_C - \tau) \mu_f}{n_j} \right) \right) \frac{n_j}{\frac{n_j}{\mu_f} + z_T - z_C - \tau} (z_T - z_C - \tau)
\end{align}}
For the early car users, we have the following expression for the queuing delay:
{\setlength{\abovedisplayskip}{1pt}
\setlength{\belowdisplayskip}{1pt}
\begin{align} \label{eq:early-queuing}
    C_Q^E(\tau) &= \int_{t_A}^{t_B} \mu(n(t)) w(n(t)) dt = \int_{t_A}^{t_B} \mu(n(t)) \left( \frac{n(t)}{\mu(n(t))} - \frac{n_c}{\mu_f} \right) dt = \int_{t_A}^{t_B} n_j \left( 1 - \frac{\mu(n(t))}{\mu_f} \right) dt, \nonumber \\
    &= \int_{0}^{t_B - t_A} n_j \left( 1 - \frac{\frac{n_j}{\frac{n_j}{\mu_f} + (z_T - z_C - \tau - e \Delta )}}{\mu_f} \right)  d \Delta = \int_{0}^{t_B - t_A} n_j \frac{z_T - z_C - \tau - e \Delta}{\frac{n_j}{\mu_f} + z_T - z_C - \tau - e \Delta} d \Delta.
\end{align}}
Analogously, the queuing delay of the late car users is: $C_Q^L(\tau) = \int_{0}^{t_D - t_C} n_j \frac{z_T - z_C - \tau - L \Delta}{\frac{n_j}{\mu_f} + z_T - z_C - \tau - L \Delta} d \Delta.$

\begin{comment}
\paragraph{Schedule Delay:} We first derive the total earliness delay in the system. To this end, for a user that exits the system at $t$ with a desired exit time $t^*(t)$, their earliness delay is given by: $t^*(t) - t = t_1 + \frac{t - t_A}{t_B - t_A} (t_B - t_1) - t = \frac{(t_1 - t_A)(t_B - t)}{t_B - t_A}.$

We now use this expression to derive the total earliness delay, which is given by:
{\setlength{\abovedisplayskip}{2pt}
\setlength{\belowdisplayskip}{2pt}
\begin{align*}
    C_S^E(\tau) &= e \int_{t_A}^{t_B} \mu(n(t)) \frac{(t_1 - t_A)(t_B - t)}{t_B - t_A} dt = e \int_{t_A}^{t_B} \frac{n_j}{\frac{n_j}{\mu_f} + (z_T - z_C - \tau - e (t_B - t) )} \frac{(t_1 - t_A)(t_B - t)}{t_B - t_A} dt, \nonumber \\
    &= \frac{n_j}{e} \left( z_T - z_C - \tau - \frac{n_j}{\lambda} \ln \left( 1 + \frac{(z_T - z_C - \tau) \mu_f}{n_j} \right) \right) \left[ 1 - \frac{n_j}{\mu_f(z_T - z_C - \tau)} \ln \left( 1 + \frac{(z_T - z_C - \tau) \mu_f}{n_j} \right) \right].
\end{align*}}
%where the final equality is obtained by solving the integral.

Analogously, the total lateness delay is given by:
{\setlength{\abovedisplayskip}{2pt}
\setlength{\belowdisplayskip}{2pt}
\begin{align*}
    C_S^L(\tau) = \frac{n_j}{L} \left( z_T - z_C - \tau - \frac{n_j}{\lambda} \ln \left( 1 + \frac{(z_T - z_C - \tau) \mu_f}{n_j} \right) \right) \left[ 1 - \frac{n_j}{\mu_f(z_T - z_C - \tau)} \ln \left( 1 + \frac{(z_T - z_C - \tau) \mu_f}{n_j} \right) \right].
\end{align*}}    
\end{comment}

\textbf{Schedule Delay:}
For an early car user exiting the system at time $t\in[t_A,t_B]$, ordering by desired crossing times and flow conservation imply $t^*(t) = t_1+\frac{1}{\lambda}\int_{t_A}^{t}\mu(n(s))\,ds.$ Since $t^*(t_B)=t_B$, the earliness delay can equivalently be written as $t^*(t)-t = t_B-t -\frac{1}{\lambda}\int_t^{t_B}\mu(n(s))\,ds.$ %Moreover, the waiting time on $[t_A,t_B]$ is $w(t)=e(t-t_A)$, so $\mu(n(t)) = \frac{n_j}{ \frac{n_j}{\mu_f}+e(t-t_A)}.$ 
Then, integration yields
{\setlength{\abovedisplayskip}{1pt}
\setlength{\belowdisplayskip}{1pt}
\begin{align*}
C_S^E(\tau)
=
\frac{n_j}{e}
\left[
\left(
\frac{n_j}{\mu_f}
+
z_T-z_C-\tau
\right)
\ln\left(
1+
\frac{(z_T-z_C-\tau)\mu_f}{n_j}
\right)
-
(z_T-z_C-\tau)
-
\frac{n_j}{2\lambda}
\ln^2\left(
1+
\frac{(z_T-z_C-\tau)\mu_f}{n_j}
\right)
\right].
\end{align*}}

Analogously, the total lateness delay is
{\setlength{\abovedisplayskip}{1pt}
\setlength{\belowdisplayskip}{1pt}
\begin{align*}
C_S^L(\tau)
=
\frac{n_j}{L}
\left[
\left(
\frac{n_j}{\mu_f}
+
z_T-z_C-\tau
\right)
\ln\left(
1+
\frac{(z_T-z_C-\tau)\mu_f}{n_j}
\right)
-
(z_T-z_C-\tau)
-
\frac{n_j}{2\lambda}
\ln^2\left(
1+
\frac{(z_T-z_C-\tau)\mu_f}{n_j}
\right)
\right].
\end{align*}}

Then, the total social cost is given by: $SC(\tau) = C_T(\tau) + C_F(\tau) + C_Q^O(\tau) + C_Q^E(\tau) + C_Q^L(\tau) + C_S^E(\tau) + C_S^L(\tau).$

\end{document}